\documentclass[draftclsnofoot]{IEEEtran}
\onecolumn
\usepackage{latexsym}
\usepackage{cite}
\usepackage{amssymb}
\usepackage{bm}
\usepackage{amsmath, amssymb}
\usepackage[dvips]{graphics}
\usepackage{graphicx}
\usepackage{psfrag}
\usepackage{subfloat}
\usepackage{subfig}
\usepackage{setspace}
\usepackage{algorithm}
\usepackage{algpseudocode}
\usepackage{theorem}
 \allowdisplaybreaks

\newtheorem{theorem}{Theorem}
\newtheorem{lemma}[theorem]{Lemma}

\begin{document}

\newcommand{\be}{\begin{equation}}
\newcommand{\ee}{\end{equation}}
\newcommand{\bea}{\begin{eqnarray}}
\newcommand{\eea}{\end{eqnarray}}
\newcommand{\beaa}{\begin{eqnarray*}}
\newcommand{\eeaa}{\end{eqnarray*}}
\newcommand{\p}[1]{\left(#1\right)}
\newcommand{\pp}[1]{\left[#1\right]}
\newcommand{\ppp}[1]{\left\{#1\right\}}

\title{Capacity Region of Finite State Multiple-Access Channel with Delayed  State Information at the Transmitters}
\author{Uria Basher,  Avihay Shirazi, and Haim Permuter}
\maketitle \vspace{-1.4cm}
\begin{abstract}%
A single-letter characterization is provided for
the capacity region of finite-state multiple access channels.
The channel state is a Markov process, the transmitters have access to
delayed state information, and channel
state information is available at the receiver. The delays
of the channel state information are assumed to be asymmetric
at the transmitters.
We apply the result to obtain the capacity region for a finite-state Gaussian MAC, and for a finite-state multiple-access fading channel.
We derive  power control strategies that maximize the capacity region for these channels.
\end{abstract}
\vspace{-0.0cm}
\section{INTRODUCTION}
Wireless communication is an example of channels where the channel
characteristics are time-varying. In a wireless setting, the user's motion and the changes in the environment, as well
as the interference, may lead to temporal changes in the
channel quality. Such channel variation models can include fast fading due to multi-path and
slow fading due to shadowing. In fast fading, the channel state is assumed
to be changing for every channel use, while in slow fading, the channel is assumed to
be constant for each finite block length.

In such communication problems, the channel state information (CSI) can be transmitted to the transmitters either explicitly, or through output CSI feedback. Frequently, the CSI feedback is not instantaneous;
the transmitters have only delayed information regarding the state of the channel.
The availability of the delayed CSI at the transmitters will possibly  increase the
capacity region. The increase in the capacity region due to
CSI depends on the CSI delays relative to the rate at which
the channel is time-varying.
When a channel is slowly time-varying and the delays are small, CSI may significantly
increase the capacity region. However, if the channel is changing rapidly relative
to the CSI delays, the transmitters can no longer adapt to
the channel variations. Hence, availability of delayed CSI  may not result in any
significant capacity region improvement.  Therefore, we are motivated to study the
effect of channel memory and delays on the multiple access channel (MAC)  capacity region.

Let us now present a brief literature review. We are modeling a time-varying channel as a
finite-state Markov channel (FSMC)\cite{Capacity_mutual_information},\cite{H.S._Wang_and_N.Moayeri}.
The FSMC is a channel with a finite number of states. During each symbol transmission, the channel's state is fixed.
The channel transition probability function is determined by the channel state.
The time variation in the
channel characteristics is modeled by the statistics of the underlying state process.

Capacity of memoryless channels, with different cases of state
information being available in a causal or  non causal manner at the 
transmitter  and at the receiver, has been studied by Shannon \cite{C.Shannon} and by Gelfand and Pinsker \cite{Gelfand_Pinsker}.
In \cite{Goldsmith_Varaiya}, Goldsmith and  Varaiya consider the  fading channels with perfect CSI at the transmitter and at the receiver. They proved that with instantaneous and perfect state information, the transmitter
can adapt the data rates for each channel state to  maximize the average transmission rate.
Viswanathan\cite{Viswanathan} loosened  this assumption of perfect instantaneous
CSI, and gave a single letter characterization of the capacity of Markov
channels with delayed CSI.
Caire and Shamai \cite{Shamai99}  consider the case  that  the
channel state is identically distributed (i.i.d.), and the CSI at the transmitter is a
deterministic function of the CSI at the receiver.
They showed that optimal coding is particularly simple.
Chen and Berger in \cite{Chen_Berger} found the capacity of an FSC with
inter-symbol interference (ISI), where current CSI is available at the transmitter and the receiver.
For a comprehensive survey on channel coding with state information see \cite{GuyKeshet_YossefSteinberg_NeriMerhav}.

The MAC with state has received much attention in recent years due to its importance in wireless communication systems.
On the one hand, complete knowledge of the CSI at the transmitters is an unrealistic assumption in wireless communications.
On the other hand, it is reasonable to assume that the receiver does possess full knowledge of the CSI. This practical
consideration has motivated the investigation of a MAC where
each transmitter is informed with its own CSI,  while the receiver is informed with the full CSI.%

Our work is also related to \cite{Como},\cite{Das_Narayan}, and \cite{Permuter_Weissman}. In \cite{Como}   Como and Y\"{u}ksel found the capacity region of FS-MAC,
where the channel state process is i.i.d.,  the transmitters have access to
partial (quantized) CSI, and complete CSI is available at the receiver.
In \cite{Das_Narayan} the capacity of general FS-MAC with varying degrees of causal CSI at the transmitters
is characterized in non-single-letter formulas.
In \cite{Permuter_Weissman} the capacity region of the FS-MAC with feedback that may be an
arbitrary time-invariant function of the channel output has been derived.
Recent related work also includes \cite{Lapidoth_Steinberg},  which studies the state-dependent MAC
with causal and strictly causal side information at the transmitters.

In this work, we consider the capacity region of a finite state Markov Multiple-access channel (FSM-MAC) with CSI at the decoder (receiver) and delayed CSI at the encoders (transmitters) with delays $d_1$ and $d_2$
as illustrated in
Fig. \ref{Multiple-Access Channel with Receiver CSI and Asymmetrical Delayed Feedback}.
\begin{figure}[h]
\begin{center}
\begin{psfrags}
    \psfragscanon
    \psfrag{A}{$Encoder1$}
    \psfrag{B}{$Encoder2$}
    \psfrag{C}{$\ p(y|x_1,x_2,s)$}
    \psfrag{D}{$\ Decoder$}
    \psfrag{E}{$ \ State \ Process$}
    \psfrag{F}{$M_1$}
    \psfrag{G}{$M_2$}
    \psfrag{H}{$\ \ \ X_{1}^{n}$}
    \psfrag{J}{$\ \ \ X_{2}^{n}$}
    \psfrag{K}{$Y^{n}$}
     \psfrag{L}{$\hat{M}_1,\hat{M}_2$}
     \psfrag{P}{$S_{i}$}
     \psfrag{T}{$S_{i}$}
     \psfrag{U}{$S_{i}$}
     \psfrag{Q}{$\ \ \ Channel$}
     \psfrag{Y}{$S_{i-d_2}$}
     \psfrag{Z}{$S_{i-d_2}$}
     \psfrag{W}{$Feedback \ Delay$}
     \psfrag{R}{$S_{i-d_1}$}
      \psfrag{X}{$S_{i-d_1}$}

\includegraphics[scale=0.7]{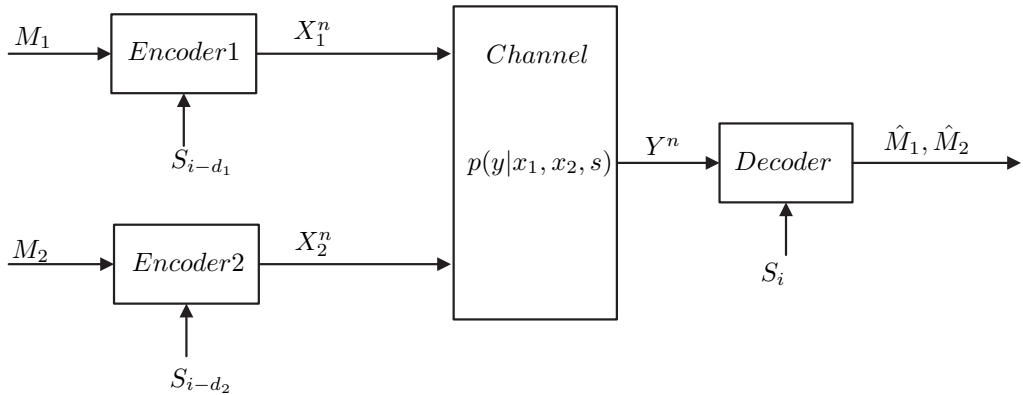}
\caption{ FSM-MAC with CSI at the decoder and delayed CSI at the encoders with delays $d_1$ and $d_2$-
The state process has memory and is assumed to be FSM.
The CSI is fed back to the encoders through a noiseless feedback channel.
CSI from the decoder is received at Encoder $1$ and Encoder $2$ after  time
delays of $d_1 $ and $d_2$ symbol durations, respectively. We are considering the above problem setting in the cases where $d_1>d_2$, $d_1=d_2$, and $d_2<d_1=\infty$.} \label{Multiple-Access Channel with Receiver CSI and Asymmetrical Delayed Feedback}
\psfragscanoff
\end{psfrags}
\end{center}
\end{figure}
The channel probability function at each time instant depends on the state of an underlying  finite-state Markov process.
The decoder, in addition to the channel output, also receives the channel state at each time instant (perfect CSI).
The channel state is fed back to the encoders through a noiseless feedback channel.
CSI from the decoder is received at Encoder $1$ and Encoder $2$ after  time delays of $d_1$ and $d_2$ symbol durations, respectively.
Each encoder, at each time instant, chooses the channel input based on the message to be transmitted and
the CSI that he possesses.
A formal description of the system model is presented in Section \ref{s_preliminary}.
The main result of this paper  is  a computable characterization of the capacity region for this
channel model.

The remainder of the paper is organized as follows: In Section \ref{s_preliminary}, we concretely describe
the communication model. In Section \ref{Main results}, we state our main results, which are the capacity regions for different cases of time delays.
Section \ref{CONVERSE} provides the upper bound on the capacity region of FSM-MAC with CSI at the decoder  and asymmetrical delayed  CSI at the encoders.
In Section \ref{ACHIEVABILITY }, we complete the proof of the capacity region, by providing the proof of the  achievability.
In Section \ref{ALTERNATIVE PROOF}, we provide alternative proof for capacity region. The alternative proof is based on a multi-letter expression for the capacity region of FS-MAC with time-invariant feedback \cite{Permuter_Weissman}.
In Section \ref{EXAMPLES}, we apply the general results of Section \ref{Main results} to obtain the capacity region for a finite-state Gaussian MAC, and for a finite-state multiple-access fading channel. We derive optimization problems on
the power allocation that maximize the capacity region for these channels. This power allocation would be the optimal power control policy for maximizing throughput in the presence of delayed CSI.
We conclude in Section \ref{summery} with a summary of this work.
\section{CHANNEL MODEL AND NOTATION \label{s_preliminary}}
\subsection{Channel Model}
In this paper, we consider the communication system of FSM-MAC with CSI at the decoder and delayed CSI at the encoders with delays $d_1$ and $d_2$, respectively,
as illustrated in Fig. \ref{Multiple-Access Channel with Receiver CSI and Asymmetrical Delayed Feedback}.
The MAC setting consists of two senders and one receiver. Each
sender $j\in\{1,2\}$ chooses an index $m_j$ uniformly from
the set $\left\{1, ..., 2^{nR_j}\right\}$ and independent of the other sender. The input to the channel from encoder $j\in\{1,2\}$ is denoted by $\{X_{j,1},X_{j,2},X_{j,3}, ...\}$, and the output of the channel is denoted
by $\{Y_1, Y_2, Y_3, ...\}$. We use the notation $V^n$ to denote the sequence $(V_1,...,V_n)$, therefore,
$X_{j}^n$, $Y^n$ denote the sequences  $\{X_{j,1},...,X_{j,n}\}$, $\{Y_1,..., Y_n\}$, respectively.
A finite-state Markov channel is, at each time instant, in one of a
finite number of states $\mathcal{S}=\{s_1,s_2,...,s_k\}$. In each state, the channel
is a DMC with inputs alphabet $\mathcal{X}_1,\mathcal{X}_2$ and output alphabet $\mathcal{Y}$. Let
the random variables $S_i$ ,$S_{i-d}$ denote the channel state at times $i$ and $i-d$, respectively.
Similarly, denote by $X_{1,i},X_{2,i}$, and $Y_i$ the inputs and the output of the channel at time $i$.
The channel transition probability function at time $i$ depends on the
state $S_i$, and the inputs $X_{1,i},X_{2,i}$  at time $i$, and is given by $P(y_i|x_{1,i},x_{2,i},s_i)$.
The channel output at any time $i$ is assumed to depend only on the channel inputs and
state at time $i$. Hence
\begin{eqnarray}
P(y_i|x_{1}^{i},x_{2}^{i},s_{1}^{i})=P(y_i|x_{1,i},x_{2,i},s_i).
\end{eqnarray}
The state process $\{S_i\}$ is assumed to be an irreducible, aperiodic,
finite-state homogeneous Markov chain and hence is ergodic.
The state process is independent of the channel inputs and output when
conditioned on the previous states, i.e.,
\begin{eqnarray}
P(s_i|s^{i-1},x_{1}^{i-1},x_{2}^{i-1},y^{i-1})=P(s_i|s_{i-1}).
\end{eqnarray}
Furthermore, we assume that the state process is independent of $M_1$ and $M_2$,
\begin{eqnarray}
P(s^n,m_1,m_2)=P(s^n)P(m_1)P(m_2)=\prod_{i=1}^{n}P(s_i|s_{i-1})P(m_1)P(m_2).
\end{eqnarray}

Now, let $K$ be the one step state transition probability matrix of the Markov process, and let $\pi$ be the steady state probability distribution of the Markov process. The $(S_{i},S_{i-d})$ joint distribution is stationary and is given by
\begin{eqnarray}
\pi_d(S_i=s_{l},S_{i-d}=s_j)=\pi(s_j)K^{d}(s_{l},s_{j}),
\end{eqnarray}
where $K^{d}(s_{l},s_j)$ is the $(l,j)$th element of the d-step transition
probability matrix $K^{d}$ of the Markov state process.
For simplicity, let us define  $S,\tilde{S}_1$ as the variables that have the same joint distribution as $(S_{i},S_{i-d_1})$, i.e.,
\begin{eqnarray}
P(S=s_l,\tilde{S}_1=s_j)=\pi_{d_1}(S_i=s_{l},S_{i-d_1}=s_j)=\pi(s_j)K^{d_1}(s_{l},s_{j}). \label{s1s2sdefinition}
\end{eqnarray}
Similarly, we define $S,\tilde{S}_2$ as the variables that have the same joint distribution as $(S_{i},S_{i-d_2})$.
\subsection{Code Description}
An $(n,2^{nR_1},2^{nR_2},d_1,d_2)$ code for FSM-MAC with CSI at the decoder and delayed CSI at the encoders with delay $d_1$ and $d_2$  consists of
\begin{enumerate}
\item Two sets of integers $\mathcal{M}_1=\{1,2,...,2^{nR_1}\}$ and $\mathcal{M}_2=\{1,2,...,2^{nR_2}\}$, called the { \it{message sets}}.
\item For each encoder, an encoding function $f_j$, $j\in \{1,2\}$,  maps the set of messages
 to channel input words of block length $n$. Each $f_j$ works through
a sequence of functions $f_{j,i}$ that depend only on the message $M_j$
and the  channel states up to time $i-d_j$. For encoder $1$ ($j=1$):
\begin{eqnarray}
X_{1,i}=\left\{ \begin{array}{rcl}
f_{1,i}(M_1),& 1\leq i\leq d_1\\
f_{1,i}(M_1,S^{i-d_1}),& d_1+1\leq i\leq n\\
\end{array}\right\}
\end{eqnarray}
Similarly for encoder $2$ ($j=2$):
\begin{eqnarray}
X_{2,i}=\left\{ \begin{array}{rcl}
f_{2,i}(M_2),& 1\leq i\leq d_2\\
f_{2,i}(M_2,S^{i-d_2}),& d_2+1\leq i\leq n\\
\end{array}\right\}
\end{eqnarray}

\item A decoding function $\psi$ that maps a received sequence of $n$
channel outputs and channel states to the messages set
\begin{eqnarray}
\psi : Y^{n}\times S^{n} \rightarrow \mathcal{M}_{1}\times \mathcal{M}_{2}.
\end{eqnarray}
We define the average probability of error for the  $(n,2^{nR_1},2^{nR_2},d_1,d_2)$ code as follows:
\begin{eqnarray}
P_{e}^{(n)}= \frac{1}{2^{n(R_1+R_2)}}\sum_{m_1,m_2}\sum_{s_{1}^{n}}P_{S^{n}}(s^{n})\Pr\{ \psi(y^{n},s^{n})\neq(m_1,m_2)|(m_1,m_2) \mbox{was sent}\}.
\end{eqnarray}
\end{enumerate}
We use standard definitions \cite{Cover_Thomas} of achievability and capacity region, namely, a pair rate $(R_1,R_2)$ is {\it{achievable}} for FSM-MAC with CSI at the decoder and delayed CSI at the encoders with delays $d_1$ and $d_2$, if there exists a sequence of $(n,2^{nR_1},2^{nR_2},d_1,d_2)$  codes with $P_{e}^{(n)}\rightarrow0$ as $n$ goes to infinity. The {\it{capacity region}} is the closure of the set of achievable $(R_1,R_2)$ rate pairs.
\section{MAIN RESULTS}\label{Main results}
Here we present the main results of this paper. Recall, that the joint distributions of ($S,\tilde{S}_1$), and ($S,\tilde{S}_2$) is given in (\ref{s1s2sdefinition}). Without loss of generality, let us assume that $d_1\geq d_2$.
\begin{theorem}\label{Capacity region-  MAC with delayed CSI feedback t1}{\it{(Capacity region of FSM-MAC with delayed CSI  $d_1 \geq d_2$)}}\\
The capacity region of FSM-MAC with CSI at the decoder and asymmetrical delayed  CSI  at the encoders with delays $d_1$ and $d_2$ as showen in Fig. \ref{Multiple-Access Channel with Receiver CSI and Asymmetrical Delayed Feedback}
is given by:
\begin{eqnarray}
\mathcal{R}=\bigcup_{P(u|\tilde{s}_1)P(x_{1}|\tilde{s}_1,u)P(x_{2}|\tilde{s}_1,\tilde{s}_2,u)} \left( \begin{array}{rcl}
R_{1}<I(X_{1};Y|X_{2},S,\tilde{S}_1,\tilde{S}_2,U),\\
R_{2}<I(X_{2};Y|X_{1},S,\tilde{S}_1,\tilde{S}_2,U),\\
R_{1}+R_{2}<I(X_{1},X_{2};Y|S,\tilde{S}_1,\tilde{S}_2,U),
\end{array}\right),
\end{eqnarray}
where $U$ is an auxiliary  random variable with cardinality  $|{\cal U}|\leq 3$.
\end{theorem}
The proof of Theorem \ref{Capacity region-  MAC with delayed CSI feedback t1} is presented in Sections \ref{CONVERSE}, and \ref{ACHIEVABILITY }. In Section \ref{CONVERSE} we prove the  upper bound of the capacity region, and  Section \ref{ACHIEVABILITY } is devoted to the proof of the  achievability. The proof of the  achievability is based on a multiplexing coding scheme, and successive decoding. In addition, we provide alternative proof of Theorem \ref{Capacity region-  MAC with delayed CSI feedback t1} in Section \ref{ALTERNATIVE PROOF}.
The proof for the cardinality bound of $U$ is presented in Appendix \ref{cardinality_proof}.

Now, directly from Theorem \ref{Capacity region-  MAC with delayed CSI feedback t1} we can derive the capacity region in the case of $d_1=d_2$. Since $d_1=d_2$ we have  $\tilde{S}_1=\tilde{S}_2$, hence we denote $\tilde{S}=\tilde{S}_1=\tilde{S}_2$. Using Theorem \ref{Capacity region-  MAC with delayed CSI feedback t1} we get,
\begin{theorem} \label{Capacity region-  MAC with delayed CSI feedback t2}{\it{(Capacity region of FSM-MAC with symmetrical delayed CSI  $d_1=d_2$)}}\\
The capacity region of FSM-MAC with CSI at the decoder and symmetrical delayed  CSI  at the encoders with delay $d$ is given by:
\begin{eqnarray}
\mathcal{R}=\bigcup_{P(u|\tilde{s})P(x_{1}|\tilde{s},u)P(x_{2}|\tilde{s},u)} \left( \begin{array}{rcl}
R_{1}<I(X_{1};Y|X_{2},S,\tilde{S},U),\\
R_{2}<I(X_{2};Y|X_{1},S,\tilde{S},U),\\
R_{1}+R_{2}<I(X_{1},X_{2};Y|S,\tilde{S},U),
\end{array}\right),
\end{eqnarray}
where $U$ is an auxiliary  random variable with cardinality  $|{\cal U}|\leq 3$.
\end{theorem}
Now we consider the case that encoder $1$ does not have state information at all, i.e., $d_1=\infty$.
\begin{theorem}\label{Capacity region-  MAC with delayed CSI feedback only to one encoder}{\it{( Capacity region of FSM-MAC with delayed CSI  only to one encoder)}}\\
The capacity region of FSM-MAC with CSI at the decoder and  delayed  CSI  only to one encoder
is given by :       %
\begin{eqnarray}
\mathcal{R}=\bigcup_{P(q)P(x_{1}|q)P(x_{2}|\tilde{s},q)} \left( \begin{array}{rcl}
R_{1}<I(X_{1};Y|X_{2},S,\tilde{S},Q),\\
R_{2}<I(X_{2};Y|X_{1},S,\tilde{S},Q),\\
R_{1}+R_{2}<I(X_{1},X_{2};Y|S,\tilde{S},Q),
\end{array}\right),
\end{eqnarray}
where $Q$ is an auxiliary  random variable with cardinality  $|{\cal Q}|\leq 3$.
\end{theorem}
The proof of Theorem \ref{Capacity region-  MAC with delayed CSI feedback only to one encoder} is quite similar to the proof of Theorem \ref{Capacity region-  MAC with delayed CSI feedback t1}; the details are presented in Appendix \ref{PROOF OF THEOREM3}.
\section{CONVERSE  \label{CONVERSE}}
In this section we provide the upper bound on the capacity region of MAC with receiver CSI and asymmetrical delayed  CSI feedback, i.e., we give the converse proof for Theorem \ref{Capacity region-  MAC with delayed CSI feedback t1}. Without loss of generality let us assume that $d_1\geq d_2$.
\begin{proof}
Given an  achievable rate $(R_{1},R_{2})$ we  need to show that there exists joint distribution of the form
$P(s,\tilde{s}_1,\tilde{s}_2)P(u|\tilde{s}_1)P(x_{1}|\tilde{s}_1,u)P(x_{2}|\tilde{s}_1,\tilde{s}_2,u)P(y|x_{1},x_{2},s)$ such that,
\begin{eqnarray}
R_{1}<I(X_{1};Y|X_{2},S,\tilde{S}_1,\tilde{S}_2,U),\nonumber \\
R_{2}<I(X_{2};Y|X_{1},S,\tilde{S}_1,\tilde{S}_2,U),\nonumber \\
R_{1}+R_{2}<I(X_{1},X_{2};Y|S,\tilde{S}_1,\tilde{S}_2,U),\nonumber
\end{eqnarray}
where $U$ is an auxiliary  random variable with cardinality  $|{\cal U}|\leq 3$. The proof for the cardinality bound is presented in Appendix \ref{cardinality_proof}. Since $(R_1,R_2)$ is an achievable pair-rate, there exists a code $(n,2^{nR_1},2^{nR_2},d_1,d_2)$ with a probability of error  $P_{e}^{(n)}$ arbitrarily small. By Fano's inequality,
\begin{eqnarray}
H(M_1,M_2|Y^n,S^n)\leq n(R_1+R_2)P_{e}^{(n)}+H(P_{e}^{(n)})\triangleq n\varepsilon_n,
\end{eqnarray}
and it is clear that $\varepsilon_n\rightarrow0$ as $P_{e}^{(n)}\rightarrow\infty$. Then we have
\begin{eqnarray}
H(M_1|Y^n,S^n)\leq H(M_1,M_2|Y^n,S^n)\leq \varepsilon_n, \\
H(M_2|Y^n,S^n)\leq H(M_1,M_2|Y^n,S^n)\leq \varepsilon_n.
\end{eqnarray}

We can now bound the rate $R_1$ as
  \begin{eqnarray}
  nR_{1}&=& H(M_1)  \nonumber\\
        &=& H(M_1)+H(M_1|Y^{n},S^{n})-H(M_1|Y^{n},S^{n})                     \nonumber\\
        &\stackrel{(a)}\leq& I(M_1;Y^n,S^{n})+n \varepsilon_n                    \nonumber\\
        &\stackrel{(b)}=& I(M_1;Y^n|S^{n})+I(M_1;S^{n})+n \varepsilon_n     \nonumber\\
        &\stackrel{(c)}=& I(M_1;Y^n|S^{n})+n \varepsilon_n                       \nonumber\\
        &\stackrel{(d)}=&  I(X_{1}^{n};Y^n|S^{n})+n \varepsilon_n             \nonumber\\
        &=& H(X_{1}^{n}|S^{n})-H(X_{1}^{n}|Y^n,S^{n})+n \varepsilon_n        \nonumber\\
        &\stackrel{(e)}=& H(X_{1}^{n}|X_{2}^{n},S^{n})-H(X_{1}^{n}|Y^n,S^{n})+n \varepsilon_n       \nonumber\\
        &\stackrel{(f)}\leq& H(X_{1}^{n}|X_{2}^{n},S^{n})-H(X_{1}^{n}|Y^n,X_{2}^{n},S^{n})+n \varepsilon_n  \nonumber\\
        &=& I(X_{1}^{n};Y^n|X_{2}^{n},S^{n})+n \varepsilon_n        \nonumber\\
        &=& H(Y^{n}|X_2^{n},S^{n})-H(Y^{n}|X_1^{n},X_2^{n},S^{n})+n \varepsilon_n  \nonumber\\
        &=& \sum_{i=1}^{n} H(Y_{i}|Y^{i-1},X_{2}^{n},S^{n})- H(Y_{i}|Y^{i-1},X_{1}^{n},X_{2}^{n},S^{n})
        +n \varepsilon_n \nonumber\\
        &\stackrel{(g)}\leq& \sum_{i=1}^{n} H(Y_{i}|X_{2,i},S_{i},S_{i-d_2},S_{i-d_1},S^{i-d_1-1})- H(Y_{i}|Y^{i-1},X_{1}^{n},X_{2}^{n},S^{n})
        +n \varepsilon_n \nonumber\\
        &\stackrel{(h)}=&  \sum_{i=1}^{n} H(Y_{i}|X_{2,i},S_{i},S_{i-d_2},S_{i-d_1},S^{i-d_1-1})- H(Y_{i}|X_{1,i},X_{2,i},S_{i},S_{i-d_2},S_{i-d_1},S_{1}^{i-d_1-1})
        +n \varepsilon_n \nonumber\\
        &=& \sum_{i=1}^{n}I(Y_{i};X_{1,i}|X_{2,i},S_{i},S_{i-d_2},S_{i-d_1},S_{1}^{i-d_1-1})+n \varepsilon_n , \nonumber
  \end{eqnarray}
  where\\
  (a) follows from Fano's inequality.\\
  (b) follows from chain rule.\\
  (c) follows from the fact that  $M_1$ and $S^{n}$ are independent.\\
  (d) follows from the fact that $X_{1}^{n}$ is a deterministic function of $(M_1,S^n)$ and the Markov chain $(M_1,S^n) - (X_{1}^n,S^n) - Y^n$.\\
  (e) follows from the fact that $X_{1}^{n}$ and $M_2$ are independent, and the fact that $X_{2}^{n}$ is a deterministic  function of $(M_2,S^n)$. Therefore,  $X_{1}^{n}$ and $X_{2}^{n}$ are independent given $S^{n}$ .\\
  (f) and (g) follow from the fact that conditioning reduces entropy.\\
  (h) follows from the fact  that the channel output at time $i$ depends only on the state $S_i$ and the the inputs $X_{1,i}$ and $X_{2,i}$.\\
  Hence, we have
  \begin{eqnarray}
  R_{1}\leq \frac{1}{n}\sum_{i=1}^{n}I(Y_{i};X_{1,i}|X_{2,i},S_{i},S_{i-d_2},S_{i-d_1},S_{1}^{i-d_1-1})+ \varepsilon_n.  \label{r1}
  \end{eqnarray}
  Similarly, we have
  \begin{eqnarray}
  R_{2}\leq \frac{1}{n}\sum_{i=1}^{n}I(Y_{i};X_{2,i}|X_{1,i},S_{i},S_{i-d_2},S_{i-d_1},S_{1}^{i-d_1-1})+ \varepsilon_n.  \label{r2}
  \end{eqnarray}
  To bound the sum of the rates, consider
  \begin{eqnarray}
   n(R_{1}+R_{2})&=& H(M_1,M_2)  \nonumber\\
                 &=& H(M_1,M_2)+H(M_1,M_2|Y^{n},S^{n})-H(M_1,M_2|Y^{n},S^{n})  \nonumber\\
                 &\stackrel{(a)}\leq& I(M_1,M_2;Y^n,S^{n})+n \varepsilon_n                    \nonumber\\
                 &\stackrel{(b)}=& I(M_1,M_2;Y^n|S^{n})+I(M_1,M_2;S^{n})+n \varepsilon_n     \nonumber\\
                 &\stackrel{(c)}=& I(M_1,M_2;Y^n|S^{n})+n \varepsilon_n                       \nonumber\\
                 &\stackrel{(d)}=& I(X_{1}^{n},X_{2}^{n};Y^n|S^{n})+n \varepsilon_n           \nonumber\\
                 &=&H(Y^{n}|S^{n})-H(Y^{n}|X_{1}^{n},X_{2}^{n},S^{n})   \nonumber\\
                 &=& \sum_{i=1}^{n} H(Y_{i}|Y^{i-1},S^{n})-H(Y_{i}|Y^{i-1},X_{1}^{n},X_{2}^{n},S^{n})
                        +n \varepsilon_n \nonumber\\
                 &\stackrel{(e)}=&\sum_{i=1}^{n} H(Y_{i}|Y^{i-1},S^{n})-
                 H(Y_{i}|X_{1,i},X_{2,i},S_{i},S_{i-d_2},S_{i-d_1},S^{i-d_1-1})
                        +n \varepsilon_n \nonumber\\
                 &\stackrel{(f)}\leq&  \sum_{i=1}^{n} H(Y_{i}|S_{i},S_{i-d_2},S_{i-d_1},S^{i-d_1-1})-
                H(Y_{i}|X_{1,i},X_{2,i},S_{i},S_{i-d_2},S_{i-d_1},S^{i-d_1-1})
                        +n \varepsilon_n \nonumber\\
                 &=& \sum_{i=1}^{n}I(Y_{i};X_{1,i},X_{2,i}|S_{i},S_{i-d_2},S_{i-d_1},S^{i-d_1-1})+n \varepsilon_n , \nonumber          \end{eqnarray}
 where\\
  (a) follows from Fano's inequality.\\
  (b) follows from chain rule.\\
  (c) follows from the fact that  $M_1,M_2$, and $S^{n}$ are independent.\\
  (d) follows from the fact that $X_{1}^{n},X_{2}^{n}$ is a deterministic function of $(M_1,M_2,S^n)$ and the Markov chain $(M_1,M_2,S^n) - (X_{1}^n,X_{2}^{n},S^n) - Y^n$.\\
  (e) follows from the fact that the channel output at time $i$ depends only on the state $S_{i}$,  and the inputs $X_{1,i}$, and $X_{2,i}$.\\
  (f) follows from the fact that conditioning reduces entropy.\\
  Hence, we have
  \begin{eqnarray}
   R_{1}+R_{2}\leq \frac{1}{n}\sum_{i=1}^{n}I(Y_{i};X_{1,i},X_{2,i}|S_{i},S_{i-d_2},S_{i-d_1},S^{i-d_1-1})+ \varepsilon_n. \label{r1+r2}
  \end{eqnarray}
  The expressions in (\ref{r1}), (\ref{r2}), and (\ref{r1+r2}) are the average of the mutual informations calculated at
  the empirical distribution in column $i$ of the codebook. We can rewrite these equations with the new variable $Q$,
  where $Q=i\in \{1,2,...,n\}$ with probability $\frac{1}{n}$. The equations become
  \begin{eqnarray}
  R_{1}&\leq& \frac{1}{n}\sum_{i=1}^{n}I(Y_{i};X_{1,i}|X_{2,i},S_{i},S_{i-d_2},S_{i-d_1},S^{i-d_1-1})+ \varepsilon_n \nonumber\\
        &=&  \frac{1}{n}\sum_{i=1}^{n}I(Y_{Q};X_{1,Q}|X_{2,Q},S_{Q},S_{Q-d_2},S_{Q-d_1},S^{Q-d_1-1},Q=i)+ \varepsilon_n \nonumber\\
        &=& I(Y_{Q};X_{1,Q}|X_{2,Q},S_{Q},S_{Q-d_2},S_{Q-d_1},S^{Q-d_1-1},Q)+ \varepsilon_n \label{Uconverse}
  \end{eqnarray}
  Now let us denote $ X_{1}\triangleq X_{1,Q}, X_{2} \triangleq X_{2,Q}, Y \triangleq Y_{Q}, S\triangleq S_{Q}, \tilde{S}_1\triangleq S_{Q-d_1},\tilde{S}_2\triangleq S_{Q-d_2}$, and  $U\triangleq (S^{Q-d_1-1},Q)$.
  \\We have,
   \begin{eqnarray}
   R_{1}&\leq& I(X_{1};Y|X_{2},S,\tilde{S}_1,\tilde{S}_2,U)+ \varepsilon_n , \nonumber\\
   R_{2} &\leq& I(X_{2};Y|X_{1},S,\tilde{S}_1,\tilde{S}_2,U) + \varepsilon_n, \nonumber\\
   R_{1}+R_{2}&\leq& I(X_{1},X_{2};Y|S,\tilde{S}_1,\tilde{S}_2,U) + \varepsilon_n. \nonumber
   \end{eqnarray}
  To complete the converse proof we need to show the following Markov relations hold:
  \begin{enumerate}
  \item $P(u|s,\tilde{s}_1,\tilde{s}_2)=P(u|\tilde{s}_1)$ .
  \item $P(x_1|s,\tilde{s}_1,\tilde{s}_2,u)=P(x_1|\tilde{s}_1,u)$.
  \item $P(x_2|x_1,s,\tilde{s}_1,\tilde{s}_2,u)=P(x_2|\tilde{s}_1,\tilde{s}_2,u)$.
  \item $P(y|x_1,x_2,s,\tilde{s}_1,\tilde{s}_2,u)=P(y|x_1,x_2,s)$.
  \end{enumerate}
  We prove the above using the following claims:
  \begin{enumerate}
  \item  follows from the fact that $S^{i-d_1-1}-S_{i-d_1}-S_{i-d_2}-S_{i}$ and so is ($S_{1}^{Q-d_1-1},Q) - S_{Q-d_1} -S_{Q-d_2} - S_{Q}$.
  \item  follows from the fact that $X_{1,i}=f_{1,i}(M_1,S^{i-d_1})$ and that $M_1$ and $S^{n}$ are independent. Hence
  \begin{eqnarray}
  P(x_{1,q}|s_{q},s_{q-d_1},s_{q-d_2},s_{1}^{q-d_1-1},q=i)=P(x_{1,q}|s_{q-d_1},s_{1}^{q-d_1-1},q=i).\nonumber
  \end{eqnarray}
  Since this is true for all $i$,
  \begin{eqnarray}
  P(x_{1,q}|s_{q},s_{q-d_1},s_{q-d_2},s_{1}^{q-d_1-1},q)=P(x_{1,q}|s_{q-d_1},s_{1}^{q-d_1-1},q). \nonumber
  \end{eqnarray}
   Therefore we have,
    \begin{eqnarray}
    P(x_1|s,\tilde{s}_1,\tilde{s}_2,u)=P(x_1|\tilde{s}_1,u).\nonumber
    \end{eqnarray}
  \item  We assume that $d_1\geq d_2$, since $M_2$ and $(M_1,S^{n})$ are independent, and the state process is Markov chain, we have
  \begin{eqnarray}
   P(m_{2},s^{i-d_2}|s_{i},s_{i-d_1},s_{i-d_2},s^{i-d_1},m_{1})&=&P(m_2,s^{i-d_2}|s_{i-d_1},s_{i-d_2},s^{i-d_1}).\nonumber
  \end{eqnarray}
  Therefore, we have the Markov chain $(M_{2},S^{i-d_2})-(S_{i-d_1},S_{i-d_2},S^{i-d_1})-(M_1,S_i,S^{i-d_1})$. Since $X_{1,i}=f_{1,i}(M_1,S^{i-d_1})$ and $X_{2,i}=f_{2,i}(M_2,S^{i-d_2})$ where $f_{1,i},f_{2,i}$ are deterministic functions, we obtain the following Markov chain,
   \begin{eqnarray}
  X_{2,i}-(M_{2},S^{i-d_2})-(S_{i-d_1},S_{i-d_2},S^{i-d_1})-(M_1,S_i,S^{i-d_1})-X_{1,i}.
  \end{eqnarray}
  Which implies,
  \begin{eqnarray}
  P(x_{2,i}|x_{1,i},s_{i},s_{i-d_1},s_{i-d_2},s^{i-d_1-1})&=&P(x_{2,i}|s_{i-d_1},s_{i-d_2},s^{i-d_1-1}). \nonumber
  \end{eqnarray}
 Since this is true for all $i$,
  \begin{eqnarray}
  P(x_{2,q}|x_{1,q},s_{q},s_{q-d_1},s_{q-d_2},s_{1}^{q-d_1-1},q)&=&P(x_{2,q}|s_{q-d_1},s_{q-d_2},s_{1}^{q-d_1-1},q).\nonumber
  \end{eqnarray}
  Therefore we have $P(x_2|x_1,s,\tilde{s}_1,\tilde{s}_2,u)=P(x_2|\tilde{s}_1,\tilde{s}_2,u)$.
  \item  follows from the fact that the channel output at any time $i$ is assumed to depend only on the channel inputs and
state at time $i$.
  \end{enumerate}
  Hence, taking the limit as $n \rightarrow\infty$, $P_{e}^{(n)}\rightarrow 0$, we have the following converse:
  \begin{eqnarray}
   R_{1}&\leq& I(X_{1};Y|X_{2},S,\tilde{S}_1,\tilde{S}_2,U), \nonumber\\
   R_{2} &\leq& I(X_{2};Y|X_{1},S,\tilde{S}_1,\tilde{S}_2,U), \nonumber\\
   R_{1}+R_{2}&\leq& I(X_{1},X_{2};Y|S,\tilde{S}_1,\tilde{S}_2,U), \nonumber
  \end{eqnarray}
  for some choice of joint distribution $P(s,\tilde{s}_1,\tilde{s}_2)P(u|\tilde{s}_1)P(x_{1}|\tilde{s}_1,u)P(x_{2}|\tilde{s}_1,\tilde{s}_2,u)P(y|x_{1},x_{2},s)$
  and for some choice of auxiliary random variable $U$ defined on  $|{\cal U}|\leq 3$. This completes the proof of the converse.
  \end{proof}
  \section{PROOF OF THE ACHIEVABILITY OF THEOREM \ref{Capacity region-  MAC with delayed CSI feedback t1}  \label{ACHIEVABILITY }}
  In the previous section we proved the converse of the capacity region of Theorem \ref{Capacity region-  MAC with delayed CSI feedback t1}. In this section we prove the achievability part.
  The main idea of the proof is using multiplexing coding, i.e., multiplexing  the input of the channel at each encoder (the multiplexer is controlled by the delayed CSI),  then, using the CSI known at the decoder, demultiplexing  the output at the decoder.
  \begin{proof}
  To prove the achievability of the capacity region, we need to show that for a fix $P(x_{1}|\tilde{s}_1)P(x_{2}|\tilde{s}_1,\tilde{s}_2)$ and ($R_1,R_2$) that satisfy
   \begin{eqnarray}
   R_{1}&\leq& I(X_{1};Y|X_{2},S,\tilde{S}_1,\tilde{S}_2), \nonumber\\
   R_{2} &\leq& I(X_{2};Y|X_{1},S,\tilde{S}_1,\tilde{S}_2), \nonumber\\
   R_{1}+R_{2}&\leq& I(X_{1},X_{2};Y|S,\tilde{S}_1,\tilde{S}_2), \nonumber
  \end{eqnarray}
  there exists a sequence of $(n,2^{nR_{1}},2^{nR_{2}},d_1,d_2)$ codes where $P_{e}^{(n)}\rightarrow0$ as $n\rightarrow\infty$.
  Without loss of generality, we assume that the finite-state space $\mathcal{S}=\left\{1,2,...,k\right\}$, and that the steady state probability   $\pi(l)>0$ for all $l\in \mathcal{S}$.

 \emph{Encoder 1}: construct $k$  codebooks $\mathcal{C}^{\tilde{s}_1}_{1}$  (where the subscript is for Encoder $1$) for all $\tilde{S}_1\in \mathcal{S}$, when in each codebook $\mathcal{C}^{\tilde{s}_1}_{1}$ there are $2^{n_{1}(\tilde{s}_1)R_{1}(\tilde{s}_1)}$ codewords, where $n_{1}(\tilde{s}_1)=(P(\tilde{S}_1=\tilde{s}_1)-\epsilon')n$, for $\epsilon'>0$. Every codeword $\mathcal{C}^{\tilde{s}_1}_{1}(i)$ when $i\in \{1,2,..., 2^{n_{1}(\tilde{s}_1)R_{1}(\tilde{s}_1)}\}$ has a length of $n_{1}(\tilde{s}_1)$ symbols.
 Each codeword  from the $\mathcal{C}^{\tilde{s}_1}_{1}$  codebook is built $X^{\tilde{s}_1}_{1}\thicksim$ i.i.d. $ P(x^{\tilde{s}_1}_{1}|\widetilde{S}_1=\tilde{s}_1)$ (where the subscript is for Encoder $1$).
 A message $M_1$ is chosen according to a uniform distribution $\Pr (M_1=m_1)=2^{-nR_1}$, $m_1\in\left\{1,2,...,2^{nR_1}\right\}$.
 Every message $m_1$ is mapped into $k$ sub messages  $\mathcal{V}_{1}(m_1)=\left\{V^{1}_{1}(m_1),V^{2}_{1}(m_1),...,V^{k}_{1}(m_1)\right\}$ (one message from each codecook). Hence, every message  $m_{1}$ is specified by a $k$ dimensional vector.
 For a fix block length $n$, let $N_{\tilde{s}_1}$ be the number of times during the $n$ symbols
for which the feedback information at  encoder $1$ regarding the
channel state is $\tilde{S}_1=\tilde{s}_1$.
 Every time that the delayed CSI is $\tilde{S}_1=\tilde{s}_1$, encoder $1$ sends the next symbol from the $\mathcal{C}^{\tilde{s}_1}_{1}$ codebook.
 Since $N_{\tilde{s}_1}$ is not necessarily
  equivalent to $n_1(\tilde{s}_1)$, an error is declared if $N_{\tilde{s}_1}<n_1(\tilde{s}_1)$, and the code is zero-filled
if $N_{\tilde{s}_1}>n_1(\tilde{s}_1)$.
 Therefore, we can send a total of
 $2^{nR_{1}}=2^{\sum_{\tilde{s}_1\in \mathcal{S} }n_{1}(\tilde{s}_1)R_1(\tilde{s}_1)}$ messages.
%

\emph{Encoder 2}: construct  $k\times k$  codebooks $\mathcal{C}^{\tilde{s}_1,\tilde{s}_2}_{2}$  (where the subscript is for Encoder $2$) for all $(\tilde{s}_1,\tilde{s}_2)\in\{ \mathcal{S}\times\mathcal{S}\}$, when in each codebook $\mathcal{C}^{\tilde{s}_1,\tilde{s}_2}_{2}$ there are $2^{n_2(\tilde{s}_1,\tilde{s}_2)R_2(\tilde{s}_1,\tilde{s}_2)}$ codewords, where $n_{2}(\tilde{s}_1,\tilde{s}_2)=(P(\tilde{S}_1,\tilde{S}_2=\tilde{s}_1,\tilde{s}_2)-\epsilon')n$, for $\epsilon'>0$. Every codeword $\mathcal{C}^{\tilde{s}_1,\tilde{s}_2}_{2}(i)$ when $i\in \{1,2,..., 2^{n_{2}(\tilde{s}_1,\tilde{s}_2)R_{2}(\tilde{s}_1,\tilde{s}_2)}\}$ has a length of $n_2(\tilde{s}_1,\tilde{s}_2)$ symbols.
 Each codeword from the $\mathcal{C}^{\tilde{s}_1,\tilde{s}_2}_{2}$ codebook is built $X^{\tilde{s}_1,\tilde{s}_2}_{2}\thicksim$ i.i.d. $ P(x^{\tilde{s}_1,\tilde{s}_2}_{2}|(\widetilde{S}_1,\widetilde{S}_2)=(\tilde{s}_1,\tilde{s}_2))$ (where the subscript is for Encoder $2$).
 A message $M_2$ is chosen according to a uniform distribution $\Pr (M_2=m_2)=2^{-nR_2}$, $m_2\in\left\{1,2,...,2^{nR_2}\right\}$.
 Every message $m_2$ is mapped into $k\times k$ sub messages
 $\mathcal{V}_{2}(m_2)=\left\{V^{1,1}_{2}(m_1),V^{1,2}_{2}(m_2),...,V^{k,k}_{2}(m_2)\right\}$ (one message from each codecook). Hence, every message  $m_{2}$ is specified by a $k\times k$ dimensional vector.
For a fix block length $n$, let $N_{\tilde{s}_1,\tilde{s}_2}$ be the number of times during the $n$ symbols
for which the feedback information at  encoder $2$ regarding the
channel state is $(\tilde{S}_1,\tilde{S}_2)=(\tilde{s}_1,\tilde{s}_2)$.
 Every time that the delayed CSI is $(\tilde{S}_1,\tilde{S}_2)=(\tilde{s}_1,\tilde{s}_2)$, encoder $2$ sends the next symbol from the $\mathcal{C}^{\tilde{s}_1,\tilde{s}_2}_{2}$ codebook.
 Since $N_{\tilde{s}_1,\tilde{s}_2}$ is not necessarily
  equivalent to $n_2(\tilde{s}_1,\tilde{s}_2)$, an error is declared if $N_{\tilde{s}_1,\tilde{s}_2}<n_2(\tilde{s}_1,\tilde{s}_2)$, and the code is zero-filled
if $N_{\tilde{s}_1,\tilde{s}_2}>n_2(\tilde{s}_1,\tilde{s}_2)$.
 Therefore, we can send a total of $2^{nR_{2}}=2^{\sum_{\tilde{s}_1,\tilde{s}_2\in\mathcal{S}\times\mathcal{S}}n_{2}(\tilde{s}_1,\tilde{s}_2)R_2(\tilde{s}_1,\tilde{s}_2)}$ messages.

\emph{Decoding }: We use successive decoding; in this method, instead of decoding the two messages simultaneously, the decoder first decodes one of the messages by itself, where the other user's message is considered as noise.
 After decoding the first user's message, the decoder turns to decode the second message. When decoding  the second message, the decoder uses the information about the first message as side information.
 This decoding rule aims to achieve the two corner points of the rate region, i.e., ($R_1=I(X_1;Y|X_2,S,\tilde{S}_1,\tilde{S}_2)-\epsilon, R_2=I(X_2;Y|S,\tilde{S}_1,\tilde{S}_2)-\epsilon$), and ($R_1=I(X_1;Y|S,\tilde{S}_1,\tilde{S}_2)-\epsilon, R_2=I(X_2;Y|X_1,S,\tilde{S}_1,\tilde{S}_2)-\epsilon$).
 The rate region is illustrated in Fig. \ref{Sum rate}.
\begin{figure}[h!]{
  \begin{center}
\psfrag{R}[][][0.8]{ $R_2$ }  \psfrag{P}[][][0.8]{$R_1$}
\psfrag{A}[][][0.8]{ $I(X_2;Y|X_1,S,\tilde{S}_1,\tilde{S}_2)\rightarrow$ }
\psfrag{B}[][][0.8]{ $\ \ \ \ I(X_2;Y|S,\tilde{S}_1,\tilde{S}_2)\rightarrow$ }
\psfrag{C}[][][0.8]{ $I(X_1;Y|S,\tilde{S}_1,\tilde{S}_2)$ }
\psfrag{D}[][][0.8]{ $\ \ \ \ \ \ \ \ \ \ \ I(X_1;Y|X_2,S,\tilde{S}_1,\tilde{S}_2)$ }
\psfrag{E}[][][0.8]{ $\uparrow$ }
\psfrag{F}[][][0.8]{ $\uparrow$ }
\includegraphics[width=9cm]{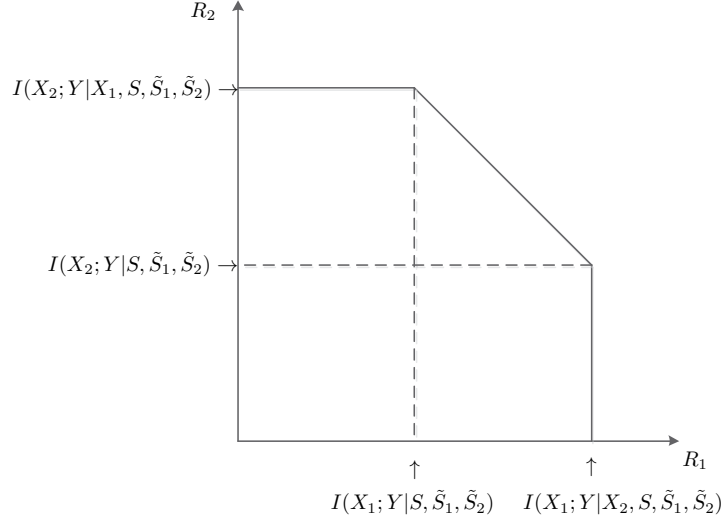}\ \ \ \ \ \ \ \
  \end{center}
\caption{The rate region}
 \label{Sum rate}
}\end{figure}

  To achieve the first point, let us analyze the case where the decoder first decodes $X_{2}^n$.
 The information  $\tilde{S}_1,\tilde{S}_2$  used to multiplex the codewords at the encoder is also available at the decoder. Hence, upon receiving a block of channel outputs and states $(Y^n,S^n)$, the decoder first demultiplexes  it into outputs corresponding to the component codebooks of encoder $2$. Then, the decoder separately decodes each component codeword $V^{\tilde{s}_1,\tilde{s}_2}_{2}$ where $(\tilde{s}_1,\tilde{s}_2)\in \mathcal{S}\times\mathcal{S}$. For each codebook $\mathcal{C}^{\tilde{s}_1,\tilde{s}_2}_{2}$, the decoder has $(Y^{n_2(\tilde{s}_1,\tilde{s}_2)},S^{n_2(\tilde{s}_1,\tilde{s}_2)})$ and searches $(X_{2}^{n_{2}(\tilde{s}_1,\tilde{s}_2)})$ such that
  $(X_{2}^{n_{2}(\tilde{s}_1,\tilde{s}_2)},Y^{n_2(\tilde{s}_1,\tilde{s}_2)},S^{n_2(\tilde{s}_1,\tilde{s}_2)})$ are strongly jointly typical sequences \cite{Cover_Thomas}, i.e.,
  $(X_{2}^{n_{2}(\tilde{s}_1,\tilde{s}_2)},Y^{n_2(\tilde{s}_1,\tilde{s}_2)},S^{n_2(\tilde{s}_1,\tilde{s}_2)}\in A_{\varepsilon}^{*(n_2(\tilde{s}_1,\tilde{s}_2))}(X_{2},Y,S))$ given $(\tilde{S}_1,\tilde{S}_2)=(\tilde{s}_1,\tilde{s}_2)$.
    The decoder declares that $\hat{m_2}$ is sent if it is a unique message such that $(X_{2}^{n_{2}(\tilde{s}_1,\tilde{s}_2)}(\hat{m_2}),Y^{n_2(\tilde{s}_1,\tilde{s}_2)},S^{n_2(\tilde{s}_1,\tilde{s}_2)}\in A_{\varepsilon}^{*(n_2(\tilde{s}_1,\tilde{s}_2))}(X_{2},Y,S))$ given $(\tilde{S}_1,\tilde{S}_2)=(\tilde{s}_1,\tilde{s}_2)$ for all $\tilde{s}_1,\tilde{s}_2\in \mathcal{S}\times\mathcal{S}$, otherwise it declares an error.
    If such $\hat{m_2}$ is found, the decoder has
    $X_{2}^{n}(\hat{m_2})$, but now the decoder is using the information  $\tilde{S}_1$ to demultiplex $(Y^n,S^n)$ into outputs corresponding to the component codebooks of encoder $1$ (which have $k$ codebooks). The decoder declares that $\hat{m_1}$ is sent if it is a unique message such that $(X_{1}^{n_{1}(\tilde{s}_1)}(\hat{m_1}),X_{2}^{n_{1}(\tilde{s}_1)}(\hat{m_2}),Y^{n_1(\tilde{s}_1)},S^{n_1(\tilde{s}_1)}\in A_{\varepsilon}^{*(n_1(\tilde{s}_1))}(X_1,X_{2},Y,S))$ given $\tilde{S}_1=\tilde{s}_1$ for all $\tilde{s}_1\in \mathcal{S}$, otherwise it declares error.

    \emph{Analysis of the probability of error}:
     First, we analyze the probability of error for the component codeword $V^{\tilde{s}_1,\tilde{s}_2}_{2}$ at encoder $2$, i.e., $\Pr(N_{\tilde{s}_1,\tilde{s}_2}<n_2(\tilde{s}_1,\tilde{s}_2))$.
    Since that the state process is stationary and ergodic
$\lim_{n\rightarrow\infty}\frac{N(\tilde{s}_1,\tilde{s}_2)}{n}=P(\tilde{s}_1,\tilde{s}_2)$ in probability. Therefore, $\Pr(N_{\tilde{s}_1,\tilde{s}_2}<n_2(\tilde{s}_1,\tilde{s}_2))\rightarrow 0$ as $n\rightarrow\infty$.
%
     Now, we analyze the probability to decode incorrectly the  component codeword $V^{\tilde{s}_1,\tilde{s}_2}_{2}$ that was sent from the $\mathcal{C}^{\tilde{s}_1,\tilde{s}_2}_{2}$ codebook of encoder $2$.
     Without loss of generality, we can assume that the first codeword was sent from the $\mathcal{C}^{\tilde{s}_1,\tilde{s}_2}_{2}$ codebook of encoder $2$, which we denote by   $\mathcal{C}^{\tilde{s}_1,\tilde{s}_2}_{2}(1)$. Since $S^{n_2(\tilde{s}_1,\tilde{s}_2)}$ is ergodic and by using the L.L.N. as $n_2(\tilde{s}_1,\tilde{s}_2) \rightarrow \infty$ we have $\Pr\left\{ S^{n_2(\tilde{s}_1,\tilde{s}_2)} \in A_{\varepsilon}^{*(n_2(\tilde{s}_1,\tilde{s}_2))}(S)\right\}  \rightarrow 1$. By  the construction of the codebook $\mathcal{C}^{\tilde{s}_1,\tilde{s}_2}_{2}(1)$,  $X_2$ and $S$ are independent given $(\tilde{S}_1,\tilde{S}_2)=(\tilde{s}_1,\tilde{s}_2)$. Hence $X_{2}^{n_2(\tilde{s}_1,\tilde{s}_2)}(1)$ and $S^{n_2(\tilde{s}_1,\tilde{s}_2)}$ are strongly jointly typical sequences with probability $1$. Finally from the codebooks construction and the channel transition probability we have that,
    \begin{eqnarray}
    p(y_i|x_{2}^i,s^i,\tilde{s}_1,\tilde{s}_2)&=&
    \sum_{x_{1,i}\in X_{1,i}}p(x_{1,i}|x_{2}^i,s^i,\tilde{s}_1,\tilde{s}_2)p(y_i|x_{1,i},x_{2}^i,s^i,\tilde{s}_1,\tilde{s}_2) \nonumber\\
    &=&\sum_{x_{1,i}\in X_{1,i}}p(x_{1,i}|\tilde{s}_1,\tilde{s}_2)p(y_i|x_{1,i},x_{2,i},s_i,\tilde{s}_1,\tilde{s}_2) \nonumber\\
    &=&p(y_i|x_{2,i},s_i,\tilde{s}_1,\tilde{s}_2).
    \end{eqnarray}
    Now using the fact that $p(y_i|x_{2}^i,s^i,\tilde{s}_1,\tilde{s}_2)=p(y_i|x_{2,i},s_i,\tilde{s}_1,\tilde{s}_2)$, and the L.L.N. we have  $\Pr\left\{ X_{2}^{n_2(\tilde{s}_1,\tilde{s}_2)}(1),S^{n_2(\tilde{s}_1,\tilde{s}_2)},Y^{n_2(1)}\in A_{\epsilon}^{*(n_2(\tilde{s}_1,\tilde{s}_2))}(X_{2},Y,S)|(\tilde{S}_1,\tilde{S}_2)=(\tilde{s}_1,\tilde{s}_2) \right\}\rightarrow 1$ as $n_2(\tilde{s}_1,\tilde{s}_2)\rightarrow\infty$. A decoding error occurs only if
    \begin{eqnarray}
   E_1&=&\ppp{\p{X_{2}^{n_2(\tilde{s}_1,\tilde{s}_2)}\p{1},Y^{n_2(\tilde{s}_1,\tilde{s}_2)},S^{n_2(\tilde{s}_1,\tilde{s}_2)}}\notin A_{\epsilon}^{*(n_2(\tilde{s}_1,\tilde{s}_2))}(X_{2},Y,S)|(\tilde{S}_1,\tilde{S}_2)=(\tilde{s}_1,\tilde{s}_2)},\\
   E_2&=&\ppp{ \exists i\neq1:\p{X_{2}^{n_2(\tilde{s}_1,\tilde{s}_2)}\p{i},Y^{n_2(\tilde{s}_1,\tilde{s}_2)},S^{n_2(\tilde{s}_1,\tilde{s}_2)}}\in A_{\epsilon}^{*(n_2(\tilde{s}_1,\tilde{s}_2))}(X_{2},Y,S)|(\tilde{S}_1,\tilde{S}_2)=(\tilde{s}_1,\tilde{s}_2)}.
\end{eqnarray}
    Then by the union of events bound,
   \begin{eqnarray}
   P_{e}^{\p{n_2(\tilde{s}_1,\tilde{s}_2)}}&=&\Pr\p{E_1\cup E_2 }\nonumber\\
   &\leq& P\p{E_1}+P\p{E_2}.
   \end{eqnarray}
   Now let us find the probability of each event,
\begin{enumerate}

   \item $P\p{E_1}$- As mentioned above as $n_2(\tilde{s}_1,\tilde{s}_2)\rightarrow\infty$ we have,
         \begin{eqnarray*}
         P\p{E_1}\rightarrow0.
         \end{eqnarray*}
    \item  $P\p{E_2}$- for $i\neq1$ the probability of error,
         \begin{eqnarray}
         P\p{E_2}&=&\Pr\p{\p{X_{2}^{n_2(\tilde{s}_1,\tilde{s}_2)}\p{i},Y_{1}^{n_2(\tilde{s}_1,\tilde{s}_2)},S_{1}^{n_2(\tilde{s}_1,\tilde{s}_2)}}\in A_{\epsilon}^{*(n_2(\tilde{s}_1,\tilde{s}_2))}|(\tilde{S}_1,\tilde{S}_2)=(\tilde{s}_1,\tilde{s}_2)}\nonumber\\
                 &\leq&\sum_{i=2}^{2^{n_2(\tilde{s}_1,\tilde{s}_2)R_{2}(\tilde{s}_1,\tilde{s}_2)}}P\p{E_{2,i}}\nonumber\\
                 &\leq& 2^{n_2(\tilde{s}_1,\tilde{s}_2)R_{2}(\tilde{s}_1,\tilde{s}_2)}\cdot 2^{-n_2(\tilde{s}_1,\tilde{s}_2)\p{I\p{X_{2};Y,S|\tilde{S}_1=\tilde{s}_1,\tilde{S}_2=\tilde{s}_2}-\epsilon}}.
         \end{eqnarray}
         For $P\p{E_2}\rightarrow0$ as $n_2(\tilde{s}_1,\tilde{s}_2)\rightarrow\infty$,  we need to choose,
         \begin{eqnarray}
         R_{2}(\tilde{s}_1,\tilde{s}_2)&<&I(X_{2};Y,S|\tilde{S}_1=\tilde{s}_1,\tilde{S}_2=\tilde{s}_2)-\epsilon\nonumber\\
               &=& I(X_{2};Y|S,\tilde{S}_1=\tilde{s}_1,\tilde{S}_2=\tilde{s}_2))+I(X_{2};S|\tilde{S}_1=\tilde{s}_2,\tilde{S}_2=\tilde{s}_2)         -\epsilon\nonumber\\
               &\stackrel{(a)}=&I(X_{2};X_{1},Y|S,\tilde{S}_1=\tilde{s}_1,\tilde{S}_2=\tilde{s}_2)- \epsilon,
         \end{eqnarray}
          where (a) follows from the independence of $X_2$ and $S$ given $(\tilde{S}_1=\tilde{s}_1,\tilde{S}_2=\tilde{s}_2)$.
\end{enumerate}
Similarly, we can analyze the probability of error to the rest of the codebooks of encoder $2$, i.e., $\mathcal{C}^{\tilde{s}_1,\tilde{s}_2}_{2}$ for every $(\tilde{s}_1,\tilde{s}_2)\in\{ \mathcal{S}\times\mathcal{S}\}$.
Therefore, as $n\rightarrow\infty$
\begin{eqnarray}
R_2&\leq&\sum_{\tilde{s}_1,\tilde{s}_2}\frac{n(\tilde{s}_1,\tilde{s}_2)}{n}R_{2}(\tilde{s}_1,\tilde{s}_2)\nonumber\\
&\leq& \sum_{\tilde{s}_1,\tilde{s}_2}\frac{n(\tilde{s}_1,\tilde{s}_2)}{n}(I(X_2;Y|S,\tilde{S}_1=\tilde{s}_1,\tilde{S}_2=\tilde{s}_2)-\epsilon)\nonumber\\
&=&\sum_{\tilde{s}_1,\tilde{s}_2}(P(\tilde{s}_1,\tilde{s}_2)-\epsilon')(I(X_2;Y|S,\tilde{S}_1=\tilde{s}_1,\tilde{S}_2=\tilde{s}_2)-\epsilon)\nonumber\\
&=& I(X_2;Y|S,\tilde{S}_1,\tilde{S}_2)-\epsilon'',
\end{eqnarray}
where $\epsilon''=\epsilon+\epsilon'\sum_{\tilde{s}_1,\tilde{s}_2}I(X_2;Y|S,\tilde{S}_1=\tilde{s}_1,\tilde{S}_2=\tilde{s}_2)-\epsilon\epsilon'$.
%

 Let us analyze the probability of error for the component codeword $V^{\tilde{s}_1}_{1}$.
As mention above, since that the state process is stationary and ergodic
$\lim_{n\rightarrow\infty}\frac{N(\tilde{s}_1)}{n}=P(\tilde{s}_1)$ in probability. Therefore,
 the probability that an error is declared at encoder $1$,  $\Pr(N_{\tilde{s}_1}<n_1(\tilde{s}_1))\rightarrow 0$ as $n\rightarrow\infty$.
Now, we analyze the probability to decode incorrectly the component codeword $V^{\tilde{s}_1}_{1}$, that was sent from the $\mathcal{C}^{\tilde{s}_1}_{1}$  codebook of encoder $1$ after $\hat{M_2}$ was decoded correctly. Without loss of generality, we can assume that the first codeword was sent from the $\mathcal{C}^{\tilde{s}_1}_{1}$  codebook of encoder $1$, i.e., $\mathcal{C}^{\tilde{s}_1}_{2}(1)$ was sent. Again from the ergodicity of $S^{n_1(\tilde{s}_1)}$, the construction of the codebooks, and channel transition probability we have that $\Pr \left\{(X_{1}^{n_1(\tilde{s}_1)}\p{1},X_{2}^{n_1(\tilde{s}_1)}(\hat{M_2}),Y^{n_1(\tilde{s}_1)},S^{n_1(\tilde{s}_1)}) \in   A_{\epsilon}^{*(n_1(\tilde{s}_1))}(X_{1},X_{2},Y,S)|\tilde{S}_1=\tilde{s}_1 \right\}\rightarrow 1  $ as $n_1(\tilde{s}_1)\rightarrow\infty$. A decoding error occurs only if
\begin{eqnarray}
   E_3&=&\ppp{\p{X_{1}^{n_1(\tilde{s}_1)}\p{1},X_{2}^{n_1(\tilde{s}_1)}(\hat{M_2}),Y^{n_1(\tilde{s}_1)},S^{n_1(\tilde{s}_1)}}\notin A_{\epsilon}^{*(n_1(\tilde{s}_1))}(X_{1},X_{2},Y,S)|\tilde{S}_1=\tilde{s}_1},\\
   E_4&=&\ppp{ \exists i\neq1:\p{X_{1}^{n_1(\tilde{s}_1)}\p{i},X_{2}^{n_1(\tilde{s}_1)}(\hat{M_2}),Y^{n_1(\tilde{s}_1)},S^{n_1(\tilde{s}_1)}}\in A_{\epsilon}^{*(n_1(\tilde{s}_1))}(X_{1},X_{2},Y,S)|\tilde{S}_1=\tilde{s}_1}.
\end{eqnarray}
Then by the union of events bound,
   \begin{eqnarray}
   P_{e}^{\p{n_1(\tilde{s}_1)}}&=&\Pr\p{E_3\cup E_4 }\nonumber\\
   &\leq& P\p{E_3}+P\p{E_4}.
   \end{eqnarray}
   Now let us find the probability of each event,
\begin{enumerate}

   \item $P\p{E_3}$- As mentioned above as $n_1(\tilde{s}_1)\rightarrow\infty$ we have,
         \begin{eqnarray*}
         P\p{E_3}\rightarrow0.
         \end{eqnarray*}
    \item  $P\p{E_4}$- for $i\neq1$ the probability of error,
         \begin{eqnarray}
         P\p{E_4}&=&\Pr\p{\p{X_{1}^{n_1(\tilde{s}_1)}\p{i},X_{2}^{n_1(\tilde{s}_1)}(\hat{M_2}),Y^{n_1(\tilde{s}_1)},S^{n_1(\tilde{s}_1)}}\in A_{\epsilon}^{(n_1(\tilde{s}_1))}|\tilde{S}_1=\tilde{s}_1}\nonumber\\
                 &\leq&\sum_{i=2}^{2^{n_1(\tilde{s}_1)R_{1}(1)}}P\p{E_{4,i}}\nonumber\\
                 &\leq& 2^{n_1(\tilde{s}_1)R_{1}(\tilde{s}_1)}\cdot 2^{-n_1(\tilde{s}_1)\p{I\p{X_{1};X_{2},Y,S|\tilde{S}_1=\tilde{s}_1}-\epsilon}}.
         \end{eqnarray}
         For $P\p{E_4}\rightarrow0$ as $n_1(1)\rightarrow\infty$,  we need to choose ,
         \begin{eqnarray}
         R_{1}(\tilde{s}_1)&<&I(X_{1};X_2,Y,S|\tilde{S}_1=\tilde{s}_1)-\epsilon\nonumber\\
               &=& I(X_{1};Y|X_2,S,\tilde{S}_1=\tilde{s}_1)+I(X_{1};X_2,S|\tilde{S}_1=\tilde{s}_1)         -\epsilon\nonumber\\
               &\stackrel{(a)}=&I(X_{1};Y|X_2,S,\tilde{S}_1=\tilde{s}_1)- \epsilon, \nonumber\\
               &=&H(Y|X_2,S,\tilde{S}_1=\tilde{s}_1)-H(Y|X_1,X_2,S,\tilde{S}_1=\tilde{s}_1)- \epsilon, \nonumber\\
               &\stackrel{(b)}=&H(Y|X_2,S,\tilde{S}_1=\tilde{s}_1,\tilde{S}_2)-H(Y|X_1,X_2,S,\tilde{S}_1=\tilde{s}_1,\tilde{S}_2)- \epsilon, \nonumber\\
               &=&I(X_{1};Y|X_2,S,\tilde{S}_1=\tilde{s}_1,\tilde{S}_2)- \epsilon, \nonumber
         \end{eqnarray}
          where (a) follows from the independence of $X_1$ and $(X_2,S)$ given $\tilde{S}_1=\tilde{s}_1$, and (b) follows from the independence of $Y$ and $\tilde{S}_2$ given $(X_2,S,\tilde{S}_1=\tilde{s}_1)$.
   \end{enumerate}
         Similarly, we can analyze the probability of error to the rest of the codbooks of encoder $1$, i.e., $\mathcal{C}^{\tilde{s}_1}_{1}$ for every $\tilde{s}_1\in\{ \mathcal{S}\}$. Therefore, as $n\rightarrow\infty$
         \begin{eqnarray}
R_1&\leq&\sum_{\tilde{s}_1}\frac{n(\tilde{s}_1)}{n}R_{1}(\tilde{s}_1)\nonumber\\
&\leq& \sum_{\tilde{s}_1}\frac{n(\tilde{s}_1)}{n}(I(X_1;Y|X_2,S,\tilde{S}_1=\tilde{s}_1,\tilde{S}_2)-\epsilon)\nonumber\\
&=&\sum_{\tilde{s}_1}(P(\tilde{s}_1)-\epsilon')(I(X_1;Y|X_2,S,\tilde{S}_1=\tilde{s}_1,\tilde{S}_2)-\epsilon)\nonumber\\
&=& I(X_1;Y|X_2,S,\tilde{S}_1,\tilde{S}_2)-\epsilon'',
\end{eqnarray}
where $\epsilon''=\epsilon+\epsilon'\sum_{\tilde{s}_1,\tilde{s}_2}I(X_1;Y|X_2,S,\tilde{S}_1=\tilde{s}_1,\tilde{S}_2)-\epsilon\epsilon'$.
%

Thus the total average probability of decoding error  $P_{e}^{(n)}\rightarrow0$ as $n\rightarrow\infty$ if
$ R_1<I(X_1;Y|X_2,S,\tilde{S}_1,\tilde{S}_2), R_2<I(X_2;Y|S,\tilde{S}_1,\tilde{S}_2)$. The achievability of the other corner point follows by changing the decoding order. To show achievability of other points in $\mathcal{R}(X_1,X_2)$, we use time sharing between corner points and points on the axes.
Thus, the probability of error, conditioned on a particular codeword being sent, goes to zero if the conditions of the following  are met:
\begin{eqnarray}
R_{1}&\leq& I(X_{1};Y|X_{2},S,\tilde{S}_1,\tilde{S}_2), \nonumber\\
R_{2} &\leq& I(X_{2};Y|X_{1},S,\tilde{S}_1,\tilde{S}_2), \nonumber\\
R_{1}+R_{2}&\leq& I(X_{1},X_{2};Y|S,\tilde{S}_1,\tilde{S}_2) \nonumber.
\end{eqnarray}
The above bound shows that the average probability of error, which by symmetry is equal to the probability for an individual pair of codewords $(m_1,m_2)$, averaged over all choices of codebooks in the random code construction, is arbitrarily small. Hence, there exists at least one code $(n,2^{nR_{1}},2^{nR_{2}},d_1,d_2)$ with arbitrarily small probability of error. To complete the proof we use time-sharing to allow any $(R_1,R_2)$ in the convex
hull to be achieved.
\end{proof}
 \section{ALTERNATIVE PROOF}\label{ALTERNATIVE PROOF}
%
 In this section we provide an alternative proof for Theorem \ref{Capacity region-  MAC with delayed CSI feedback t1}.
 The alternative proof is based on a multi-letter expression for the capacity region of FS-MAC with time-invariant feedback \cite{Permuter_Weissman}. In order to use the capacity region of FS-MAC with time-invariant feedback,
 we treat the knowledge of the state at the encoders as being part of the feedback from the decoder to the encoders.

 Throughout this section we use the causal conditioning notation $(\cdot||\cdot)$. We denote the probability mass function (pmf) of $Y^n$ causally conditioned on $X^{n-d}$, for some integer $d \geq 0$, as $P(y^n||x^{n-d})$ which is defined as
   \begin{eqnarray}
 P(y^n||x^{n-d})=\prod_{i=1}^nP(y_i|y^{i-1}, x^{i-d}),
 \end{eqnarray}
 (if $i - d \leq 0$ then $x^{i-d}$ is set to null).
 The directed information $I(X^n \rightarrow Y^n)$ was defined by Massey in \cite{Massey90} as
 \begin{eqnarray}
 I(X^n \rightarrow Y^n)\triangleq \sum_{i=1}^nI(X^i; Y_i|Y^{i-1}).
 \end{eqnarray}
 Directed information has been widely used in the characterization of capacity of point-to-point channels \cite{Chen_Berger}, \cite{Kim07_feedback}, \cite{Tatikonda00}, \cite{Yang05}, \cite{Permuter}, \cite{PermuterCuffVanRoyWeissman08}, compound channels \cite{Shrader:2009}, network capacity \cite{Kramer98}, rate distortion \cite{PradhanVenkataramananIT_feedforward07},\cite{zamir06}, and broadcast channel \cite{Dabora2010}. Directed information can also be expressed in terms of causal conditioning as
 \begin{eqnarray}
 I(X^n \rightarrow Y^n)&=&\sum_{i=1}^nI(X^i; Y_i|Y^{i-1})\nonumber\\
 &=& \mathbf{E} \left[ \log \frac{P(Y^n||X^n)}{P(Y^n)} \right],
 \end{eqnarray}
 where $\mathbf{E}$ denotes expectation. Directed information between $X^n_{1}$ to $Y^n$ causally conditioned on $X^{n}_{2}$ is defined as
 \begin{eqnarray}
 I(X_{1}^n \rightarrow Y^n||X_{2}^n)&\triangleq&\sum_{i=1}^nI(X_1^i; Y_i|Y^{i-1},X_{2}^i)\nonumber\\
 &=& \mathbf{E} \left[ \log \frac{P(Y^n||X_1^n,X_2^n)}{P(Y^n||X_2^n)} \right],
 \end{eqnarray}
 where $P(y^n||x_1^n,x_2^n)=\prod_{i=1}^n P(y_i|y^{i-1},x_1^i,x_2^i)$.

 Now let us present a result from \cite{Permuter_Weissman} that we need for the proof.
 Consider the FS-MAC with time-invariant feedback as illustrated in Fig. \ref{f_1}.
 The channel is characterized by a conditional probability $P(y_i,s_{i+1}|x_{1,i},x_{2,i},s_i)$
 that satisfies,
  \begin{eqnarray}
    P(y_i,s_{i+1}|x_{1}^{i},x_{2}^{i},s^i,y^{i-1})&=&P(y_i,s_{i+1}|x_{1,i},x_{2,i},s_i).
\end{eqnarray}
In addition, we assume that the channel is stationary, indecomposable, and without ISI, i.e.,
 \begin{eqnarray}
   P(y_i,s_{i+1}|x_{1,i},x_{2,i},s_i)=p(s_{i+1}|s_i)p(y_i|x_{1,i},x_{2,i},s_i)
  \end{eqnarray}
and
 \begin{eqnarray}
   P(s_0)=\pi(s_0),
  \end{eqnarray}
where $\pi(s_0)$ is the unique stationary distribution, i.e., $\lim_{n\rightarrow\infty}\Pr(S_n=s|s_0)=\pi(s_0)$, $\forall s_0 \in \mathcal{S}$.
\begin{figure}[h]{
 \psfrag{A}[][][0.8]{Encoder 1}
\psfrag{B}[][][0.8]{\ $x_{1,i}(m_1,z_1^{i-d_1})$}
\psfrag{C}[][][0.8]{Encoder 2}
\psfrag{D}[][][0.8]{\ $x_{2,i}(m_2,z_2^{i-d_2})$}
\psfrag{m1}[][][0.8]{$m_1$} \psfrag{m2}[][][0.8]{$\in\{
1,...,2^{nR_1}\}\;\;$} \psfrag{m3}[][][0.8]{$m_2$}
\psfrag{m4}[][][0.8]{$\in\{ 1,...,2^{nR_2}\}\;\;$}
 \psfrag{M}[][][0.8]{Finite State
MAC} \psfrag{P}[][][0.8]{$P(y_i,s_{i+1}|x_{1,i},x_{2,i},s_{i})$}
\psfrag{f}[][][0.8]{$z_{2,i}=f_2(y_{i})$}
\psfrag{f}[][][0.8]{Time-Invariant} \psfrag{i}[][][0.8]{Function}
\psfrag{zf2}[][][0.8]{$z_{2,i}(y_i)$}
\psfrag{zf1}[][][0.8]{$z_{1,i}(y_i)$}
\psfrag{z2}[][][0.8]{$z_{2,i-d_2}$}
\psfrag{z1}[][][0.8]{$z_{1,i-d_1}$} \psfrag{W}[][][0.8]{Decoder}

\psfrag{g1}[][][0.8]{Delay} \psfrag{h1}[][][0.8]{$d_1$}
\psfrag{g2}[][][0.8]{Delay} \psfrag{h2}[][][0.8]{$d_2$}
\psfrag{X}[][][0.8]{$\hat m_1(y^N)$} \psfrag{U}[][][0.8]{$\hat
m_2(y^N)$} \psfrag{Y}[][][0.8]{$\hat m_1, \hat m_2$}
\psfrag{D6}[][][0.8]{Function}
\psfrag{v7}[][][0.8]{$z_{i-1}(y_{i-1})$}
 \psfrag{Yi}[][][0.8]{$y_i$}
\psfrag{v6 }[][][0.8]{$\hat m$}\psfrag{w6a}[][][0.8]{Estimated}
\psfrag{w6b\r}[][][0.8]{Message} \centering
\includegraphics[width=13cm]{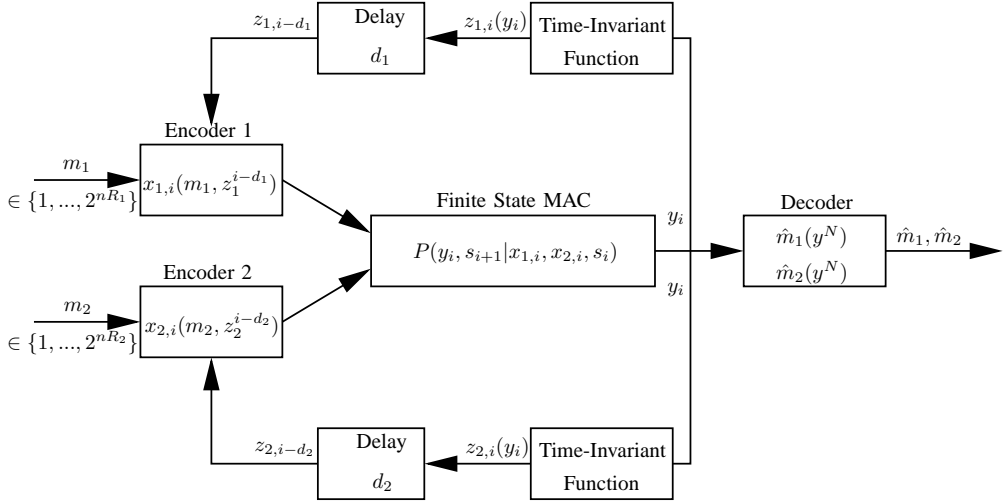}
\centering \caption{Channel with feedback, where the feedback  is a time-invariant
deterministic function of the output. } \label{f_1}
}\end{figure}
\begin{lemma}\label{lemma1}{\cite[Theorem 13]{Permuter_Weissman}}
The capacity of a stationary, indecomposable FS-MAC without ISI and with time-invariant feedback, as illustrated in Fig. \ref{f_1}, is  $ \hat{\mathcal{R}} = \lim_{n\rightarrow \infty} \hat{\mathcal{R}}_{n}$, where $\hat{\mathcal{R}}_n$ is the following region in $\mathbb{R}_{+}^2$:
\begin{eqnarray}\label{e_def_Rn}
\hat{{\mathcal R}}_n=\bigcup_{P(x_1^n||z_1^{n-d_1})P(x_2^n||z_2^{n-d_2})}\left( \begin{array}{rcl}
R_1 \leq  \frac{1}{n}I(X_1^n \to Y^n ||X_2^{n}),\\
R_1 \leq  \frac{1}{n}I(X_2^n \to Y^n ||X_1^{n}),\\
R_1+R_2 \leq  \frac{1}{n}I((X_1,X_2)^n \to Y^{n}).
\end{array}\right).
\end{eqnarray}
\end{lemma}
In \cite[Theorem 13]{Permuter_Weissman} only the case where $d_1=d_2=1$ was considered, but the result extends straightforwardly to any delay $d_1$ and $d_2$.
The following theorem provides an alternative proof for Theorem \ref{Capacity region-  MAC with delayed CSI feedback t1} based on Lemma \ref{lemma1}.
\begin{theorem}
Let us denote ${\mathcal R}_n$ and  ${\mathcal R}$ to be the following regions in $\mathbb{R}_{+}^2$:
\begin{eqnarray}
{\mathcal R}_n=\bigcup_{P(x_1^n||s^{n-d_1})P(x_2^n||s^{n-d_2})}\left( \begin{array}{rcl}
R_1 \leq  \frac{1}{n}I(X_1^n \to Y^n,S^n ||X_2^{n}),\\
R_1 \leq  \frac{1}{n}I(X_2^n \to Y^n,S^n ||X_1^{n}),\\
R_1+R_2 \leq  \frac{1}{n}I((X_1,X_2)^n \to Y^{n},S^n).
\end{array}\right).
\end{eqnarray}
\begin{eqnarray}
\mathcal{R}=\bigcup_{P(u|\tilde{s}_1)P(x_{1}|\tilde{s}_1,u)P(x_{2}|\tilde{s}_1,\tilde{s}_2,u)} \left( \begin{array}{rcl}
R_{1}<I(X_{1};Y|X_{2},S,\tilde{S}_1,\tilde{S}_2,U)\\
R_{2}<I(X_{2};Y|X_{1},S,\tilde{S}_1,\tilde{S}_2,U)\\
R_{1}+R_{2}<I(X_{1},X_{2};Y|S,\tilde{S}_1,\tilde{S}_2,U),
\end{array}\right).
\end{eqnarray}
The capacity region for the FSM-MAC with CSI at the decoder and asymmetrical delayed  CSI  at the encoders with delays $d_1$ and $d_2$, as illustrated in Fig. \ref{Multiple-Access Channel with Receiver CSI and Asymmetrical Delayed Feedback}, is $\lim_{n \to \infty}{\mathcal R}_n={\mathcal R}$.
\end{theorem}
\begin{proof}\\
In order to adapt the model in Fig. \ref{f_1} to our model, we can consider the state information at the decoder as a part of the channel's output. Therefore, the capacity region is
\begin{eqnarray}
{\mathcal R_{feedback}}=\lim_{n\to \infty}\bigcup_{P(x_1^n||z^{n-d_1})P(x_2^n||z^{n-d_2})}\left( \begin{array}{rcl}
R_1 \leq  \frac{1}{n}I(X_1^n \to Y^n,S^n ||X_2^{n}),\\
R_1 \leq  \frac{1}{n}I(X_2^n \to Y^n,S^n ||X_1^{n}),\\
R_1+R_2 \leq  \frac{1}{n}I((X_1,X_2)^n \to Y^{n},S^n).\label{extended_model_capa}
\end{array}\right).
\end{eqnarray}
Now, by choosing the deterministic function of the output $z_{1,i}(y_i,s_i)=z_{2,i}(y_i,s_i)=s_i$, (\ref{extended_model_capa}) yields the capacity region for the FSM-MAC with CSI at the decoder and asymmetrical delayed  CSI  at the encoders as shown in Fig. \ref{Multiple-Access Channel with Receiver CSI and Asymmetrical Delayed Feedback}.
Note that $\mathcal R_{feedback}=\lim_{n \to \infty}{\mathcal R}_n$, hence the capacity region is $\lim_{n \to \infty}{\mathcal R}_n$.
In order to complete the proof we need to show that $\lim_{n \to \infty}{\mathcal R}_n={\mathcal R}$.
First let us show that $\lim_{n \to \infty}{\mathcal R}_n\supseteq{\mathcal R}$,
\begin{eqnarray}
{\mathcal R}_n&=&\bigcup_{P(x_1^n||s^{n-d_1})P(x_2^n||s^{n-d_2})}\left( \begin{array}{rcl}
R_1 \leq  \frac{1}{n}I(X_1^n \to Y^n,S^n ||X_2^{n}),\\
R_1 \leq  \frac{1}{n}I(X_2^n \to Y^n,S^n ||X_1^{n}),\\
R_1+R_2 \leq  \frac{1}{n}I((X_1,X_2)^n \to Y^{n},S^n).
\end{array}\right)\nonumber\\
&=&\bigcup_{P(x_1^n||s^{n-d_1})P(x_2^n||s^{n-d_2})}\left( \begin{array}{rcl}
R_1 \leq  \frac{1}{n}\sum_{i=1}^{n}I(X_{1}^i;Y_i,S_i |X_2^{i},Y^{i-1},S^{i-1}),\\
R_1 \leq  \frac{1}{n}\sum_{i=1}^{n}I(X_{2}^i;Y_i,S_i |X_1^{i},Y^{i-1},S^{i-1}),\\
R_1+R_2 \leq  \frac{1}{n}\sum_{i=1}^{n}I(X_{1}^i,X_2^{i};Y_i,S_i |Y^{i-1},S^{i-1}).
\end{array}\right).\nonumber
\end{eqnarray}
To bound  $R_1$, consider
\begin{eqnarray}
R_1&\leq&\frac{1}{n}\sum_{i=1}^{n}I(X_{1}^i;Y_i,S_i|X_2^{i},Y^{i-1},S^{i-1})\nonumber\\
&=&\frac{1}{n}\sum_{i=1}^{n}H(Y_i,S_i|X_{2}^i,Y^{i-1},S^{i-1})-H(Y_i,S_i|X_{1}^i,X_{2}^i,Y^{i-1},S^{i-1})\nonumber\\
&=&\frac{1}{n}\sum_{i=1}^{n}H(S_i|X_{2}^i,Y^{i-1},S^{i-1})+H(Y_i|X_{2}^i,Y^{i-1},S^{i})\nonumber\\
&&-\frac{1}{n}\sum_{i=1}^{n}H(S_i|X_{1}^i,X_{2}^i,Y^{i-1},S^{i-1})+H(Y_i|X_{1}^i,X_{2}^i,Y^{i-1},S^{i})\nonumber\\
&\stackrel{(a)}=&\frac{1}{n}\sum_{i=1}^{n}H(S_i|S^{i-1})+H(Y_i|X_{2}^i,Y^{i-1},S^{i})-H(S_i|S^{i-1})-H(Y_i|X_{1,i},X_{2,i},S_{i})\nonumber\\
&=&\frac{1}{n}\sum_{i=1}^{n}H(Y_i|X_{2}^i,Y^{i-1},S^{i})-H(Y_i|X_{1,i},X_{2,i},S_{i}).\nonumber
\end{eqnarray}
Where (a) follows from the fact that the channel is without ISI, and from the fact that the channel's output at time $i$ depends only on the state $S_i$, and the inputs $X_{1,i}$, $X_{2,i}$. We can bound $R_2$ and  $R_1+R_2$ in a similar way. Hence we obtain
\begin{eqnarray}
{\mathcal R}_n&=&\bigcup_{P(x_1^n||s^{n-d_1})P(x_2^n||s^{n-d_2})}\left( \begin{array}{rcl}
R_1 \leq \frac{1}{n}\sum_{i=1}^{n}H(Y_i|X_{2}^i,Y^{i-1},S^{i})-H(Y_i|X_{1,i},X_{2,i},S_{i}),\nonumber\\
R_2 \leq \frac{1}{n}\sum_{i=1}^{n}H(Y_i|X_{1}^i,Y^{i-1},S^{i})-H(Y_i|X_{1,i},X_{2,i},S_{i}),\nonumber\\
R_1+R_2 \leq \frac{1}{n}\sum_{i=1}^{n}H(Y_i|Y^{i-1},S^{i})-H(Y_i|X_{1,i},X_{2,i},S_{i}).\nonumber\\
\end{array}\right).\nonumber
\end{eqnarray}
Now using \cite[Lemma 3]{Permuter}, we have that $P(x_1^n||s^{n-d_1})P(x_2^n||s^{n-d_2})$ determines uniquely
$\Big\{P(x_{1,i}|x_{1}^{i-1},s^{i-d_1})$ $P(x_{2,i}|x_{2}^{i-1},s^{i-d_2})\Big\}_{i=1}^{n} $, hence,
\begin{eqnarray}
{\mathcal R}_n&=&\bigcup_{\{P(x_{1,i}|x_{1}^{i-1},s^{i-d_1})P(x_{2,i}|x_{2}^{i-1},s^{i-d_2})\}_{i=1}^{n}}\left( \begin{array}{rcl}
R_1 \leq \frac{1}{n}\sum_{i=1}^{n}H(Y_i|X_{2}^i,Y^{i-1},S^{i})-H(Y_i|X_{1,i},X_{2,i},S_{i}),\nonumber\\
R_2 \leq \frac{1}{n}\sum_{i=1}^{n}H(Y_i|X_{1}^i,Y^{i-1},S^{i})-H(Y_i|X_{1,i},X_{2,i},S_{i}),\nonumber\\
R_1+R_2 \leq \frac{1}{n}\sum_{i=1}^{n}H(Y_i|Y^{i-1},S^{i})-H(Y_i|X_{1,i},X_{2,i},S_{i}).\nonumber\\
\end{array}\right).\nonumber
\end{eqnarray}
Let us assume that $d_1\geq d_2$, furthermore, we restrict the inputs of the channel by assuming that $P(x_{1,i}|x_{1}^{i-1},s^{i-d_1})=P(x_{1,i}|s_{i-d_1})$, $P(x_{2,i}|x_{2}^{i-1},s^{i-d_2})=P(x_{2,i}|s_{i-d_1}, s_{i-d_2})$. Therefore,
\begin{eqnarray}
{\mathcal R}_n&\supseteq&\bigcup_{\{P(x_{1,i}|s_{i-d_1})P(x_{2,i}|s_{i-d_1},s_{i-d_2})\}_{i=1}^{n}}\left( \begin{array}{rcl}
R_1 \leq \frac{1}{n}\sum_{i=1}^{n}H(Y_i|X_{2}^i,Y^{i-1},S^{i})-H(Y_i|X_{1,i},X_{2,i},S_{i}),\nonumber\\
R_2 \leq \frac{1}{n}\sum_{i=1}^{n}H(Y_i|X_{1}^i,Y^{i-1},S^{i})-H(Y_i|X_{1,i},X_{2,i},S_{i}),\nonumber\\
R_1+R_2 \leq \frac{1}{n}\sum_{i=1}^{n}H(Y_i|Y^{i-1},S^{i})-H(Y_i|X_{1,i},X_{2,i},S_{i}).\nonumber\\
\end{array}\right).\nonumber
\end{eqnarray}
Since we assumed that $P(x_{1,i}|x_{1}^{i-1},s^{i-d_1})=P(x_{1,i}|s_{i-d_1})$, we have the following equalities,
\begin{eqnarray}
P(y_i|x_{2}^i,y^{i-1},s^i)&=&\sum_{x_{1,i}}P(x_{1,i}|x_{2}^i,y^{i-1},s^i)P(y_i|x_{1,i},x_{2}^i,y^{i-1},s^i)\nonumber\\
                            &\stackrel{(a)}=&\sum_{x_{1,i}}P(x_{1,i}|s_i,s_{i-d_1},s_{i-d_2})P(y_i|x_{1,i},x_{2,i},s_i,s_{i-d_1},s_{i-d_2})\nonumber\\
                          &=&P(y_i|x_{2,i},s_i,s_{i-d_1},s_{i-d_2})\label{restriction},
\end{eqnarray}
where (a) follows from the fact that the channel's output at time $i$ depends only on the state $S_i$, and the inputs $X_{1,i}$, $X_{2,i}$, and from the fact that $P(x_{1,i}|x_{2}^i,y^{i-1},s^i)=P(x_{1,i}|s^{i-d_1})=P(x_{1,i}|s_{i-d_1})$. From (\ref{restriction}) we get
\begin{eqnarray}
H(Y_i|X_{2}^i,Y^{i-1},S^{i})&=&H(Y_i|X_{2,i},S_i,S_{i-d_1},S_{i-d_2}).\nonumber
\end{eqnarray}
Similarly,
 \begin{eqnarray}
H(Y_i|X_{1}^i,Y^{i-1},S^{i})&=&H(Y_i|X_{1,i},S_i,S_{i-d_1},S_{i-d_2}). \nonumber\\
H(Y_i|Y^{i-1},S^{i})&=&H(Y_i|S_i,S_{i-d_1},S_{i-d_2}). \nonumber
\end{eqnarray}
Therefore,
\begin{eqnarray}
{\mathcal R}_n&\supseteq&\bigcup_{\{P(x_{1,i}|s_{i-d_1})P(x_{2,i}|s_{i-d_1},s_{i-d_2})\}_{i=1}^{n}}\left( \begin{array}{rcl}
R_1 \leq \frac{1}{n}\sum_{i=1}^{n}I(Y_i;X_{1,i}|X_{2,i},S_i,S_{i-d_1},S_{i-d_2})   ,\nonumber\\
R_2 \leq \frac{1}{n}\sum_{i=1}^{n}I(Y_i;X_{2,i}|X_{1,i},S_i,S_{i-d_1},S_{i-d_2}),\nonumber\\
R_1+R_2 \leq \frac{1}{n}\sum_{i=1}^{n}I(Y_i;X_{1,i},X_{2,i}|S_i,S_{i-d_1},S_{i-d_2}).\nonumber\\
\end{array}\right).\nonumber
\end{eqnarray}
Now, in order to obtain that $\lim_{n \to \infty}{\mathcal R}_n\supseteq{\mathcal R}$,  we need to show that
 \begin{eqnarray}
\mathcal{R}&\subseteq&\lim_{n\to \infty}\bigcup_{\{P(x_{1,i}|s_{i-d_1})P(x_{2,i}|s_{i-d_1},s_{i-d_2})\}_{i=1}^{n}}\left( \begin{array}{rcl}
R_1 \leq \frac{1}{n}\sum_{i=1}^{n}I(Y_i;X_{1,i}|X_{2,i},S_i,S_{i-d_1},S_{i-d_2})   ,\nonumber\\
R_2 \leq \frac{1}{n}\sum_{i=1}^{n}I(Y_i;X_{2,i}|X_{1,i},S_i,S_{i-d_1},S_{i-d_2}),\nonumber\\
R_1+R_2 \leq \frac{1}{n}\sum_{i=1}^{n}I(Y_i;X_{1,i},X_{2,i}|S_i,S_{i-d_1},S_{i-d_2}).\nonumber\\
\end{array}\right). \nonumber
\end{eqnarray}
Consider the region $\mathcal{R}$, an achievable region is uniquely determined for every fixed joint distribution $P(u|\tilde{s}_1)P(x_1|\tilde{s}_1,u)P(x_2|\tilde{s}_1,\tilde{s}_2,u)$. The rate $R_1$ is given by
 \begin{eqnarray}
 R_1&\leq&I(X_{1};Y|X_{2},S,\tilde{S}_1,\tilde{S}_2,U)\nonumber\\
 &=&\sum_{\tilde{s}_1}P(\tilde{s}_1)\sum_{u}P(u|\tilde{s}_1)I(X_{1};Y|X_{2},S,\tilde{S}_1=\tilde{s}_1,\tilde{S}_2,U=u).\label{rateu}
 \end{eqnarray}
 In addition, we have
 \begin{eqnarray}
\frac{1}{n}\sum_{i=1}^{n}I(Y_i;X_{1,i}|X_{2,i},S_i,S_{i-d_1},S_{i-d_2})&=&\frac{1}{n}\sum_{i=1}^{n}\sum_{s_{i-d_1}}P(s_{i-d_1})I(Y_i;X_{1,i}|X_{2,i},S_i,S_{i-d_1}=s_{i-d_1},S_{i-d_2})\nonumber\\
 &\stackrel{(a)}=&\sum_{\tilde{s}_1}P(\tilde{s}_1)\sum_{i=1}^{n}\frac{1}{n}I(Y_i;X_{1,i}|X_{2,i},S_i,S_{i-d_1}=\tilde{s}_1,S_{i-d_2}),\label{raten}
 \end{eqnarray}
 where (a) follows from the fact  that the distribution $P(s_{i-d_1})$ is stationary, therefore  $P(s_{i-d_1})=P(\tilde{s}_1)$.
 For every $U=u$ and $\tilde{S}_1=\tilde{s}_1$,
 if $P(U=u|\tilde{S}_1=\tilde{s}_1)$ is rational, i.e., $k(u,\tilde{s}_1)/n$, where $k(u,\tilde{s}_1) \in  \mathbb{N}$, then we can chose $k(u,\tilde{s}_1)$ terms from $\{P(x_{1,i}|s_{i-d_1})P(x_{2,i}|s_{i-d_1},s_{i-d_2})\}_{i=1}^{n}$ such that $P(x_{1,i}|s_{i-d_1})P(x_{2,i}|s_{i-d_1},s_{i-d_2})=P(x_1|\tilde{s}_1,u)P(x_2|\tilde{s}_1,\tilde{s}_2,u)$.
 If $P(U=u|\tilde{S}_1=\tilde{s}_1)$ is irrational, we can get arbitrarily close to $P(U=u|\tilde{S}_1=\tilde{s}_1)$ by
 using longer and longer block lengths.
 Therefore, using (\ref{rateu}) and (\ref{raten}) we have that when $n \rightarrow \infty $, for  every given joint distribution $P(u|\tilde{s}_1)P(x_1|\tilde{s}_1,u)P(x_2|\tilde{s}_1,\tilde{s}_2,u)$, we can choose  $\{P(x_{1,i}|s_{i-d_1})P(x_{2,i}|s_{i-d_1},s_{i-d_2})\}_{i=1}^{n}$ such that
 \begin{eqnarray}
 \lim_{n\to \infty}\frac{1}{n}\sum_{i=1}^{n}I(Y_i;X_{1,i}|X_{2,i},S_i,S_{i-d_1},S_{i-d_2})&=&I(X_{1};Y|X_{2},S,\tilde{S}_1,\tilde{S}_2,U).\nonumber
  \end{eqnarray}
 By using the same argument for $R_2$ and  for $R_1+R_2$, we get that for  every given joint distribution $P(u|\tilde{s}_1)P(x_1|\tilde{s}_1,u)P(x_2|\tilde{s}_1,\tilde{s}_2,u)$, we can chose  $\{P(x_{1,i}|s_{i-d_1})P(x_{2,i}|s_{i-d_1},s_{i-d_2})\}_{i=1}^{n}$ such that the following equalities hold simultaneously,
 \begin{eqnarray}
 \lim_{n\to \infty}\frac{1}{n}\sum_{i=1}^{n}I(Y_i;X_{1,i}|X_{2,i},S_i,S_{i-d_1},S_{i-d_2})&=&I(X_{1};Y|X_{2},S,\tilde{S}_1,\tilde{S}_2,U),\label{eq1}\\
  \lim_{n\to \infty}\frac{1}{n}\sum_{i=1}^{n}I(Y_i;X_{2,i}|X_{1,i},S_i,S_{i-d_1},S_{i-d_2})&=&I(X_{2};Y|X_{1},S,\tilde{S}_1,\tilde{S}_2,U),\label{eq2}\\
  \lim_{n\to \infty}\frac{1}{n}\sum_{i=1}^{n}I(Y_i;X_{1,i},X_{2,i}|S_i,S_{i-d_1},S_{i-d_2})&=&I(X_{1},X_{2};Y|S,\tilde{S}_1,\tilde{S}_2,U)\label{eq3}.
  \end{eqnarray}
 Using equations (\ref{eq1}), (\ref{eq2}), and (\ref{eq3}), we obtain
 \begin{eqnarray}
 \lim_{n \to \infty}{\mathcal R}_n\supseteq{\mathcal R}.\label{Rnbigger}
  \end{eqnarray}
  In order to complete the proof, we need to show that $\lim_{n \to \infty}{\mathcal R}_n\subseteq{\mathcal R}$.
  We have that,
\begin{eqnarray}
{\mathcal R}_n&=&\bigcup_{\{P(x_{1,i}|x_{1}^{i-1},s^{i-d_1})P(x_{2,i}|x_{2}^{i-1},s^{i-d_2})\}_{i=1}^{n}}\left( \begin{array}{rcl}
R_1 \leq \frac{1}{n}\sum_{i=1}^{n}H(Y_i|X_{2}^i,Y^{i-1},S^{i})-H(Y_i|X_{1,i},X_{2,i},S_{i}),\nonumber\\
R_2 \leq \frac{1}{n}\sum_{i=1}^{n}H(Y_i|X_{1}^i,Y^{i-1},S^{i})-H(Y_i|X_{1,i},X_{2,i},S_{i}),\nonumber\\
R_1+R_2 \leq \frac{1}{n}\sum_{i=1}^{n}H(Y_i|Y^{i-1},S^{i})-H(Y_i|X_{1,i},X_{2,i},S_{i}).\nonumber\\
\end{array}\right)\nonumber
\end{eqnarray}
Consider the rate $R_1$,
\begin{eqnarray}
R_1 &\leq& \frac{1}{n}\sum_{i=1}^{n}H(Y_i|X_{2}^i,Y^{i-1},S^{i})-H(Y_i|X_{1,i},X_{2,i},S_{i})\nonumber\\
    &\leq& \frac{1}{n}\sum_{i=1}^{n}H(Y_i|X_{2,i},S_{i},S_{i-d_2},S^{i-d_1})-H(Y_i|X_{1,i},X_{2,i},S_{i})\nonumber\\
    &=&\frac{1}{n}\sum_{i=1}^{n}I(Y_i;X_{1,i}|X_{2,i},S_{i},S_{i-d_2},S^{i-d_1})\nonumber.
\end{eqnarray}
We can bound $R_2$ and  $R_1+R_2$ in a similar way. Hence we get
\begin{eqnarray}
\mathcal{R}_n &\subseteq& \bigcup_{\{P(x_{1,i}|x_{1}^{i-1},s^{i-d_1})P(x_{2,i}|x_{2}^{i-1},s^{i-d_2})\}_{i=1}^{n}}\left( \begin{array}{rcl}
R_1 \leq \frac{1}{n}\sum_{i=1}^{n}I(Y_i;X_{1,i}|X_{2,i},S_{i},S_{i-d_2},S^{i-d_1}),\\
R_2 \leq \frac{1}{n}\sum_{i=1}^{n}I(Y_i;X_{2,i}|X_{1,i},S_{i},S_{i-d_2},S^{i-d_1}),\\
R_1+R_2 \leq \frac{1}{n}\sum_{i=1}^{n}I(Y_i;X_{1,i},X_{2,i}|S_{i},S_{i-d_2},S^{i-d_1}).\\
\end{array}\right).\label{Rncon}
\end{eqnarray}
 Now, consider the joint distribution $P(s_i,s_{i-d_2},s^{i-d_1},x_{1,i},x_{2,i},y_i)$,
 \begin{eqnarray}
 P(s_i,s_{i-d_2},s^{i-d_1},x_{1,i},x_{2,i},y_i)&=&P(s_i,s_{i-d_2},s^{i-d_1})P(x_{1,i}|s^{i-d_1})P(x_{2,i}|x_{1,i},s^{i-d_1},s_{i-d_2})P(y_i|x_{1,i},x_{2,i},s_i)\nonumber\\
                                            &\stackrel{(a)}=&P(s_i,s_{i-d_2},s^{i-d_1})P(x_{1,i}|s^{i-d_1})P(x_{2,i}|s^{i-d_1},s_{i-d_2})P(y_i|x_{1,i},x_{2,i},s_i),\nonumber
 \end{eqnarray}
 where (a) follows from the fact that,
 \begin{eqnarray}
  P(x_{2,i}|x_{1,i},s^{i-d_1},s_{i-d_2})&=&\sum_{M_2,s_{i-d_1+1}^{i-d_2-1}}P(M_2,s_{i-d_1+1}^{i-d_2-1}|x_{1,i},s^{i-d_1},s_{i-d_2})P(x_{2,i}|x_{1,i},s^{i-d_2},M_2)\nonumber\\
                                        &=&\sum_{M_2,s_{i-d_1+1}^{i-d_2-1}}P(M_2,s_{i-d_1+1}^{i-d_2-1}|s^{i-d_1},s_{i-d_2})P(x_{2,i}|s^{i-d_2},M_2)\nonumber\\
                                        &=&P(x_{2,i}|s^{i-d_1},s_{i-d_2}).\nonumber
 \end{eqnarray}
 Note that $R_1$, $R_2$, and $R_1+R_2$ are uniquely determined by the joint distribution
$\left\{P(s_i,s_{i-d_2},s^{i-d_1},x_{1,i},x_{2,i},y_i)\right\}_{i=1}^n$.
 In the joint distribution $P(s_i,s_{i-d_2},s^{i-d_1},x_{1,i},x_{2,i},y_i)$, we control only $P(x_{1,i}|s^{i-d_1})P(x_{2,i}|s^{i-d_1},s_{i-d_2})$, since
 the distributions  $P(s_i,s_{i-d_2},s^{i-d_1})$ and $P(y_i|x_{1,i},x_{2,i},s_i)$ are determined by the channel transition probability. Hence,
 \begin{eqnarray}
{\mathcal R}_n
\subseteq\bigcup_{W}\left( \begin{array}{rcl}
R_1 \leq \frac{1}{n}\sum_{i=1}^{n}I(Y_i;X_{1,i}|X_{2,i},S_{i},S_{i-d_1},S_{i-d_2},S^{i-d_1-1}),\\
R_2 \leq \frac{1}{n}\sum_{i=1}^{n}I(Y_i;X_{2,i}|X_{1,i},S_{i},S_{i-d_1},S_{i-d_2},S^{i-d_1-1}),\\
R_1+R_2 \leq \frac{1}{n}\sum_{i=1}^{n}I(Y_i;X_{1,i},X_{2,i}|S_{i},S_{i-d_1},S_{i-d_2},S^{i-d_1-1}).\\
\end{array}\right),\nonumber
\end{eqnarray}
where $W\triangleq\left\{P(x_{1,i}|s^{i-d_1})P(x_{2,i}|s^{i-d_1},s_{i-d_2})\right\}_{i=1}^n$.
 In the same way as we did in the proof of the converse ( Section \ref{CONVERSE}, equation  (\ref{Uconverse})), we can rewrite these equations with the new variable $Q$,
  where $Q=i\in \{1,2,...,n\}$ with probability $\frac{1}{n}$.
   Furthermore, we denote $ X_{1}\triangleq X_{1,Q}, X_{2} \triangleq X_{2,Q}, Y \triangleq Y_{Q}, S\triangleq S_{Q}, \tilde{S}_1\triangleq S_{Q-d_1},\tilde{S}_2\triangleq S_{Q-d_2}$, and  $U\triangleq (S^{Q-d_1-1},Q)$. Hence we derive that,
 \begin{eqnarray}
{\mathcal R}_n
&\subseteq&
\bigcup_{P(u|\tilde{s}_1)P(x_{1}|\tilde{s}_1,u)P(x_{2}|\tilde{s}_1,\tilde{s}_2,u)} \left( \begin{array}{rcl}
R_{1}<I(X_{1};Y|X_{2},S,\tilde{S}_1,\tilde{S}_2,U)\\
R_{2}<I(X_{2};Y|X_{1},S,\tilde{S}_1,\tilde{S}_2,U)\\
R_{1}+R_{2}<I(X_{1},X_{2};Y|S,\tilde{S}_1,\tilde{S}_2,U),
\end{array}\right).
\end{eqnarray}
Which completes the alternative proof of Theorem \ref{Capacity region-  MAC with delayed CSI feedback t1}.
\end{proof}
  \section{EXAMPLES}\label{EXAMPLES}
  In this section we apply the general results of Section \ref{Main results} to obtain the capacity region for a finite-state Gaussian MAC, and for the finite-state multiple-access fading channel.
  We derive optimization problems on the power allocation that maximizes  the capacity region for these channels. This power allocation would be the optimal power control policy for maximizing throughput in the presence of feedback delay.
\subsection{Capacity  Region for a Finite State Additive Gaussian MAC }\label{GAUSSIAN }
We now apply Theorem \ref{Capacity region-  MAC with delayed CSI feedback t1} to compute the capacity region of a power-constrained FS additive Gaussian noise (AGN) MAC, and illustrate the effect of the delayed CSI on the capacity region. For a finite state AGN MAC the channel output $Y_i$ at time $i$, given the channel inputs $X_{1,i},X_{2,i}$, is given by
\begin{eqnarray}
  Y_i &=&X_{1,i}+X_{2,i} + N_{S_i} ,
\end{eqnarray}
where $N_{S_i}$ is a zero-mean Gaussian random variable with variance depending on the state $S_i$ of the channel at time $i$. In addition to the channel output $Y_i$  the receiver has accesses to the state  $S_i$. The receiver feeds back the CSI to the transmitters through a noiseless feedback channel.
The CSI from the receiver is received at transmitter $1$ and transmitter $2$ after a time
delays of $d_1 ,d_2$ symbol durations, respectively. The state process is assumed to be Markov with steady state distribution $\pi(s)$ and one step transition matrix $K$. It is clear that the finite state AGN is an FSMC. While the capacity region formula derived in Section \ref{Main results} (Theorem \ref{Capacity region-  MAC with delayed CSI feedback t1}) was for  finite inputs and output alphabets, the result can be generalized to continuous alphabets with inputs constraints.
 First, we apply only the sum rate formula to explicitly  determine the sum rate of the finite state Markov AGN MAC with transmitters power constraints $\mathcal{P}_1$ and $\mathcal{P}_2$.
\begin{equation}\label{}
    R_1+R_2< \max_{p(u|\tilde{s}_1)p(x_1|\tilde{s}_1,u)p(x_2|\tilde{s}_1,\tilde{s}_2,u)}I(X_1,X_2;Y|S,\tilde{S}_1,\tilde{S}_2,U),
\end{equation}
subject to the power constraints,
\begin{eqnarray}
  &&\sum_{\tilde{s}_1}\pi(\tilde{s}_1)\sum_{u}P(u|\tilde{s}_1)E[X_1^2|\tilde{s}_1,u] \leq \mathcal{P}_1 ,\\
  &&\sum_{\tilde{s}_1}\pi(\tilde{s}_1)\sum_{\tilde{s}_2}P(\tilde{s}_2|\tilde{s}_1)\sum_{u}P(u|\tilde{s}_1)E[X_2^2|\tilde{s}_1,\tilde{s}_2,u] \leq \mathcal{P}_2.
  \end{eqnarray}

To compute the maximum sum rate explicitly, we have to first determine the distributions $P(x_1|\tilde{s}_1,u)$ and $P(x_2|\tilde{s}_1,\tilde{s}_2,u)$ for each $\tilde{S}_1$, $\tilde{S}_2$, and $U$.
Suppose $\mathcal{P}_1(\tilde{s}_1,u)$ , $\mathcal{P}_2(\tilde{s}_1,\tilde{s}_2,u)$ is the power allocated to states $(\tilde{s}_1,\tilde{s}_2)$ and $u$. Therefore the sum rate,
\begin{eqnarray}
  I(X_1,X_2;Y|S,\tilde{S}_1,\tilde{S}_2,U) &=& \sum_{\tilde{s}_1}\pi(\tilde{s}_1)\sum_{\tilde{s}_2}P(\tilde{s}_2|\tilde{s}_1)\sum_{s}P(s|\tilde{s}_2)\sum_{u}P(u|\tilde{s}_1)I(X_1,X_2;Y|s,\tilde{s}_1,\tilde{s}_2,u)\nonumber\\
   &\stackrel{(a)}=&\sum_{\tilde{s}_1}\pi(\tilde{s}_1)\sum_{\tilde{s}_2}P(\tilde{s}_2|\tilde{s}_1)\sum_{s}P(s|\tilde{s}_2)\sum_{u}P(u|\tilde{s}_1) \nonumber\\&&\times (h(X_1+X_2+N_s|s,\tilde{s}_1,\tilde{s}_2,u)-h(N_s|s)) \nonumber\\
   &\stackrel{(b)}\leq&\sum_{\tilde{s}_1}\pi(\tilde{s}_1)\sum_{\tilde{s}_2}P(\tilde{s}_2|\tilde{s}_1)\sum_{s}P(s|\tilde{s}_2)\sum_{u}P(u|\tilde{s}_1)\nonumber\\&&\times\frac{1}{2}\log \left(\frac{E[(X_1+X_2+N_s)^2|s,\tilde{s}_1,\tilde{s}_2,u]}{E[N_s^2|s]}  \right)\nonumber\\
   &\stackrel{(c)}=&\sum_{\tilde{s}_1}\pi(\tilde{s}_1)\sum_{\tilde{s}_2}P(\tilde{s}_2|\tilde{s}_1)\sum_{s}P(s|\tilde{s}_2)\sum_{u}P(u|\tilde{s}_1)\nonumber\\&&\times\frac{1}{2}\log \left(1+\frac{\mathcal{P}_1(\tilde{s}_1,u)+\mathcal{P}_2(\tilde{s}_1,\tilde{s}_2,u)}{\sigma_{s}^2}  \right)\nonumber\\
   &\stackrel{(d)}\leq&\frac{1}{2}\sum_{\tilde{s}_1}\pi(\tilde{s}_1)\sum_{\tilde{s}_2}P(\tilde{s}_2|\tilde{s}_1)\sum_{s}P(s|\tilde{s}_2)\log \left(1+\frac{\mathcal{P}_1(\tilde{s}_1)+\mathcal{P}_2(\tilde{s}_1,\tilde{s}_2)}{\sigma_{s}^2}\right),\label{FSM AGN MAX}
\end{eqnarray}
where \\
(a) follows from the fact that $N_s$ is independent of  $\tilde{S}_1,\tilde{S}_2,U$ given $S$.\\
(b) follows from the fact that Gaussian distribution has the largest entropy for a given variance.\\
(c) follows from the fact that $X_1$, $X_2$ are independent of $N_s$ and independent of each other given $S, \tilde{S}_1, \tilde{S}_2$, and $U$. Furthermore, we denote $\mathcal{P}_1(\tilde{s}_1)=E[X_1^2|s,\tilde{s}_1]$, and $\mathcal{P}_2(\tilde{s}_1,\tilde{s}_2,u)=E[X_2^2|\tilde{s}_1,\tilde{s}_2,u]$.   \\
(d) follows from Jensen's inequality. \\
Furthermore, we can achieve (\ref{FSM AGN MAX}) if we choose $X_1(\tilde{s}_1,u)$, to be zero-mean Gaussian with variance $\mathcal{P}_1(\tilde{s}_1)$, and $X_2(\tilde{s}_1,\tilde{s}_1,u)$ to be zero-mean Gaussian with variance $\mathcal{P}_2(\tilde{s}_1,\tilde{s}_2)$, both independent of $N_s$ and independent of each other. We now have the following result, For an FSM AGN MAC with average power constraints $\mathcal{P}_1$ and $\mathcal{P}_2$ and  CSI at the transmitters  with delays $d_1$ and $d_2$,
\begin{eqnarray}
  R_1+R_2 &=& \max_{\mathcal{P}_1(\tilde{s}_1),\mathcal{P}_2(\tilde{s}_1,\tilde{s}_2)} \frac{1}{2}\sum_{\tilde{s}_1}\pi(\tilde{s}_1)\sum_{\tilde{s}_2}P(\tilde{s}_2|\tilde{s}_1)\sum_{s}P(s|\tilde{s}_2)\log \left(1+\frac{\mathcal{P}_1(\tilde{s}_1)+\mathcal{P}_2(\tilde{s}_1,\tilde{s}_2)}{\sigma_{s}^2}\right)\nonumber \\
  &=&\max_{\mathcal{P}_1(\tilde{s}_1),\mathcal{P}_2(\tilde{s}_1,\tilde{s}_2)} \frac{1}{2}\sum_{\tilde{s}_1}\pi(\tilde{s}_1)\sum_{\tilde{s}_2}K^{d_1-d_2}(\tilde{s}_2,\tilde{s}_1)\sum_{s}K^{d_2}(s,\tilde{s}_2) \nonumber\\&&\times\log\left(1+\frac{\mathcal{P}_1(\tilde{s}_1)+\mathcal{P}_2(\tilde{s}_1,\tilde{s}_2)}{\sigma_{s}^2}\right),\label{R1R2maxas}
\end{eqnarray}
subject to the power constraints,
\begin{eqnarray}
  &&\sum_{\tilde{s}_1}\pi(\tilde{s}_1)\mathcal{P}_1(\tilde{s}_1) \leq \mathcal{P}_1 ,\\
  &&\sum_{\tilde{s}_1}\pi(\tilde{s}_1)\sum_{\tilde{s}_2}P(\tilde{s}_2|\tilde{s}_1)\mathcal{P}_2(\tilde{s}_1,\tilde{s}_2) \leq \mathcal{P}_2.
  \end{eqnarray}
  Similarly, we can derive maximization on $R_1$ and $R_2$, for $R_1$:
   \begin{eqnarray}
   R_1=\max_{\mathcal{P}_1(\tilde{s}_1)} \frac{1}{2}\sum_{\tilde{s}_1}\pi(\tilde{s}_1)\sum_{\tilde{s}_2}K^{d_1-d_2}(\tilde{s}_2,\tilde{s}_1)\sum_{s}K^{d_2}(s,\tilde{s}_2) \log\left(1+\frac{\mathcal{P}_1(\tilde{s}_1)}{\sigma_{s}^2}\right),\label{R1max}
    \end{eqnarray}
  subject to the power constraint,
\begin{eqnarray}
\sum_{\tilde{s}_1}\pi(\tilde{s}_1)\mathcal{P}_1(\tilde{s}_1) \leq \mathcal{P}_1 ,
 \end{eqnarray}
and for $R_2$:
\begin{eqnarray}
  \max_{\mathcal{P}_2(\tilde{s}_1,\tilde{s}_2)} \frac{1}{2}\sum_{\tilde{s}_1}\pi(\tilde{s}_1)\sum_{\tilde{s}_2}K^{d_1-d_2}(\tilde{s}_2,\tilde{s}_1)\sum_{s}K^{d_2}(s,\tilde{s}_2) \log\left(1+\frac{\mathcal{P}_2(\tilde{s}_1,\tilde{s}_2)}{\sigma_{s}^2}\right),\label{R2max}
\end{eqnarray}
subject to the power constraint,
\begin{eqnarray}
\sum_{\tilde{s}_1}\pi(\tilde{s}_1)\sum_{\tilde{s}_2}P(\tilde{s}_2|\tilde{s}_1)\mathcal{P}_2(\tilde{s}_1,\tilde{s}_2) \leq \mathcal{P}_2.
\end{eqnarray}
It is important to mention that in the general case the three equations (\ref{R1R2maxas}), (\ref{R1max}), and (\ref{R2max}) do not achieve their maximum in the same distribution, i.e., not in the same power allocation. In the same way we can derive the maximization problem for two special cases. The first case is $d=d_1=d_2$, since the delays are the same we denote $\tilde{S}=\tilde{S}_1=\tilde{S}_2$, hence we have,
\begin{eqnarray}
&& R_1=\max_{\mathcal{P}_1(\tilde{s})} \frac{1}{2}\sum_{\tilde{s}}\pi(\tilde{s})\sum_{s}K^{d}(s,\tilde{s})\log\left(1+\frac{\mathcal{P}_1(\tilde{s})}{\sigma_{s}^2}\right),\\
&& R_2=\max_{\mathcal{P}_2(\tilde{s})} \frac{1}{2}\sum_{\tilde{s}}\pi(\tilde{s})\sum_{s}K^{d}(s,\tilde{s})\log\left(1+\frac{\mathcal{P}_2(\tilde{s})}{\sigma_{s}^2}\right),\\
&&R_1+R_2=\max_{\mathcal{P}_1(\tilde{s}),\mathcal{P}_2(\tilde{s})} \frac{1}{2}\sum_{\tilde{s}}\pi(\tilde{s})\sum_{s}K^{d}(s,\tilde{s})\log\left(1+\frac{\mathcal{P}_1(\tilde{s})+\mathcal{P}_2(\tilde{s})}{\sigma_{s}^2}\right),\label{R1R2maxsy}
\end{eqnarray}
subject to the power constraints,
\begin{eqnarray}
  &&\sum_{\tilde{s}}\pi(\tilde{s})\mathcal{P}_1(\tilde{s}) \leq \mathcal{P}_1 ,\\
  &&\sum_{\tilde{s}}\pi(\tilde{s})\mathcal{P}_2(\tilde{s}) \leq \mathcal{P}_2 .
  \end{eqnarray}
The second case is $d_2\leq d_1=\infty$, let us denote $d=d_2$ and $\tilde{S}=\tilde{S}_2$, therefore we have,
\begin{eqnarray}
&& R_1= \frac{1}{2}\sum_{\tilde{s}}\pi(\tilde{s})\sum_{s}K^{d}(s,\tilde{s})\log\left(1+\frac{\mathcal{P}_1}{\sigma_{s}^2}\right),\\
&& R_2=\max_{\mathcal{P}_2(\tilde{s})} \frac{1}{2}\sum_{\tilde{s}}\pi(\tilde{s})\sum_{s}K^{d}(s,\tilde{s})\log\left(1+\frac{\mathcal{P}_2(\tilde{s})}{\sigma_{s}^2}\right),\\
&&R_1+R_2=\max_{\mathcal{P}_2(\tilde{s})} \frac{1}{2}\sum_{\tilde{s}}\pi(\tilde{s})\sum_{s}K^{d}(s,\tilde{s})\log\left(1+\frac{\mathcal{P}_1+\mathcal{P}_2(\tilde{s})}{\sigma_{s}^2}\right),\label{R1R2only one}
\end{eqnarray}
subject to the power constraints,
\begin{eqnarray}
  \sum_{\tilde{s}}\pi(\tilde{s})\mathcal{P}_2(\tilde{s}) &\leq& \mathcal{P}_2 .
  \end{eqnarray}
Now to gain some intuition on the capacity region, we consider the case when there are only two states. At any given time $i$ the channel is in one of two possible states $G$ or $B$. In the good state $G$, the channel is "good" and the noise variance is $\sigma_G^2$, and in the bad state $B$, the  channel is "bad" and the noise variance is $\sigma_B^2$,
where $\sigma_B^2>\sigma_G^2$. The state process is specified by the transition
probabilities given by
\begin{eqnarray}
  P(G|B)&=&g,\nonumber \\
  P(B|G)&=&b.\nonumber
  \end{eqnarray}
  The state process is illustrated in Fig. \ref{Two-state AGN channel}, the steady state distribution of the Markov chain is given by
  \begin{eqnarray}
  \pi(G)&=&\frac{g}{g+b},\nonumber \\
  \pi(G)&=&\frac{b}{b+g}.\nonumber
  \end{eqnarray}
\begin{figure}[h!]{
  \begin{center}
\psfrag{g}[][][1]{ $g$ }  \psfrag{b}[][][1]{$b$}
\psfrag{B}[][][1]{ $B$ }  \psfrag{G}[][][1]{$G$}
\psfrag{a}[][][1]{\ \ \ \  $1-g$ }  \psfrag{d}[][][1]{$1-b$}
\includegraphics[width=9cm]{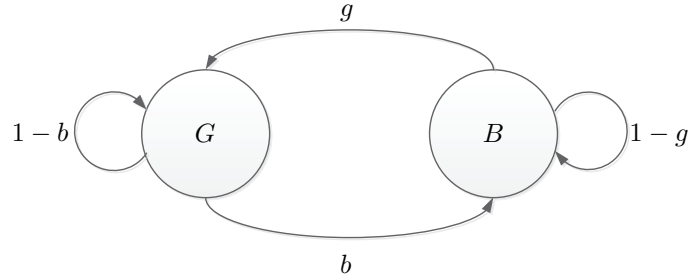}\ \ \ \ \ \ \ \
  \end{center}
\caption{Two-state AGN channel}
 \label{Two-state AGN channel}
}\end{figure}\\
By solving  the optimization problems (\ref{R1R2maxas}), (\ref{R1R2maxsy}), and (\ref{R1R2only one}) for the two state example, we present the maximum sum rate versus delay plot in Fig. \ref{Sum rate vs delay}, which shows  the effect of the CSI delay on the sum rate for  $\mathcal{P}_1=10, \mathcal{P}_2=10, \sigma_G^2=1, \sigma_B^2=100,g=0.1,b=0.1$. The details on solving the optimization problem for the two state example are presented in Appendix \ref{DETERMINATION OF THE TWO-STATE MAC CAPACITY REGION}.
\begin{figure}[h!]{
  \begin{center}
\psfrag{D}[][][0.8]{Delay $d_2$ (symbols)}  \psfrag{S}[][][0.8]{Sum rate (bits/symbol)} \psfrag{T}[][][0.8]{Sum rate vs.  delay $d_2$ ($d_1=\infty$)}\psfrag{Q}[][][0.8]{Delay $d$ (symbols)}\psfrag{E}[][][0.8]{Sum rate vs.  delay $d$ (symmetrical delay $d_1=d_2=d$)} \psfrag{Y}[][][0.8]{Delay $d_1$ (symbols)} \psfrag{U}[][][0.8]{Sum rate vs.  delay $d_1$ (asymmetrical delay $d_2=0$)}
\subfloat[]{\includegraphics[width=8.5cm]{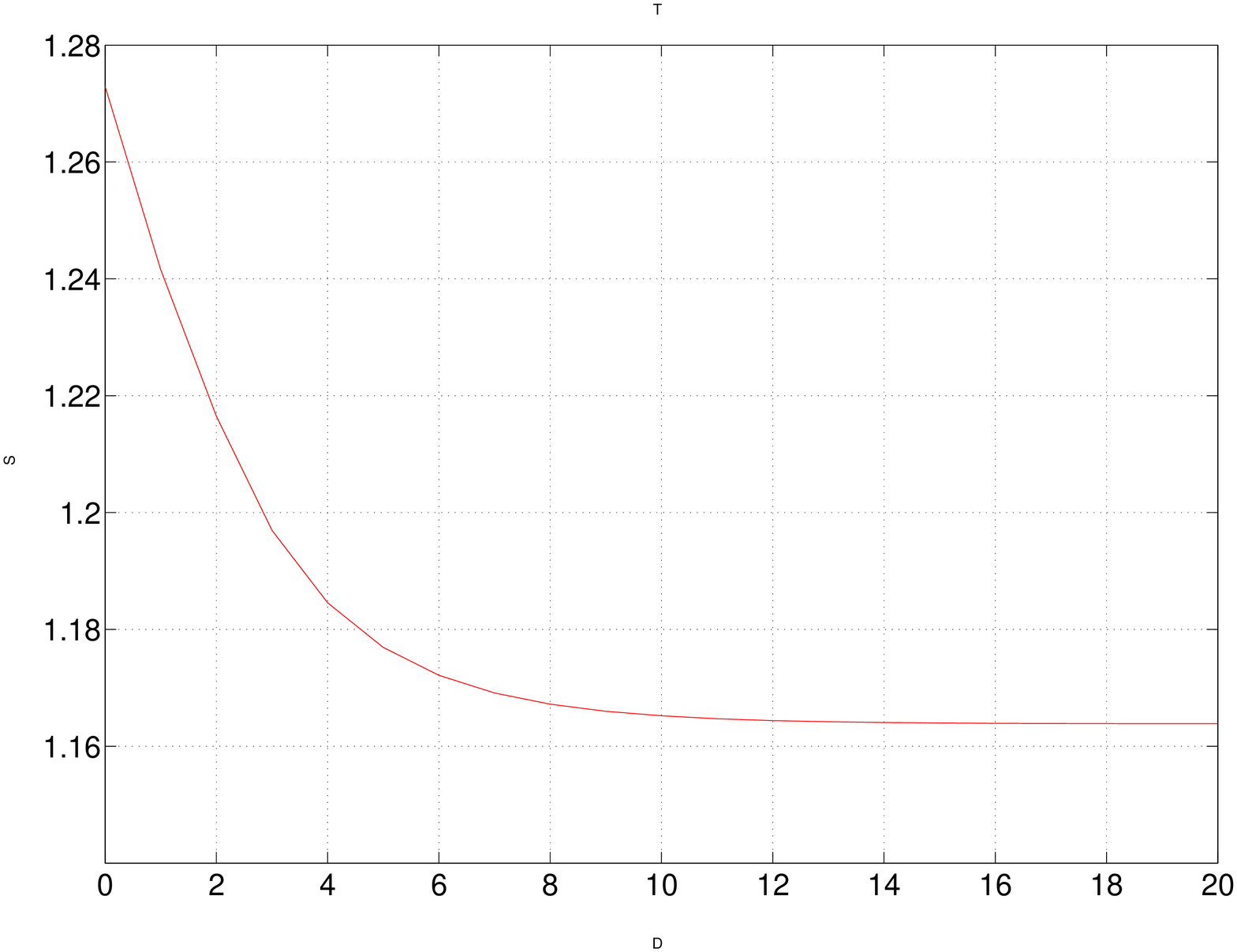}}\ \ \ \ \ \ \ \
\subfloat[]{\includegraphics[width=8.5cm]{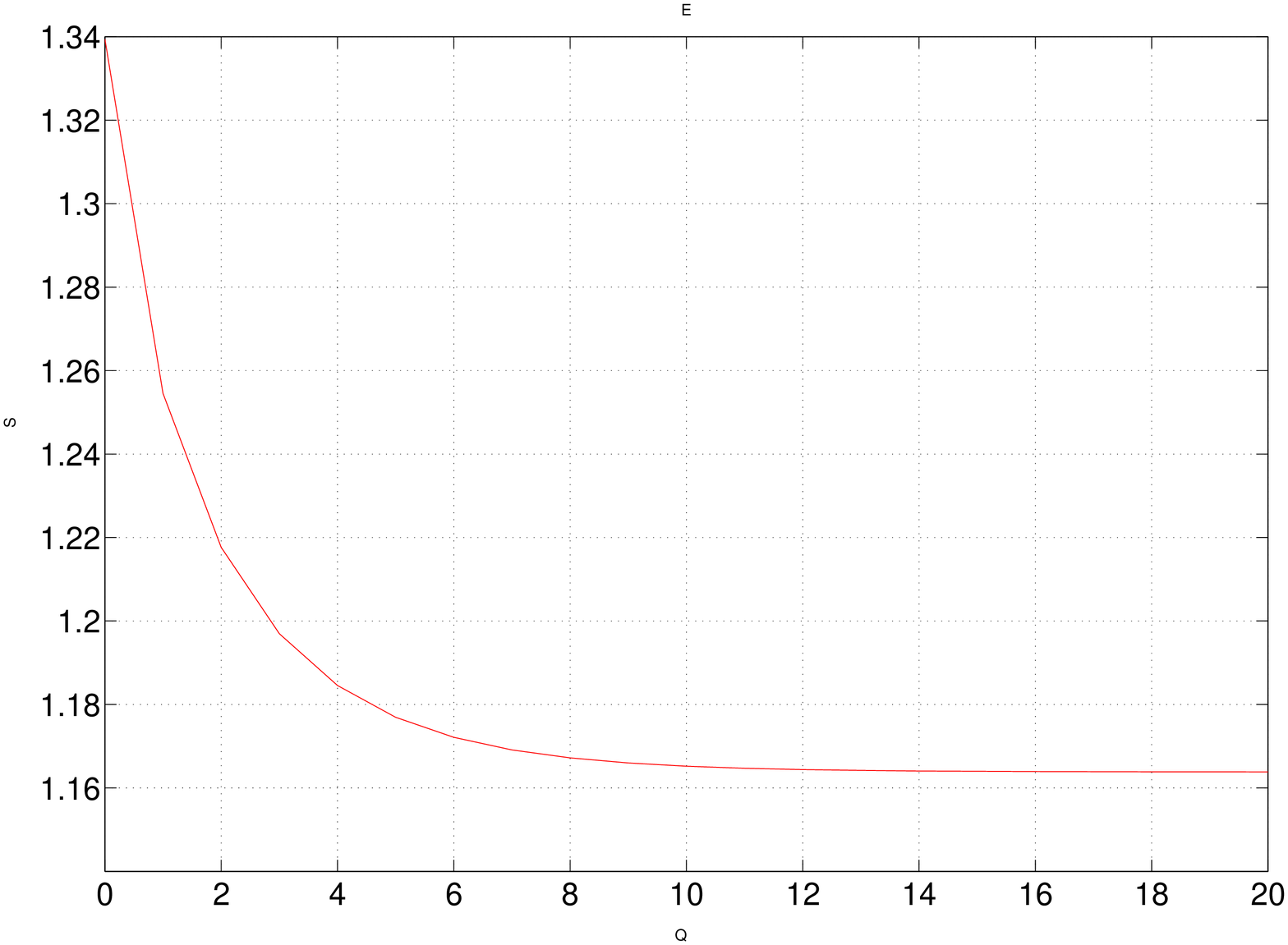}}
\subfloat[]{\includegraphics[width=8.5cm]{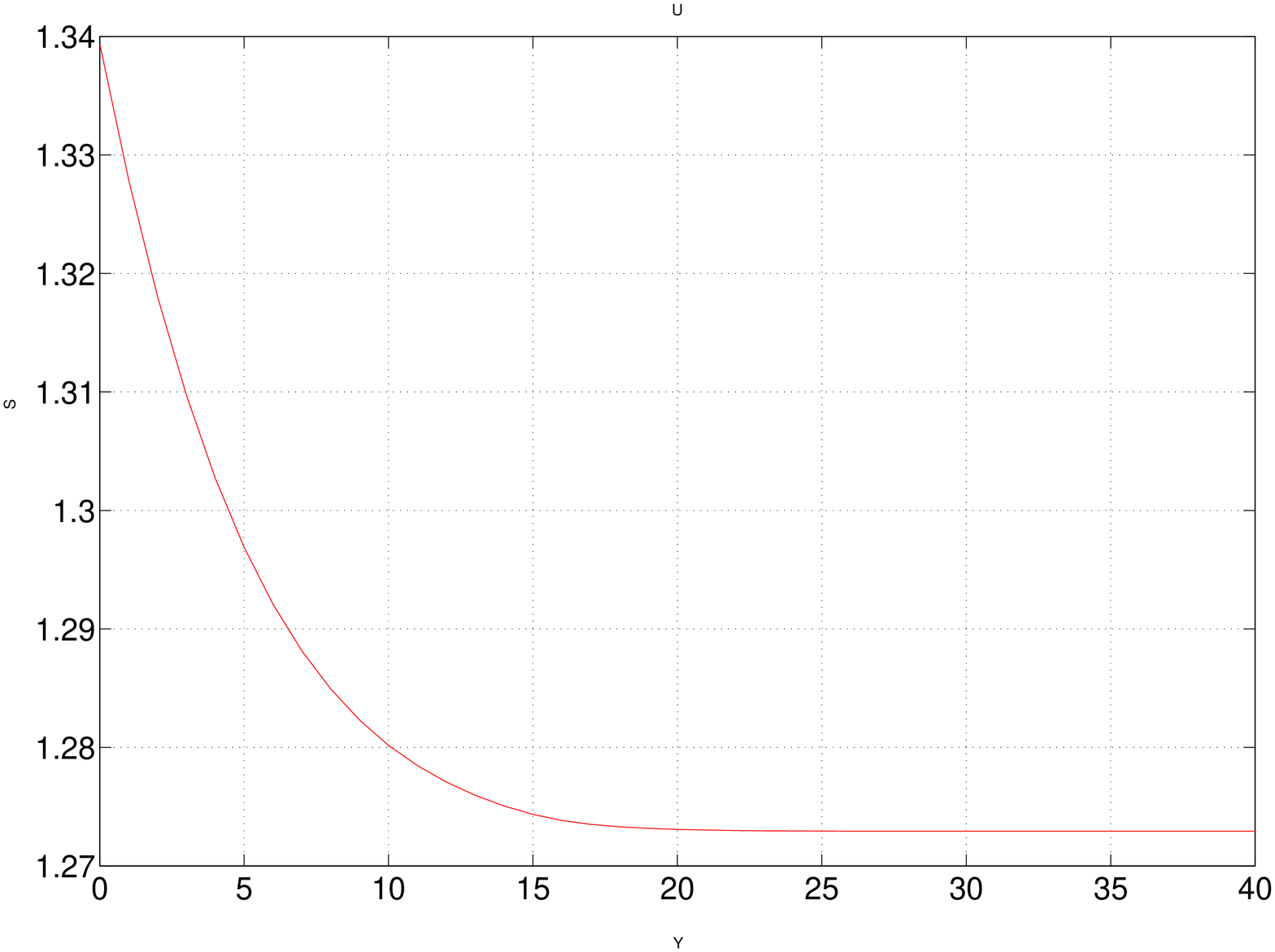}}
  \end{center}
\caption{The sum rate versus delay for the two state channel: (a) $d_2\leq d_1=\infty$,\ \ \  (b) $d_1=d_2$,\ \ \ (c) $0=d_2\leq d_1$.}
 \label{Sum rate vs delay}
}\end{figure}
\newpage
Perhaps it seems  that the  improvement  in the sum rate due to CSI is small, however, we should remember that when we encode large blocks,  this small improvement in the sum rate can be of importance. In addition, this improvement  in the sum rate due to CSI is for the specific example of two states AGN-MAC.
In Fig. \ref{power control policy versus delay}  we present the power control policy versus delay that achieves the maximum sum rates for the three cases.
\begin{figure}[h!]{
  \begin{center}
\psfrag{D}[][][0.8]{Delay $d_2$ (symbols)}  \psfrag{P}[][][0.8]{power} \psfrag{T}[][][0.8]{Power control policy vs.  delay $d_2$ ($d_1=\infty$)}\psfrag{B}[][][0.8]{$\mathcal{P}_2(\tilde{S}_2=B)$}\psfrag{G}[][][0.8]{$\mathcal{P}_2(\tilde{S}_2=G)$}
\psfrag{Q}[][][0.8]{Delay $d$ (symbols)}\psfrag{E}[][][0.8]{Power control policy vs.  delay  $d$ (symmetrical delay $d_1=d_2=d$)} \psfrag{X}[][][0.8]{$\mathcal{P}_1(\tilde{S}=B)=\mathcal{P}_2(\tilde{S}=B)$}
\psfrag{Z}[][][0.8]{$\mathcal{P}_1(\tilde{S}=G)=\mathcal{P}_2(\tilde{S}=G)$}
\psfrag{U}[][][0.8]{Delay $d_1$ (symbols)}\psfrag{M}[][][0.8]{Power control policy vs.  delay  $d_1$ (asymmetrical delay $d_2=0$)}
\psfrag{g}[][][0.6]{$\mathcal{P}_2(\tilde{S}_1=G,\tilde{S}_2=G)$}
\psfrag{j}[][][0.6]{$\mathcal{P}_2(\tilde{S}_1=B,\tilde{S}_2=G)$}
\psfrag{l}[][][0.6]{$\mathcal{P}_2(\tilde{S}_1=B,\tilde{S}_2=B)=\mathcal{P}_2(\tilde{S}_1=G,\tilde{S}_2=B)$}
\psfrag{h}[][][0.6]{$\mathcal{P}_1(\tilde{S}_1=G)$}
\psfrag{k}[][][0.6]{$\mathcal{P}_1(\tilde{S}_1=B)$}
\subfloat[]{\includegraphics[width=9cm]{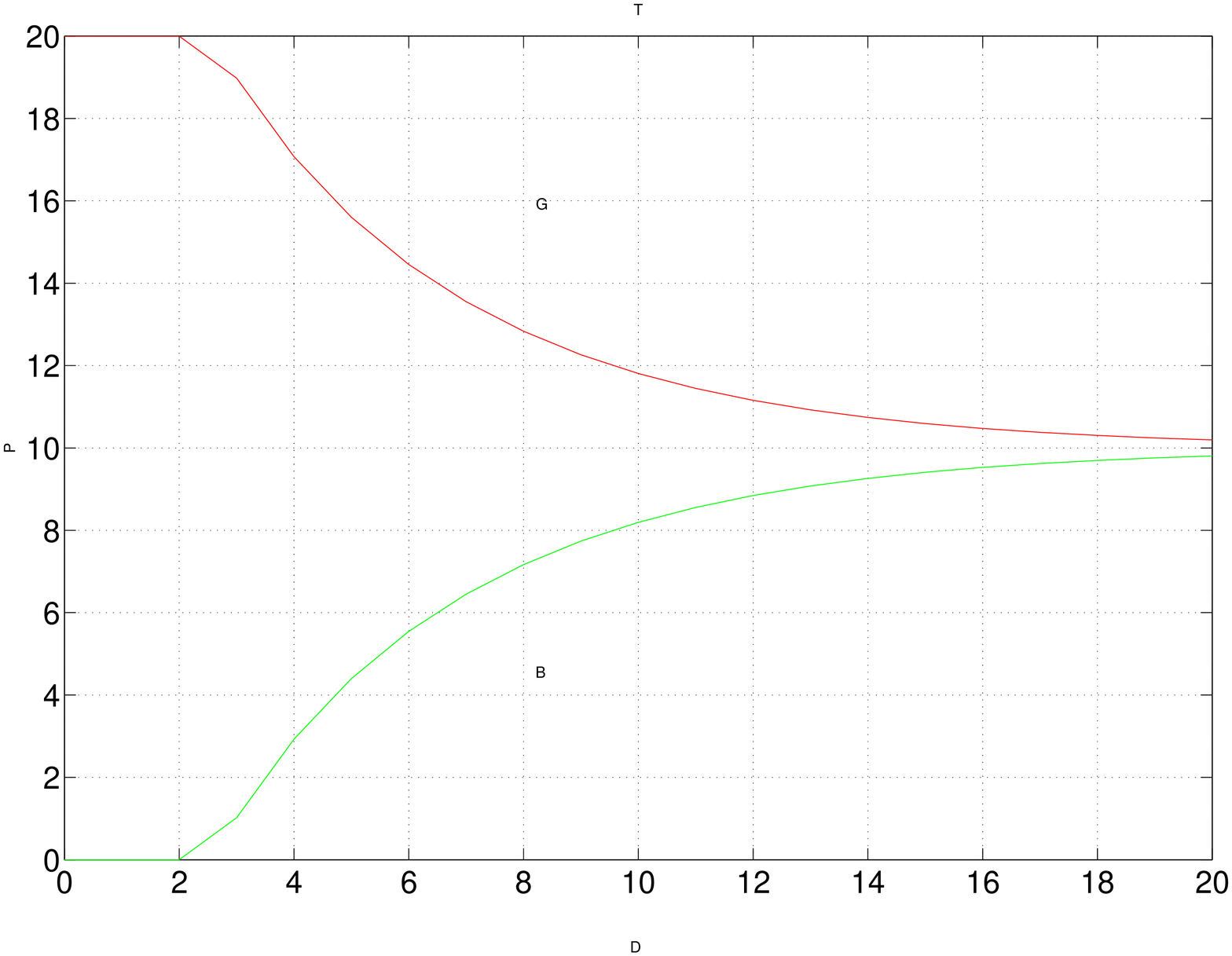}}\ \ \ \ \ \ \ \
\subfloat[]{\includegraphics[width=9cm]{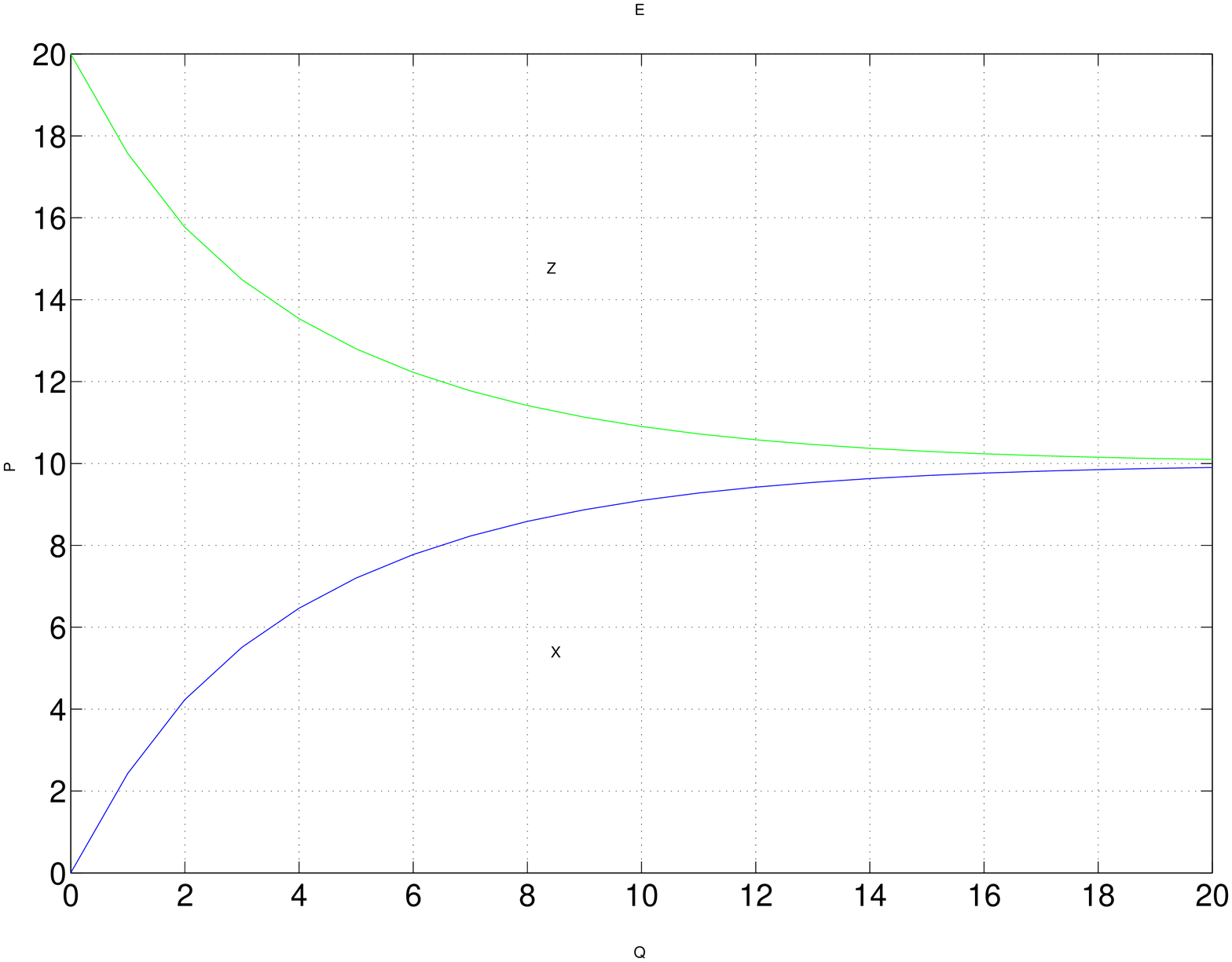}}
\subfloat[]{\includegraphics[width=9cm]{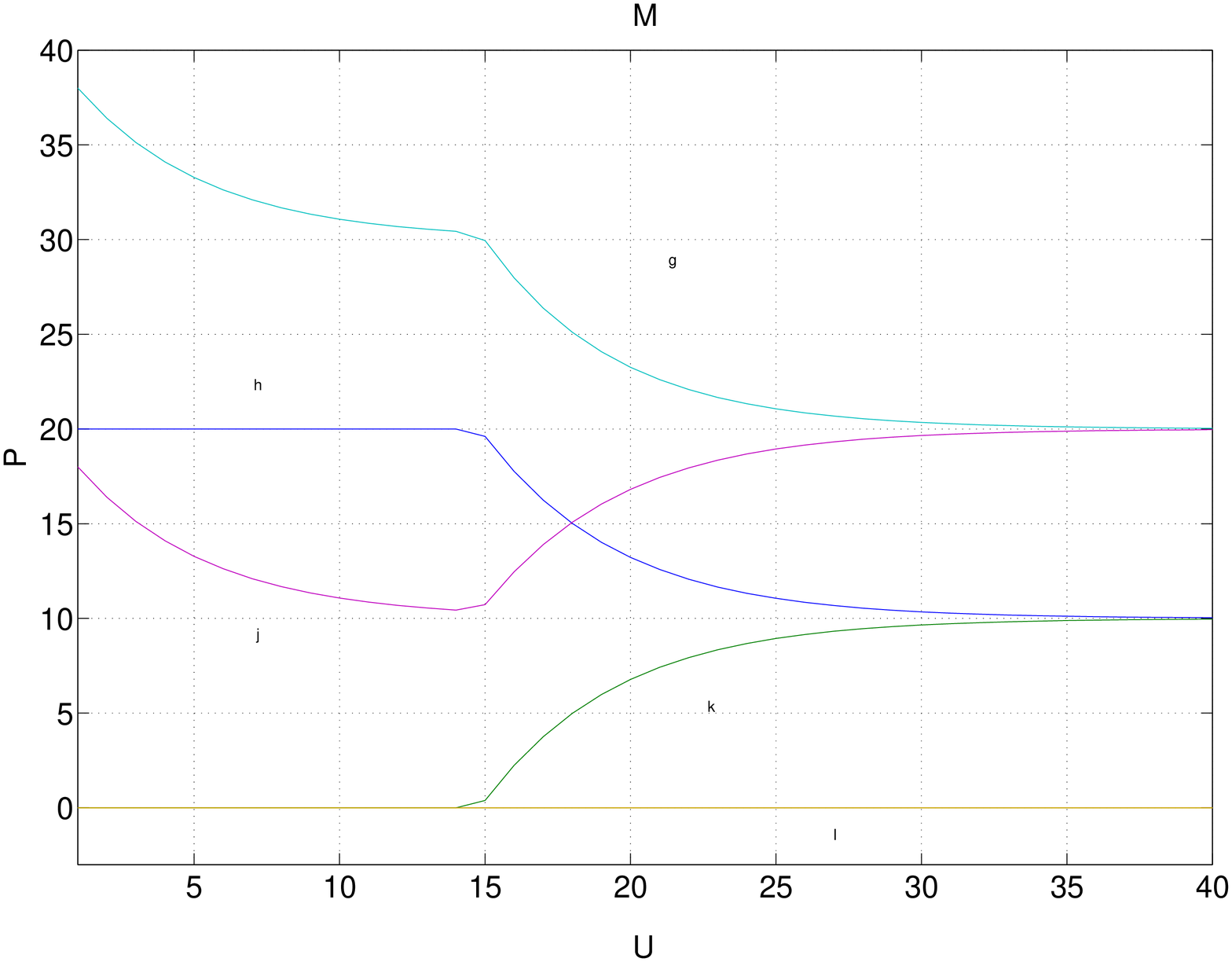}}
  \end{center}
\caption{The power control policy versus delay that achieves the maximum sum rate: (a) $d_2\leq d_1=\infty$,\ \ \  (b) $d_1=d_2$,\ \ \ (c) $0=d_2\leq d_1$.}
 \label{power control policy versus delay}
}\end{figure}
\newpage
Now, we present the capacity  rate region for the two states AGN-MAC in the asymmetrical case $d_1\geq d_2$ by solving numerically the following optimization problem for different values of $\alpha$,
\begin{eqnarray}
  &&\max_{R_1,R_2} \alpha R_1+R_2 ,\label{max_using_cvx}
\end{eqnarray}
subject to the constraints,
\begin{eqnarray}
  &&R_1\leq \frac{1}{2}\sum_{\tilde{s}_1}\pi(\tilde{s}_1)\sum_{\tilde{s}_2}K^{d_1-d_2}(\tilde{s}_2,\tilde{s}_1)\sum_{s}K^{d_2}(s,\tilde{s}_2) \log\left(1+\frac{\mathcal{P}_1(\tilde{s}_1)}{\sigma_{s}^2}\right),\\
   &&R_2\leq \frac{1}{2}\sum_{\tilde{s}_1}\pi(\tilde{s}_1)\sum_{\tilde{s}_2}K^{d_1-d_2}(\tilde{s}_2,\tilde{s}_1)\sum_{s}K^{d_2}(s,\tilde{s}_2) \log\left(1+\frac{\mathcal{P}_2(\tilde{s}_1,\tilde{s}_2)}{\sigma_{s}^2}\right),\\
   &&R_1+R_2\leq  \frac{1}{2}\sum_{\tilde{s}_1}\pi(\tilde{s}_1)\sum_{\tilde{s}_2}K^{d_1-d_2}(\tilde{s}_2,\tilde{s}_1)\sum_{s}K^{d_2}(s,\tilde{s}_2) \log\left(1+\frac{\mathcal{P}_1(\tilde{s}_1)+\mathcal{P}_2(\tilde{s}_1,\tilde{s}_2)}{\sigma_{s}^2}\right),\\
   &&\sum_{\tilde{s}_1}\pi(\tilde{s}_1)\mathcal{P}_1(\tilde{s}_1) \leq \mathcal{P}_1 ,\\
  &&\sum_{\tilde{s}_1}\pi(\tilde{s}_1)\sum_{\tilde{s}_2}P(\tilde{s}_2|\tilde{s}_1)\mathcal{P}_2(\tilde{s}_1,\tilde{s}_2) \leq \mathcal{P}_2.
\end{eqnarray}
In order to solve the optimization problem (\ref{max_using_cvx}) we used \texttt{CVX}, a package for specifying and solving convex optimization problems \cite{cvx}.
The capacity  rate region for $d_2=0$ and different values of $d_1$ are presented in Fig. \ref{Sum rate region asymmetrical}.
\begin{figure}[h!]{
  \begin{center}
\psfrag{A}[][][1]{ $R_2$ }  \psfrag{B}[][][1]{$R_1$}
\psfrag{T}[][][1]{ Capacity  rate region ($d_2=0$) }
\psfrag{Q}[][][1]{\ \ \ \ \ $d_1=0$}
\psfrag{W}[][][1]{\ \ \  \ \ \  \ \ \ \ \ \  \ $d_1=2\longrightarrow$}
\psfrag{E}[][][1]{ \ \ \  \ \ \  \ \ \  \ \ \  \ \ \   \ \ \ \ \  $d_1=4\longrightarrow$}
\psfrag{R}[][][1]{\ \ \ \ \  \ \ \  \ \ \  \ \ \  \ \ \  \ \ \ $d_1=100\longrightarrow$}
\includegraphics[width=12cm]{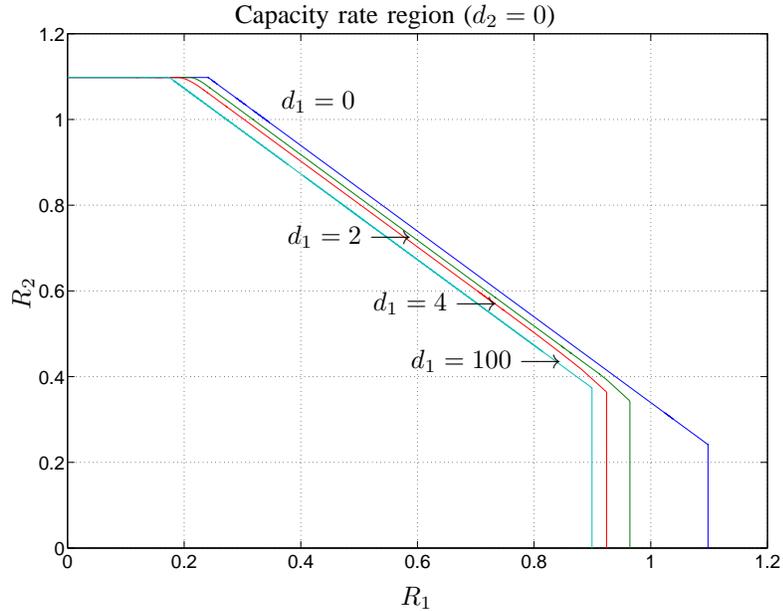}\ \ \ \ \ \ \ \
  \end{center}
\caption{Capacity  rate region for the two states AGN-MAC - asymmetrical case $d_2=0$.}
 \label{Sum rate region asymmetrical}
}\end{figure}\\
Similarly, we solve the optimization problem  for the  symmetrical case $d_1=d_2$, and for the case that transmitter $1$ does not have any CSI, i.e., $d_2<d_1=\infty$. The rate regions are illustrated in Fig. \ref{Sum rate region d_1=d_2}, and Fig. \ref{Sum rate region for the two states AGN-MAC-Transmitter $1$ doesn't have the CSI}, respectively.
\begin{figure}[h!]{
  \begin{center}
\psfrag{A}[][][1]{ $R_2$ }  \psfrag{B}[][][1]{$R_1$}
\psfrag{T}[][][1]{ Capacity rate region ($d_1=d_2$) }
\psfrag{Q}[][][1]{ \ \ \  $d=0$}
\psfrag{W}[][][1]{\ \ \  \ $d=2$}
\psfrag{E}[][][1]{\ \ \ \ $d=4\longrightarrow$}
\psfrag{R}[][][1]{\ \ \ \ \ \ $d=100\rightarrow$}
\includegraphics[width=12cm]{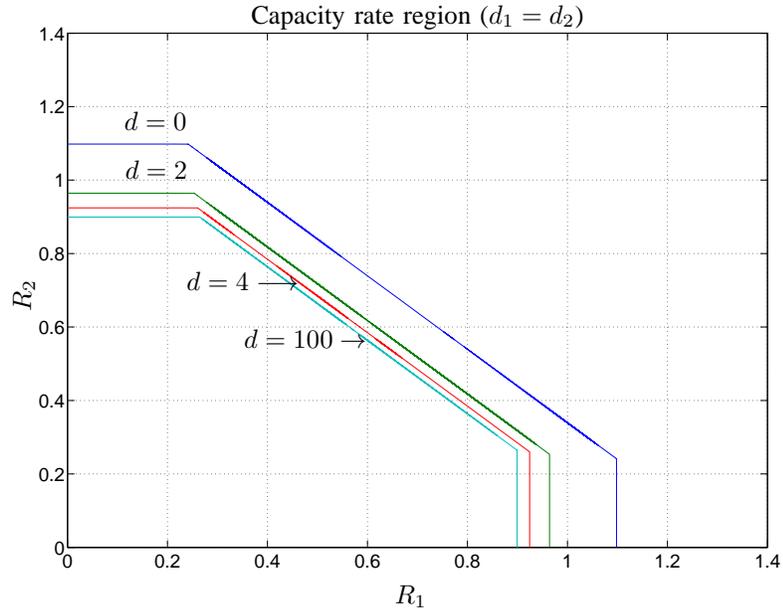}\ \ \ \ \ \ \ \
  \end{center}
\caption{Capacity  rate region for the two states AGN-MAC - symmetrical case $d=d_1=d_2$.}
 \label{Sum rate region d_1=d_2}
}\end{figure}\\
\begin{figure}[h!]{
  \begin{center}
\psfrag{A}[][][1]{ $R_2$ }  \psfrag{B}[][][1]{$R_1$}
\psfrag{T}[][][1]{ Capacity  rate region ($d_2\leq d_1=\infty$) }
\psfrag{Q}[][][1]{\ \ \ \ \ $d_2=0$}
\psfrag{W}[][][1]{\ \ \ \ \  $d_2=2$}
\psfrag{E}[][][1]{\ \ \ \ \ \ \ \ \ \ $d_2=4\longrightarrow$}
\psfrag{R}[][][1]{\ \ \ \ \ \ \ \ \ \ \ $d_2=100\rightarrow$}
\includegraphics[width=12cm]{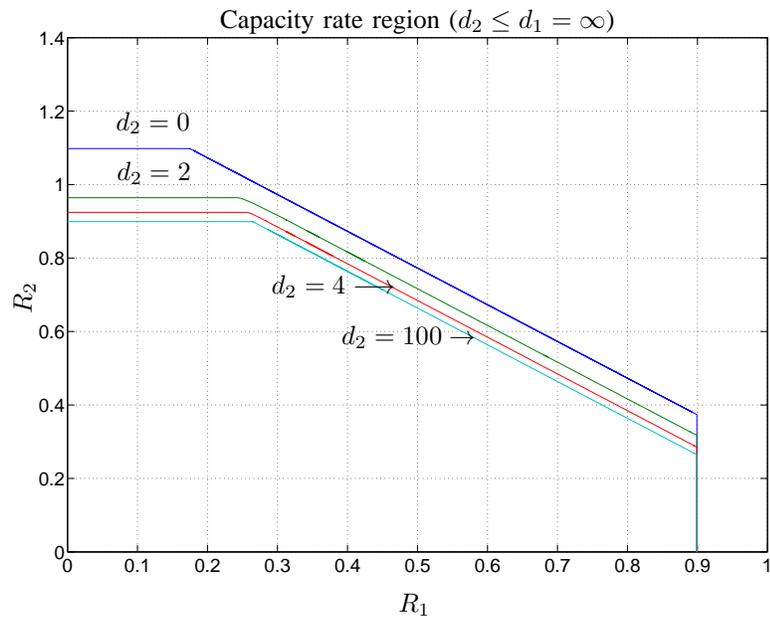}\ \ \ \ \ \ \ \
  \end{center}
\caption{Capacity  rate region for the two states AGN-MAC - Transmitter $1$ does not have the CSI $d_2\leq d_1=\infty$.}
 \label{Sum rate region for the two states AGN-MAC-Transmitter $1$ doesn't have the CSI}
}\end{figure}\\
\subsection{Capacity  Region for a Finite State Multiple-Access Fading Channel}\label{FADING_CHANNEL}
We  apply Theorem \ref{Capacity region-  MAC with delayed CSI feedback t1} to compute the capacity region of a power constrained FS Multiple-Access fading  channel, and illustrate the effect of the delayed CSI on the capacity region.
Consider the discrete-time multiple-access Gaussian channel,
\begin{eqnarray}
  Y_i = h_1(s_i)X_{1,i} + h_2(s_i)X_{2,i} + N_{S_i},
\end{eqnarray}
where $X_{1,i},X_{2,i}$ are the transmitted waveform, and $h_1(s_i),h_2(s_i)$ are the fading process of the users. The terms $h_1(s_i),h_2(s_i)$ are deterministic functions of $s_i$. The noise
$N_{S_i}$ is a zero-mean Gaussian random variable with variance depending on the state of the channel at time $i$.
Furthermore,
the users are subject to the average transmitter power constraints of $\mathcal{P}_1$, and $\mathcal{P}_2$.
The state process is assumed to be Markov with steady state distribution $\pi(s)$ and one step transition matrix $K$,
as described in Section \ref{s_preliminary}.
The FS Multiple-Access fading channel is illustrated in Fig. \ref{fading_channel}.
\begin{figure}[h!]{
  \begin{center}
\psfrag{a}[][][1]{ $X_1$ }  \psfrag{b}[][][1]{$X_2$}
\psfrag{c}[][][1]{\ \ \ \ $h_1(s)$ }  \psfrag{d}[][][1]{\ \ \ \ $h_2(s)$}
\psfrag{e}[][][1]{$Y$ }
\psfrag{f}[][][1]{\ \ $N_s$ }
\includegraphics[width=6cm]{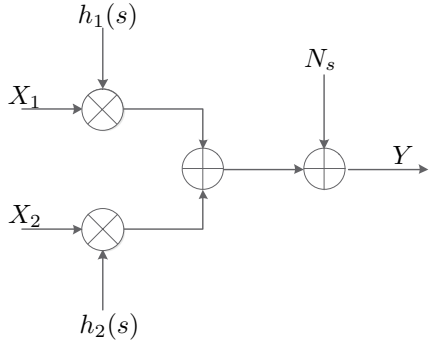}\ \ \ \ \ \ \ \
  \end{center}
\caption{The fading channel.}
 \label{fading_channel}
}\end{figure}
We apply the capacity region formula to explicitly determine the capacity region of the multiple-access Gaussian fading channel with transmitters power constraints $\mathcal{P}_1$ and $\mathcal{P}_2$.
 In a similar way to the FSM Additive Gaussian MAC, it can be shown that the capacity achieving distributions are
 $X_1(\tilde{s}_1,u)$  zero-mean Gaussian with variance $\mathcal{P}_1(\tilde{s}_1)$, and $X_2(\tilde{s}_1,\tilde{s}_1,u)$ zero-mean Gaussian with variance $\mathcal{P}_2(\tilde{s}_1,\tilde{s}_2)$, both independent of $N_s$ and independent of each other. We derive the following optimization problem,
 \begin{eqnarray}
   R_1&=&\max_{\mathcal{P}_1(\tilde{s}_1)} \frac{1}{2}\sum_{\tilde{s}_1}\pi(\tilde{s}_1)\sum_{\tilde{s}_2}K^{d_1-d_2}(\tilde{s}_2,\tilde{s}_1)\sum_{s}K^{d_2}(s,\tilde{s}_2) \log\left(1+\frac{h_1(s)^{2}\mathcal{P}_1(\tilde{s}_1)}{\sigma_{s}^2}\right),\label{R1maxfading}\\
   R_2&=&\max_{\mathcal{P}_2(\tilde{s}_1,\tilde{s}_2)} \frac{1}{2}\sum_{\tilde{s}_1}\pi(\tilde{s}_1)\sum_{\tilde{s}_2}K^{d_1-d_2}(\tilde{s}_2,\tilde{s}_1)\sum_{s}K^{d_2}(s,\tilde{s}_2) \log\left(1+\frac{h_2(s)^2\mathcal{P}_2(\tilde{s}_1,\tilde{s}_2)}{\sigma_{s}^2}\right),\label{R2maxfading}\\
   R_1+R_2 &=&\max_{\mathcal{P}_1(\tilde{s}_1),\mathcal{P}_2(\tilde{s}_1,\tilde{s}_2)} \frac{1}{2}\sum_{\tilde{s}_1}\pi(\tilde{s}_1)\sum_{\tilde{s}_2}K^{d_1-d_2}(\tilde{s}_2,\tilde{s}_1)\sum_{s}K^{d_2}(s,\tilde{s}_2) \nonumber\\&&\times\log\left(1+\frac{h_1(s)^{2}\mathcal{P}_1(\tilde{s}_1)+h_2(s)^{2}\mathcal{P}_2(\tilde{s}_1,\tilde{s}_2)}{\sigma_{s}^2}\right),\label{R1R2maxasfading}
\end{eqnarray}
subject to the power constraints,
\begin{eqnarray}
  &&\sum_{\tilde{s}_1}\pi(\tilde{s}_1)\mathcal{P}_1(\tilde{s}_1) \leq \mathcal{P}_1 ,\\
  &&\sum_{\tilde{s}_1}\pi(\tilde{s}_1)\sum_{\tilde{s}_2}P(\tilde{s}_2|\tilde{s}_1)\mathcal{P}_2(\tilde{s}_1,\tilde{s}_2) \leq \mathcal{P}_2.
  \end{eqnarray}
In the same way, we can derive the optimization  problem for
the  symmetrical case $d_1=d_2$, and for the case that transmitter $1$ does not have any CSI, i.e., $d_2<d_1=\infty$.
Let us solve the optimization  problems for the following FSM multiple-access fading channel examples:
\subsubsection{Example 1 {\it{(AGN switch channel)}}}
Consider the discrete-time multiple-access Gaussian two state switch channel as described in Fig. \ref{fading_channel_switch}.
We solve the optimization problem: $\max (\alpha R_1+R_2)$, for different values of $\alpha$ in the same way we did in the FS additive Gaussian noise (AGN) MAC example. In Fig. \ref{Sum rate region asymmetrical_fading}, \ref{Sum rate region d_1=d_2_fading}, and \ref{Sum rate region for the two states_switch-Transmitter $1$ doesn't have the CSI} we present the capacity rate region for $\mathcal{P}_1=10$, $\mathcal{P}_2=10$, $\sigma_G^2=1$, $\sigma_B^2=10$, $g=0.1$, $b=0.1$, $h_1(G)=1$, $h_1(B)=0$, $h_2(G)=0,h_2(B)=1$, in the following cases: asymmetrical, symmetrical, and the case that transmitter $1$ does not have any CSI.
\begin{figure}[h!]{
  \begin{center}
\psfrag{a}[][][1]{ $X_1$ }  \psfrag{b}[][][1]{$X_2$}
\psfrag{c}[][][1]{\ \ \ \ $N_G$ }  \psfrag{d}[][][1]{\ \ \ \ $N_B$}
\psfrag{e}[][][1]{$Y$ }
\psfrag{G}[][][1]{$S=G$}
\psfrag{B}[][][1]{$S=B$}
\includegraphics[width=6cm]{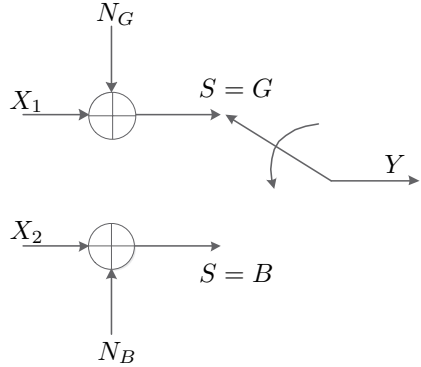}\ \ \ \ \ \ \ \
  \end{center}
\caption{The channel behaves like a switch, at any given time $i$ the channel is in one of two possible states $G$ or $B$, where $\sigma_B^2>\sigma_G^2$. The state process is illustrated in Fig. \ref{Two-state AGN channel}.}
 \label{fading_channel_switch}
}\end{figure}


\begin{figure}[ht!]{
  \begin{center}
\psfrag{A}[][][1]{ $R_2$ }  \psfrag{B}[][][1]{$R_1$}
\psfrag{T}[][][1]{ Capacity  rate region ($d_2=0$) }
\psfrag{Q}[][][1]{\ \ \ \ \ \  \ \ $d_1=0$}
\psfrag{W}[][][1]{\ \ \ \ \ \ $d_1=2$}
\psfrag{E}[][][1]{\ \  $d_1=4\longrightarrow$}
\psfrag{R}[][][1]{  \ \ $d_1=100\rightarrow$}
\includegraphics[width=12cm]{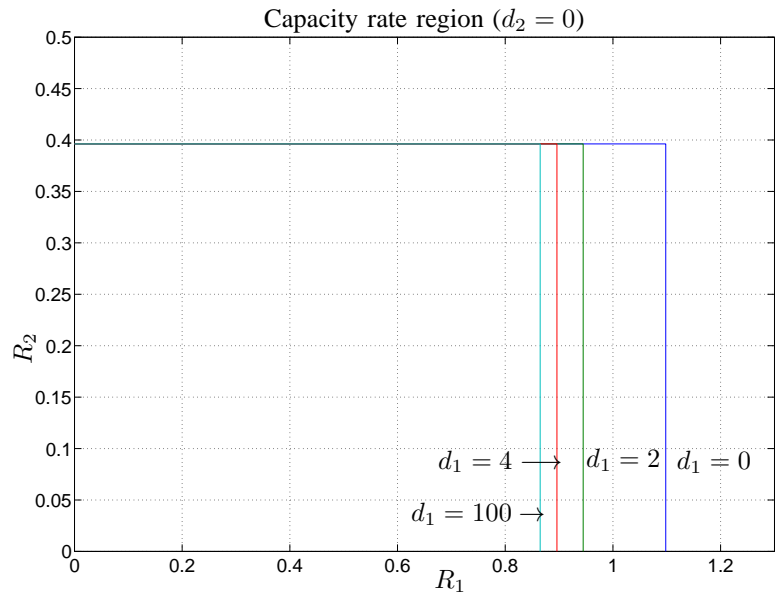}\ \ \ \ \ \ \ \
  \end{center}
\caption{Capacity  rate region for the two states switch channel - asymmetrical case $d_2=0$.}
 \label{Sum rate region asymmetrical_fading}
}\end{figure}
\begin{figure}[ht!]{
  \begin{center}
\psfrag{A}[][][1]{ $R_2$ }  \psfrag{B}[][][1]{$R_1$}
\psfrag{T}[][][1]{ Capacity  rate region ($d_1=d_2$) }
\psfrag{Q}[][][1]{ \ \ \  $d=0$}
\psfrag{W}[][][1]{\ \ \  $d=2$}
\psfrag{E}[][][1]{\ \ \ \ \ \ \ \ $d=4\longrightarrow$}
\psfrag{R}[][][1]{ \ \ $d=100\rightarrow$}
\includegraphics[width=12cm]{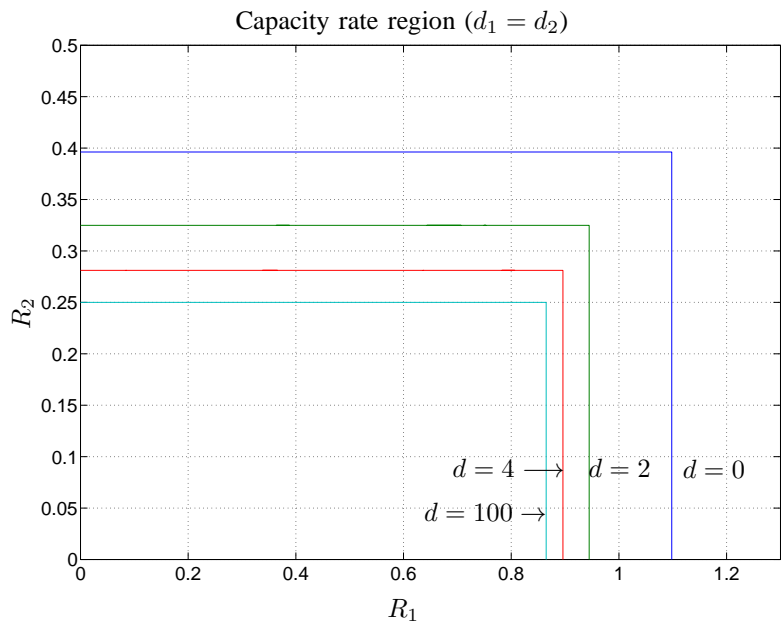}\ \ \ \ \ \ \ \
  \end{center}
\caption{Capacity  rate region for the two states switch channel - symmetrical case $d=d_1=d_2$.}
 \label{Sum rate region d_1=d_2_fading}
}\end{figure}
\newpage
\begin{figure}[ht!]{
  \begin{center}
\psfrag{A}[][][1]{ $R_2$ }  \psfrag{B}[][][1]{$R_1$}
\psfrag{T}[][][1]{ Capacity  rate region ($d_2\leq d_1=\infty$) }
\psfrag{Q}[][][1]{\ \ \ \ \ $d_2=0$}
\psfrag{W}[][][1]{\ \ \ \ \  $d_2=2$}
\psfrag{E}[][][1]{\ \ \ \ \ $d_2=4$}
\psfrag{R}[][][1]{\ \ \ \ \ \ \ $d_2=100$}
\includegraphics[width=12cm]{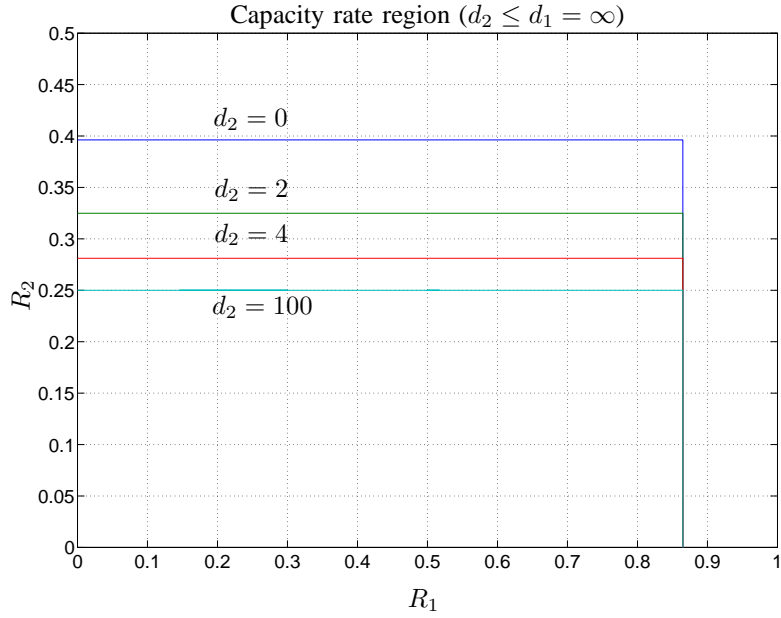}\ \ \ \ \ \ \ \
  \end{center}
\caption{Capacity  rate region for the two states switch channel - Transmitter $1$ does not have the CSI $d_2\leq d_1=\infty$.}
 \label{Sum rate region for the two states_switch-Transmitter $1$ doesn't have the CSI}
}\end{figure}
As one can see from Fig. \ref{Sum rate region asymmetrical_fading}, \ref{Sum rate region d_1=d_2_fading}, and \ref{Sum rate region for the two states_switch-Transmitter $1$ doesn't have the CSI} the capacity rate region shape indicates that the users do not interrupt each other, so each  of them can transmit at its own maximal rate independently of the other user. This makes perfect sense, since the transmission of each one of them is dependent only on the switch and not on the other's transmission.
\subsubsection{Example 2 {\it{(Multiple-Access fading channel)}}}
Consider the power constrained FS Multiple-Access fading  channel as illustrated in Fig.\ref{fading_channel} with only
two states: $S=1$, $S=2$. The state process is Markov and illustrated in Fig. \ref{Two-state AGN channel}, with a slight change, instead of denoting the states "good" and "bad" we use $S=1$, $S=2$.
We solve the optimization problem: $\max (\alpha R_1+R_2)$, for different values of $\alpha$ in the same way we did before. In Fig. \ref{Sum rate region asymmetrical_fading_notsw}, \ref{Sum rate region d_1=d_2_fading_notsw}, and \ref{Sum rate region for the two states_not_switch-Transmitter $1$ doesn't have the CSI} we present the capacity rate region for $\mathcal{P}_1=10$, $\mathcal{P}_2=10$, $\sigma_{s=1}^2=\sigma_{s=2}^2=1$, $g=0.1$, $b=0.1$, $h_1(s=1)=1$, $h_1(s=2)=0.5$, $h_2(s=1)=0.5$, $h_2(s=2)=1$.
\begin{figure}[ht!]{
  \begin{center}
\psfrag{A}[][][1]{ $R_2$ }  \psfrag{B}[][][1]{$R_1$}
\psfrag{T}[][][1]{ Capacity  rate region ($d_2=0$) }
\psfrag{Q}[][][1]{\ \ \ \  $d_1=0$}
\psfrag{W}[][][1]{\ \ \ \ \ \ \ \ $d_1=2\longrightarrow$}
\psfrag{E}[][][1]{\ \ \ \ \ \    $d_1=4\longrightarrow$}
\psfrag{R}[][][1]{\ \ \ \  $d_1=100\rightarrow$}
\includegraphics[width=12cm]{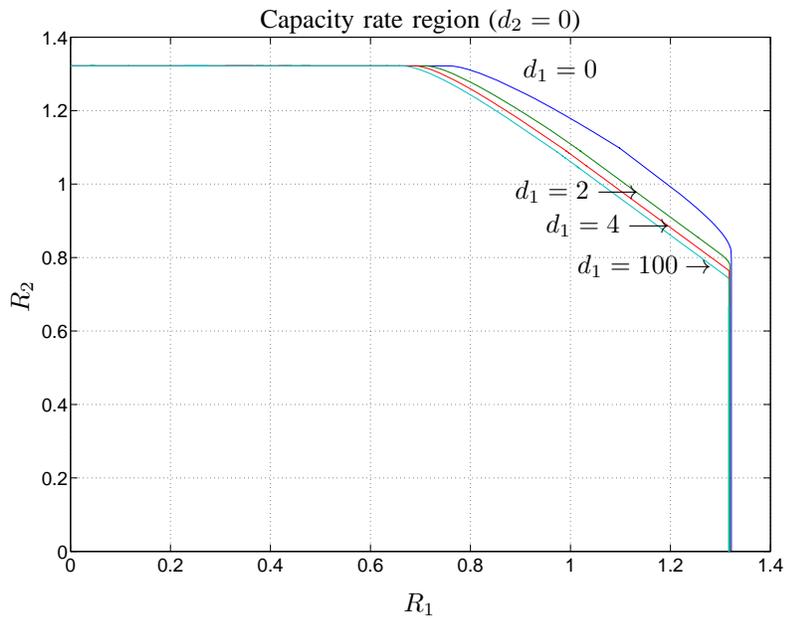}\ \ \ \ \ \ \ \
  \end{center}
\caption{Capacity  rate region for the two states fading channel - asymmetrical case $d_2=0$.}
 \label{Sum rate region asymmetrical_fading_notsw}
}\end{figure}
\begin{figure}[ht!]{
  \begin{center}
\psfrag{A}[][][1]{ $R_2$ }  \psfrag{B}[][][1]{$R_1$}
\psfrag{T}[][][1]{ Capacity  rate region ($d_1=d_2$) }
\psfrag{Q}[][][1]{ \ \  $d=0$}
\psfrag{W}[][][1]{\ \ \ \ $d=2$}
\psfrag{E}[][][1]{\ \ \ \ \ \ \ \ \ $d=4\longrightarrow$}
\psfrag{R}[][][1]{\ \ \  $d=100\rightarrow$}
\includegraphics[width=12cm]{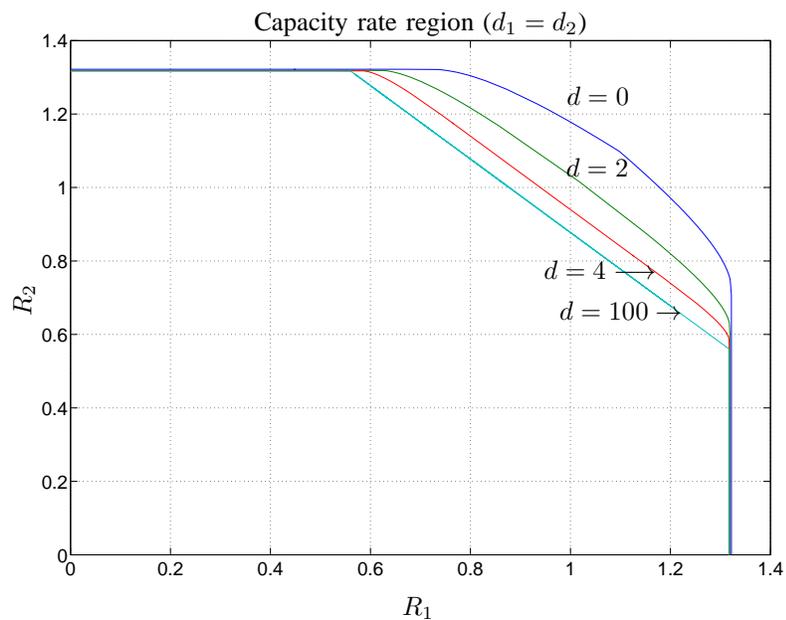}\ \ \ \ \ \ \ \
  \end{center}
\caption{Capacity  rate region for the two states fading channel - symmetrical case $d=d_1=d_2$.}
 \label{Sum rate region d_1=d_2_fading_notsw}
}\end{figure}

\begin{figure}[ht!]{
  \begin{center}
\psfrag{A}[][][1]{ $R_2$ }  \psfrag{B}[][][1]{$R_1$}
\psfrag{T}[][][1]{ Capacity  rate region ($d_2\leq d_1=\infty$) }
\psfrag{Q}[][][1]{\ \ \ \ \ \ \  $d_2=0$}
\psfrag{W}[][][1]{ \ \ $d_2=2\longrightarrow$}
\psfrag{E}[][][1]{\ \ \ \ \ \ \ \    $d_2=4\longrightarrow$}
\psfrag{R}[][][1]{ \ \ \  \ \ \ \ \ \   $d_2=100\rightarrow$}
\includegraphics[width=12cm]{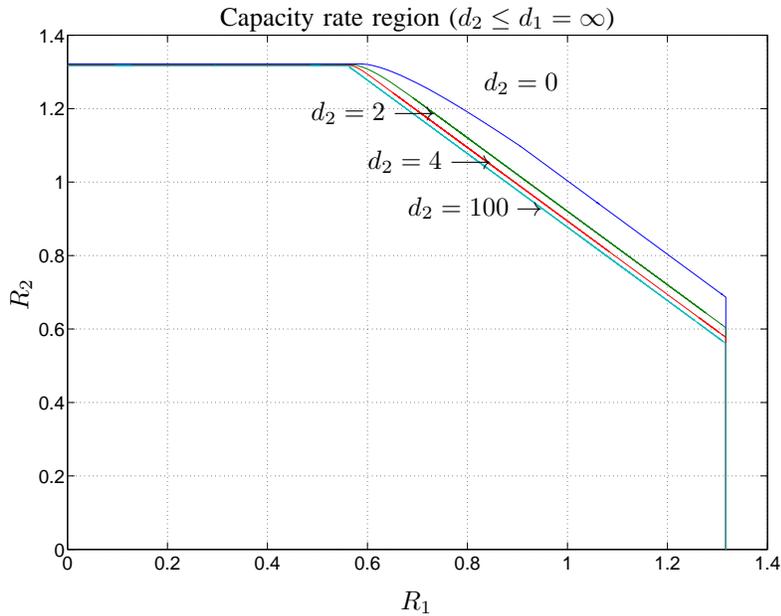}\ \ \ \ \ \ \ \
  \end{center}
\caption{Capacity  rate region for the two states fading channel - Transmitter $1$ does not have the CSI $d_2\leq d_1=\infty$.}
 \label{Sum rate region for the two states_not_switch-Transmitter $1$ doesn't have the CSI}
}\end{figure}
\newpage

\section{SUMMARY}\label{summery}
The requirement for high rates multi-user communications systems is constantly increasing,
so it becomes essential to achieve capacity
by deriving the benefit from the channel structure. Motivated by
this we studied the problem of finite-state MAC,
where the channel state is a Markov process, the transmitters have access to
delayed state information, and channel
state information is available at the receiver. The delays
of the channel state information is assumed to be asymmetric
at the transmitters. We obtained a
computable characterization of the capacity region for this channel.
 We provide the upper bound on the capacity region  and  the proof of the achievability, which is based on multiplexing coding. In addition, we provide alternative proof for the capacity region.
 The alternative proof is based on a multi-letter expression for the capacity region of FS-MAC with time-invariant feedback.
 Then we apply the result to derive power control strategies to maximize the capacity region for finite-state additive Gaussian MAC, and for the multiple-access fading channel.
 The results and the insight in this paper are an intermediate step toward understanding  network communication with delayed state information.
\appendices
\section{CARDINALITY BOUND OF THE AUXILIARY RANDOM VARIABLE $U$}\label{cardinality_proof}
Let us prove now the cardinality bound for Theorem 1, which is derived directly from the Fenchel - Eggleston - Carath\'eodry theory \cite{cardinality}. Let us denote the set ${\cal Z}$ to be ${\cal Z} \triangleq {\cal X}_1 \times {\cal X}_2 \times  {\cal S} \times \tilde{\cal S}_1 \times \tilde{\cal S}_2$, let ${\cal P}({\cal Z})$ be the set of PMFs on ${\cal Z}$, and let ${\cal P}({\cal Z} | {\cal U}) \subseteq {\cal P}({\cal Z})$ be a collection of PMFs $p(z|u)$ on ${\cal Z}$ indexed by $u \in \cal U$. Let $g_j,\ j=1, \dots, k$ be continues functions on ${\cal P}({\cal Z}|{\cal U})$. Then, for any $U \sim F_U(u)$, there exists a finite random variables $U' \sim p(u')$ taking at most $k$ values in $\cal U$ such that
\begin{eqnarray}
    \mathbb{E} \Big[g_j(p_{Z|U}(z|U))\Big] &=& \int_{\cal U} g_j(p_{Z|U}(z|u)){\rm d}F(u)\\
    &=& \sum_{u'} g_j(p_{Z|U}(z|u'))p(u').
\end{eqnarray}
Let us denote,
\begin{eqnarray}
    g_1 \big( p(z|u)\big) &=& I(X_1;Y|X_2, S, \tilde{S}_1, \tilde{S}_2, U=u)\\
    g_2 \big( p(z|u)\big) &=& I(X_2;Y|X_1, S, \tilde{S}_1, \tilde{S}_2, U=u)\\
    g_3 \big( p(z|u)\big) &=& I(X_1,X_2;Y|S, \tilde{S}_1, \tilde{S}_2, U=u),
\end{eqnarray}
then, by using the given technique, we can see that $|{\cal U}| \leq 3$.
By utilizing the same technique, and similar considerations, we can bound the cardinality of the auxiliary variable in Theorem 2 to be $|{\cal U}| \leq 3$ and the cardinality of the auxiliary variable in Theorem 3 to be  $|{\cal Q}| \leq 3$.
\section{PROOF OF THEOREM \ref{Capacity region-  MAC with delayed CSI feedback only to one encoder}}\label{PROOF OF THEOREM3}
The proof of Theorem  \ref{Capacity region-  MAC with delayed CSI feedback only to one encoder} is similar to the case where the CSI is available at  the decoder and asymmetrical delayed  CSI is available at the encoders with delays $d_1$ and $d_2$ ($d_1\geq d_2$), only  now $d_1\rightarrow\infty$. We give here the proof of the converse, and only a brief outline of the  achievability proof. Since only encoder $2$ has the CSI we denote $d=d_2$ and $\tilde{S}=\tilde{S}_2$.
\subsection{Converse Theorem  \ref{Capacity region-  MAC with delayed CSI feedback only to one encoder}}
Given an  achievable rate $(R_{1},R_{2})$ we need to show that there exists joint distribution of the form
$P(s,\tilde{s})P(q)P(x_{1}|q)P(x_{2}|\tilde{s},q)P(y|x_{1},x_{2},s)$ such that,
\begin{eqnarray}
R_{1}<I(X_{1};Y|X_{2},S,\tilde{S},Q),\nonumber \\
R_{2}<I(X_{2};Y|X_{1},S,\tilde{S},Q),\nonumber \\
R_{1}+R_{2}<I(X_{1},X_{2};Y|S,\tilde{S},Q),\nonumber
\end{eqnarray}
where $Q$ is an random variable with a cardinality bound  $|{\cal Q}| \leq 3$. The proof of the cardinality bound is similar to the proof in Appendix \ref{cardinality_proof}.
Since $(R_1,R_2)$ is an achievable pair-rate, there exists a code $(n,2^{nR_1},2^{nR_2},d)$ with a probability of error  $P_{e}^{(n)}$ arbitrarily small. By Fano's inequality,
\begin{eqnarray}
H(M_1,M_2|Y^n,S^n)\leq n(R_1+R_2)P_{e}^{(n)}+H(P_{e}^{(n)})\triangleq n\varepsilon_n,
\end{eqnarray}
and it is clear that $\varepsilon_n\rightarrow0$ as $P_{e}^{(n)}\rightarrow\infty$. Then we have
\begin{eqnarray}
H(M_1|Y^n,S^n)\leq H(M_1,M_2|Y^n,S^n)\leq \varepsilon_n, \\
H(M_2|Y^n,S^n)\leq H(M_1,M_2|Y^n,S^n)\leq \varepsilon_n.
\end{eqnarray}

We can now bound the rate $R_1$ as
  \begin{eqnarray}
  nR_{1}&=& H(M_1)  \nonumber\\
        &=& H(M_1)+H(M_1|Y^{n},S^{n})-H(M_1|Y^{n},S^{n})                     \nonumber\\
        &\stackrel{(a)}\leq& I(M_1;Y^n,S^{n})+n \varepsilon_n                    \nonumber\\
        &\stackrel{(b)}=& I(M_1;Y^n|S^{n})+I(M_1;S^{n})+n \varepsilon_n     \nonumber\\
        &\stackrel{(c)}=& I(M_1;Y^n|S^{n})+n \varepsilon_n                       \nonumber\\
        &\stackrel{(d)}=&  I(X_{1}^{n};Y^n|S^{n})+n \varepsilon_n             \nonumber\\
        &=& H(X_{1}^{n}|S^{n})-H(X_{1}^{n}|Y^n,S^{n})+n \varepsilon_n        \nonumber\\
        &\stackrel{(e)}=& H(X_{1}^{n}|X_{2}^{n},S^{n})-H(X_{1}^{n}|Y^n,S^{n})+n \varepsilon_n       \nonumber\\
        &\stackrel{(f)}\leq& H(X_{1}^{n}|X_{2}^{n},S^{n})-H(X_{1}^{n}|Y^n,X_{2}^{n},S^{n})+n \varepsilon_n  \nonumber\\
        &=& I(X_{1}^{n};Y^n|X_{2}^{n},S^{n})+n \varepsilon_n        \nonumber\\
        &=& H(Y^{n}|X_2^{n},S^{n})-H(Y^{n}|X_1^{n},X_2^{n},S^{n})+n \varepsilon_n  \nonumber\\
        &=& \sum_{i=1}^{n} H(Y_{i}|Y^{i-1},X_{2}^{n},S^{n})- H(Y_{i}|Y^{i-1},X_{1}^{n},X_{2}^{n},S^{n})
        +n \varepsilon_n \nonumber\\
        &\stackrel{(g)}\leq& \sum_{i=1}^{n} H(Y_{i}|X_{2,i},S_{i},S_{i-d})- H(Y_{i}|Y^{i-1},X_{1}^{n},X_{2}^{n},S^{n})
        +n \varepsilon_n \nonumber\\
        &\stackrel{(h)}=&  \sum_{i=1}^{n} H(Y_{i}|X_{2,i},S_{i},S_{i-d})- H(Y_{i}|X_{1,i},X_{2,i},S_{i},S_{i-d})
        +n \varepsilon_n \nonumber\\
        &=& \sum_{i=1}^{n}I(Y_{i};X_{1,i}|X_{2,i},S_{i},S_{i-d})+n \varepsilon_n , \nonumber
  \end{eqnarray}
  where\\
  (a) follows from Fano's inequality.\\
  (b) follows from chain rule.\\
  (c) follows from the fact that  $M_1$ and $S^{n}$ are independent.\\
  (d) follows from the fact that $X_{1}^{n}$ is a deterministic function of $(M_1,S^n)$ and the Markov chain $(M_1,S^n) - (X_{1}^n,S^n) - Y^n$.\\
  (e) follows from the fact that $X_{1}^{n}$ and $M_2$ are independent, and the fact that $X_{2}^{n}$ is a deterministic  function of $(M_2,S^n)$. Therefore,  $X_{1}^{n}$ and $X_{2}^{n}$ are independent given $S^{n}$.\\
  (f) and (g) follow from the fact that conditioning reduces entropy.\\
  (h) follows from the fact  that the channel output at time $i$ depends only on the state $S_i$ and the the inputs $X_{1,i}$ and $X_{2,i}$.\\
  Hence, we have
  \begin{eqnarray}
  R_{1}\leq \frac{1}{n}\sum_{i=1}^{n}I(Y_{i};X_{1,i}|X_{2,i},S_{i},S_{i-d})+ \varepsilon_n.  \label{r1t3}.
  \end{eqnarray}
Similarly, we have
  \begin{eqnarray}
  R_{2}\leq \frac{1}{n}\sum_{i=1}^{n}I(Y_{i};X_{2,i}|X_{1,i},S_{i},S_{i-d})+ \varepsilon_n.  \label{r2t3},
\end{eqnarray}
and the sum rate,
\begin{eqnarray}
   R_{1}+R_{2}\leq \frac{1}{n}\sum_{i=1}^{n}I(Y_{i};X_{1,i},X_{2,i}|S_{i},S_{i-d})+ \varepsilon_n. \label{r1+r2t3}
  \end{eqnarray}
  The expressions in (\ref{r1t3}), (\ref{r2t3}), and (\ref{r1+r2t3}) are the average of the mutual informations calculated at
  the empirical distribution in column $i$ of the codebook. We can rewrite these equations with the new variable Q,
  where $Q=i\in \{1,2,...,n\}$ with probability $\frac{1}{n}$. The equations become
  \begin{eqnarray}
  R_{1}&\leq& \frac{1}{n}\sum_{i=1}^{n}I(Y_{i};X_{1,i}|X_{2,i},S_{i},S_{i-d})+ \varepsilon_n \nonumber\\
        &=&  \frac{1}{n}\sum_{i=1}^{n}I(Y_{Q};X_{1,Q}|X_{2,Q},S_{Q},S_{Q-d},Q=i)+ \varepsilon_n \nonumber\\
        &=& I(Y_{Q};X_{1,Q}|X_{2,Q},S_{Q},S_{Q-d},Q)+ \varepsilon_n.
  \end{eqnarray}
  Now let us denote $ X_{1}\triangleq X_{1,Q}, X_{2} \triangleq X_{2,Q}, Y \triangleq Y_{Q}, S\triangleq S_{Q}$, and $\tilde{S}\triangleq S_{Q-d}$.
  \\we have
   \begin{eqnarray}
   R_{1}&\leq& I(X_{1};Y|X_{2},S,\tilde{S},Q)+ \varepsilon_n , \nonumber\\
   R_{2} &\leq& I(X_{2};Y|X_{1},S,\tilde{S},Q) + \varepsilon_n, \nonumber\\
   R_{1}+R_{2}&\leq& I(X_{1},X_{2};Y|S,\tilde{S},Q) + \varepsilon_n. \nonumber
   \end{eqnarray}
Now we need to show the following Markov relations hold:
  \begin{enumerate}
  \item $P(q|s,\tilde{s})=P(q)$ .
  \item $P(x_1|s,\tilde{s},q)=P(x_1|q)$.
  \item $P(x_2|x_1,s,\tilde{s},q)=P(x_2|\tilde{s},q)$.
  \item $P(y|x_1,x_2,s,\tilde{s},q)=P(y|x_1,x_2,s)$.
  \end{enumerate}
  We prove the above using the following claims:
  \begin{enumerate}
  \item  follows from the fact that $Q$ and the state process $S^n$ are  independent.
  \item  follows from the fact that $X_{1,i}=f_{1,i}(M_1)$ and that $M_1$ and $S^{n}$ are independent.

  \item  follows from the fact that $M_2$ and $(M_1,S^{n})$ are independent, and the fact that state process is a Markov chain, Hence
  \begin{eqnarray}
   P(m_{2},s^{i-d}|s_{i},s_{i-d},m_{1})&=&P(m_2,s^{i-d}|s_{i-d}).\nonumber
  \end{eqnarray}
  Therefore, we have the Markov chain $(M_{2},S^{i-d})-S_{i-d}-(M_1,S_i)$. Since $X_{1,i}=f_{1,i}(M_1)$ and $X_{2,i}=f_{2,i}(M_2,S^{i-d})$, where $f_{1,i},f_{2,i}$ are deterministic functions, we get the following Markov chain,
   \begin{eqnarray}
  X_{2,i}-(M_{2},S^{i-d})-S_{i-d}-(M_1,S_i)-X_{1,i}.
  \end{eqnarray}
  Therefore,
  \begin{eqnarray}
  P(x_{2,i}|x_{1,i},s_{i},s_{i-d})&=&P(x_{2,i}|s_{i-d}). \nonumber
  \end{eqnarray}
  Since this is true for all $i$,
  \begin{eqnarray}
  P(x_{2,q}|x_{1,q},s_{q},s_{q-d},q)&=&P(x_{2,q}|s_{q-d},q).\nonumber
  \end{eqnarray}
  We have $P(x_2|x_1,s,\tilde{s},q)=P(x_2|\tilde{s},q)$.
  \item  follows from the fact that the channel output at time $i$ depends only on the state $S_i$ and the the inputs $X_{1,i}$ and $X_{2,i}$.\\
  \end{enumerate}
  Hence, taking the limit as $n \rightarrow\infty$, $P_{e}^{(n)}\rightarrow 0$, we have the following converse:
  \begin{eqnarray}
   R_{1}&\leq& I(X_{1};Y|X_{2},S,\tilde{S},Q), \nonumber\\
   R_{2} &\leq& I(X_{2};Y|X_{1},S,\tilde{S},Q), \nonumber\\
   R_{1}+R_{2}&\leq& I(X_{1},X_{2};Y|S,\tilde{S},Q), \nonumber
  \end{eqnarray}
  for some choice of joint distribution $P(s,\tilde{s})P(q)P(x_{1}|q)P(x_{2}|\tilde{s},q)P(y|x_{1},x_{2},s)$
  and for some choice of random variable $Q$ defined on $|{\cal Q}| \leq 3$. This completes the proof of the converse.
\subsection{Achievability Theorem  \ref{Capacity region-  MAC with delayed CSI feedback only to one encoder}}
To prove the achievability of the capacity region, we need to show that for a fixed $P(x_{1})P(x_{2}|\tilde{s})$ and $(R_1,R_2)$ that satisfy,
   \begin{eqnarray}
   R_{1}&\leq& I(X_{1};Y|X_{2},S,\tilde{S}), \nonumber\\
   R_{2} &\leq& I(X_{2};Y|X_{1},S,\tilde{S}), \nonumber\\
   R_{1}+R_{2}&\leq& I(X_{1},X_{2};Y|S,\tilde{S}), \nonumber
  \end{eqnarray}
  there exists a sequence of $(n,2^{nR_{1}},2^{nR_{2}},d)$ codes where $P_{e}^{(n)}\rightarrow0$ as $n\rightarrow\infty$.
  Without loss of generality we assume that the finite-state space $\mathcal{S}=\left\{1,2,...,k\right\}$, and that the steady state probability   $\pi(l)>0$ for all $l\in \mathcal{S}$.

\emph{Encoder 1}: construct $2^{nR_1}$ independent codewords $X_{1}^n(i)$ where $i\in \left\{1,2,..,2^{nR_1} \right\}$ of length $n$, generate each symbol i.i.d., $X_1^n(i)\sim \prod_{l=1}^{n} P(X_{1,l})$.

 \emph{Encoder 2}: construct $k$  codebooks $\mathcal{C}^{\tilde{s}}_{2}$  (where the subscript is for Encoder $2$) for all $\tilde{S}\in \mathcal{S}$, when in each codebook $\mathcal{C}^{\tilde{s}}_{2}$ there are $2^{n_{2}(\tilde{s})R_{2}(\tilde{s})}$ codewords, where $n_{2}(\tilde{s})=(P(\tilde{S}=\tilde{s})-\epsilon')n$, for $\epsilon'>0$. Every codeword $\mathcal{C}^{\tilde{s}}_{2}(i)$ where $i\in \{1,2,..., 2^{n_{2}(\tilde{s})R_{2}(\tilde{s})}\}$ has a length of $n_{2}(\tilde{s})$ symbols.
 Each codeword  from the $\mathcal{C}^{\tilde{s}}_{2}$  codebook is built $X^{\tilde{s}}_{2}\thicksim$ i.i.d. $ P(x^{\tilde{s}}_{2}|\widetilde{S}=\tilde{s})$ (where the subscript is for Encoder $2$).
 A message $M_2$ is chosen according to a uniform distribution $\Pr (M_2=m_2)=2^{-nR_2}$, $m_2\in\left\{1,2,...,2^{nR_2}\right\}$.
 Every message $m_2$ is mapped into $k$ sub messages  $\mathcal{V}_{2}(m_2)=\left\{V^{1}_{2}(m_2),V^{2}_{2}(m_2),...,V^{k}_{2}(m_2)\right\}$ (one message from each codebook). Hence, every message  $m_{2}$ is specified by a $k$ dimensional vector.
 For a fix block length $n$, let $N_{\tilde{s}}$ be the number of times during the $n$ symbols
for which the feedback information at  encoder $2$ regarding the
channel state is $\tilde{S}=\tilde{s}$.
 Every time that the delayed CSI is $\tilde{S}=\tilde{s}$, encoder $2$ sends the next symbol from $\mathcal{C}^{\tilde{s}}_{2}$ codebook.
 Since $N_{\tilde{s}}$ is not necessarily
  equivalent to $n_2(\tilde{s})$, an error is declared if $N_{\tilde{s}}<n_2(\tilde{s})$, and the code is zero-filled
if $N_{\tilde{s}}>n_2(\tilde{s})$.
 Therefore we can send total of
 $2^{nR_{2}}=2^{\sum_{\tilde{s}\in \mathcal{S} }n_{2}(\tilde{s})R_2(\tilde{s})}$ messages.

\emph{Decoding}: we use successive decoding, similar to the decoding in section \ref{ACHIEVABILITY }.
It can be shown that  the probability of error, conditioned on a particular codeword being sent, goes to zero if the conditions of the following  are met:
\begin{eqnarray}
   R_{1}&\leq& I(X_{1};Y|X_{2},S,\tilde{S}), \nonumber\\
   R_{2} &\leq& I(X_{2};Y|X_{1},S,\tilde{S}), \nonumber\\
   R_{1}+R_{2}&\leq& I(X_{1},X_{2};Y|S,\tilde{S}).\nonumber
  \end{eqnarray}
  The above bound shows that the average probability of error, which by symmetry is equal to the probability for an individual pair of codewords $(m_1,m_2)$, averaged over all choices of codebooks in the random code construction, is arbitrarily small. Hence there exists at least one code $(n,2^{nR_{1}},2^{nR_{2}},d)$ with an arbitrarily small probability of error. To complete the proof we use time-sharing to allow any $(R_1,R_2)$ in the convex hull to be achieved.
\section{DETERMINATION OF THE TWO-STATE MAC CAPACITY REGION}\label{DETERMINATION OF THE TWO-STATE MAC CAPACITY REGION}
\vspace{-3 mm}
For simplicity we give here  the solution to the constrained optimization only for the symmetrical case, i.e., both CSI delays are the same ($d_1=d_2$), the solution of the other cases are obtained in a similar way. The optimization problem is:
\begin{eqnarray}
R_1+R_2=\max_{\mathcal{P}_1(\tilde{s}),\mathcal{P}_2(\tilde{s})} \frac{1}{2}\sum_{\tilde{s}}\pi(\tilde{s})\sum_{s}K^{d}(s,\tilde{s})\log\left(1+\frac{\mathcal{P}_1(\tilde{s})+\mathcal{P}_2(\tilde{s})}{\sigma_{s}^2}\right),
\end{eqnarray}
subject to the power constraints,
\begin{eqnarray}
  &&\sum_{\tilde{s}}\pi(\tilde{s})\mathcal{P}_1(\tilde{s}) \leq \mathcal{P}_1 ,\\
  &&\sum_{\tilde{s}}\pi(\tilde{s})\mathcal{P}_2(\tilde{s}) \leq \mathcal{P}_2 ,\\
  &&\mathcal{P}_1(\tilde{s})\geq 0   \ \  \forall \tilde{s},\\
  &&\mathcal{P}_2(\tilde{s})\geq 0   \ \  \forall \tilde{s}.
  \end{eqnarray}
The solution can be obtained by the Lagrange multiplier method. Since the objective function is monotonically increasing with  respect $\mathcal{P}_1$, and $\mathcal{P}_2$, it follows that the maximum is achieved when
 \begin{eqnarray}
  &&\sum_{\tilde{s}}\pi(\tilde{s})\mathcal{P}_1(\tilde{s}) = \mathcal{P}_1 ,\label{constraint1} \\
  &&\sum_{\tilde{s}}\pi(\tilde{s})\mathcal{P}_2(\tilde{s}) = \mathcal{P}_2 .\label{constraint2}
  \end{eqnarray}
Since $\log$ is a concave function, and  $\pi(\tilde{s}),K^{d}(s,\tilde{s})\geq 0$.  We get that objective function is concave in both variables $\mathcal{P}_1(\tilde{s})$, and $\mathcal{P}_2(\tilde{s})$. Also the constraints functions (\ref{constraint1}), and (\ref{constraint2}) are affine. So we can use the Kuhn-Tucker conditions \cite[Chapter 5.3.3] {Boyd} as a sufficient conditions to solve the optimization problem. Application of the Kuhn-Tucker conditions gives the following conditions of optimality:                                                                                                                                                                                                                                                                                                                                                                                                                                                                                                                                                                                                                                                                                                                                                                                                                                                                                                         \begin{eqnarray}
&&\frac{1}{2}\sum_{s}\frac{K^{d}(s,\tilde{s}_{i})}{\sigma_{s}^{2}+\mathcal{P}_1^{*}(\tilde{s}_{i})+\mathcal{P}_2^{*}(\tilde{s}_{i})}\leq \nu_1  \ \ \ \ , \ \ \ \forall \tilde{s}_{i}\in \{s_1,s_2,..,s_k\}, \label{kkt1} \\
&&\frac{1}{2}\sum_{s}\frac{K^{d}(s,\tilde{s}_{i})}{\sigma_{s}^{2}+\mathcal{P}_1^{*}(\tilde{s}_{i})+\mathcal{P}_2^{*}(\tilde{s}_{i})}\leq \nu_2  \ \ \ \ , \ \ \ \forall \tilde{s}_{i}\in \{s_1,s_2,..,s_k\},\label{kkt2} \\
&&\sum_{\tilde{s}}\pi(\tilde{s})\mathcal{P}_1^{*}(\tilde{s}) = \mathcal{P}_1 , \\
&&\sum_{\tilde{s}}\pi(\tilde{s})\mathcal{P}_2^{*}(\tilde{s}) = \mathcal{P}_2 ,
\end{eqnarray}
with equality in (\ref{kkt1}) whenever  $\mathcal{P}_1^{*}(\tilde{s}_i)\geq 0$, and equality in (\ref{kkt2}) whenever  $\mathcal{P}_2^{*}(\tilde{s}_i)\geq 0$. For the two state Gaussian MAC example in Section \ref{GAUSSIAN } we have,
\begin{eqnarray}
K^d=
\begin{bmatrix}
  1-\frac{g}{g+b}\left(1-(1-g-b)^d\right) & \frac{g}{g+b}(1-(1-g-b)^d) \\
  \frac{b}{b+g}(1-(1-g-b)^d) & 1-\frac{b}{b+g}(1-(1-g-b)^d) \\
\end{bmatrix}
\end{eqnarray}
Now the solution to the constrained optimization problem is obtained by finding $\mathcal{P}_1^{*}(\tilde{s}_i)$, and $\mathcal{P}_2^{*}(\tilde{s}_i)$ that satisfy the Kuhn-Tucker  conditions. For simplicity, in order to solve the optimization problem we used \texttt{CVX}, a package for specifying and solving convex optimization problems \cite{cvx}.

\end{document}